\documentclass[12pt]{article}
\usepackage{amsmath,amssymb,amsthm,adjustbox,graphicx,setspace,xcolor,multirow,array,bbm,arydshln,enumitem,subcaption,threeparttable,bm,booktabs}
\usepackage[hidelinks]{hyperref}
\usepackage{float}
\usepackage[longnamesfirst]{natbib}
\newtheorem{lemma}{Lemma}
\newtheorem{corollary}{Corollary}
\newtheorem{condition}{Condition}

\newtheorem{theorem}{Theorem}
\newtheorem{proposition}{Proposition}
\newtheorem{assumption}{Assumption}

\newcommand{\minitab}[2][l]{\begin{tabular}{#1}#2\end{tabular}}
\usepackage{accents}

\usepackage[margin=1in]{geometry}
\allowdisplaybreaks

\begin{document}
\title{Policy Relevant Treatment Effects with Multidimensional Unobserved Heterogeneity\thanks{We thank Bulat Gafarov, and seminar participants at the University of Amsterdam, University of Bologna, University of Groningen, Tilburg University, COMPIE conference, and EC$^2$ for helpful comments and conversations.}}
\author{Takuya Ura\thanks{T. Ura: Department of Economics, University of California, Davis, One Shields Avenue, Davis, CA 95616. Email: takura@ucdavis.edu}\\ Department of Economics\\ University of California, Davis \and Lina Zhang\thanks{L. Zhang: Amsterdam School of Economics, University of Amsterdam. Roetersstraat 11, 1018 WB Amsterdam, The Netherlands. Email: l.zhang5@uva.nl}\\  University of Amsterdam, and\\Tinbergen Institute}
\date{}
\maketitle

\begin{abstract}
This paper provides a unified framework for bounding policy relevant treatment effects using instrumental variables. In this framework, the treatment selection may depend on multidimensional unobserved heterogeneity. We derive bilinear constraints on the target parameter by extracting information from identifiable estimands. We apply a convex relaxation method to these bilinear constraints and provide conservative yet computationally simple bounds. Our convex-relaxation bounds extend and robustify the bounds by Mogstad, Santos, and Torgovitsky (2018) which require the threshold-crossing structure for the treatment: if this condition holds, our bounds are simplified to theirs for a large class of target parameters; even if it does not, our bounds include the true parameter value whereas theirs may not and are sometimes empty. Linear shape restrictions can be easily incorporated to narrow the proposed bounds. Numerical and simulation results illustrate the informativeness of our convex-relaxation bounds.






\medskip

\begin{description}
\item[Keywords:]  Multidimensional unobserved heterogeneity, 
instrumental variables, shape restriction, convex relaxation, partial identification 
\end{description}
\end{abstract}

\newpage

\section{Introduction}
Policy relevant treatment effects (PRTEs) have emerged as pivotal target parameters in both empirical and theoretical studies. Instrumental variables (IVs) are commonly used to identify PRTEs with the endogenous treatment choice. In the threshold-crossing model for the treatment with one-dimensional unobserved heterogeneity, \citet{heckman/vytlacil:1999,heckman/vytlacil:2001,heckman/vytlacil:2005} study the identification of various PRTEs as the weighted average of the marginal treatment effects. In the same threshold-crossing model, \citet*{mogstad/santos/torgovitsky:2017} propose a linear programming method to construct a bound on  them.

The threshold-crossing assumption for the treatment can be restrictive in many empirical scenarios. As highlighted by \citet[Section 6]{heckman/vytlacil:2005}, it requires ``uniformity'' of treatment selection across all individuals and restricts the heterogeneity of the treatment selection. 
\citet{mogstad2020policy} also demonstrates that this condition becomes particularly strong when multiple IVs provide distinct incentives that lead to heterogeneous treatment selection. In fact, a number of recent papers have proposed specification tests for the threshold-crossing structure together with the instrument validity.\footnote{See \cite{huber/mellace:2015,kitagawa_2015,laffers2017note,mourifie2017testing,kedagni2020generalized,SUN2023105523,carr2023testing} among others.} Although unobserved heterogeneity plays a crucial role in individual choice models and causal inference studies \citep{heckman2001micro,imbens2014instrumental,mogstad2024instrumental}, to our knowledge, it is still an open question whether and how the framework of \citet{mogstad/santos/torgovitsky:2017} can be extended to bound PRTEs with multidimensional unobserved heterogeneity affecting both treatment and outcome.

This paper aims to address this question when we have a binary endogenous treatment with IVs. We offer a systematic approach to bound various PRTEs without imposing any threshold-crossing structure on the treatment selection process. Our method explicitly models the unobserved heterogeneity by a multidimensional random vector $V$, and we assume that treatment and potential outcomes are mean-independent once conditioned on $V$. This conditional mean-independence assumption then implies that both the target PRTEs and identifiable estimands are weighted averages of bilinear functions of the unknown propensity score and marginal treatment response functions.\footnote{\cite{yap2022sensitivity} also derives a similar expression for PRTEs when investigating violations of the threshold-crossing structure.} 

The unknown propensity score and marginal treatment response functions appear \textit{bilinearly} in both the target parameter and the identifiable estimands. Consequently, the PRTE bounds can only be formulated as solutions to nonconvex optimization problems with nonlinear (in particular, bilinear) constraints. 
We propose to use the convex relaxation method of \cite{mccormick1976computability}, where we replace each bilinear function in the optimization problems with a single function and characterize its feasible region using McCormick Envelopes. The resulting bounds are obtained through linear problems.

The convex relaxation method has several attractive features. First, it provides conservative yet computationally attractive PRTE bounds, as the bounds can be easily computed using well-established linear programming techniques.
Second, while the convex relaxation produces an outer set of the PRTE sharp identifed set, we demonstrate that this outer set remains informative. Specifically, we provide sufficient conditions under which the outer set leads to no identification power loss. 
Third, it offers a unified framework for incorporating linear shape restrictions or combinations of these shape restrictions—some of which may not have been explored in the literature—to further tighten the bounds. This allows researchers to avoid deriving PRTE bounds under different shape restrictions on a case-by-case basis. Lastly, we can conduct inference by directly applying the regularized support function method introduced by \citet{gafarov2019inference}. This method yields closed-form asymptotic Gaussian estimators for PRTE bounds, so that uniformly valid confidence sets can be constructed using the standard normal critical values, which substantially reduces the computational burden. 

It is worthwhile dicussing further the second feature in the above paragraph. For a large class of target parameters, our bounds are simplified to the identified set of \citet{mogstad/santos/torgovitsky:2017} when the treatment selection happens to satisfy the threshold-crossing assumption with one-dimensional unobserved heterogeneity. This result is formally stated in Theorem \ref{coro_equav_MST_CVR} in Section \ref{section:no_id_power_loss} and has a practical takeaway. In practice, the researcher may not know whether the treatment equation satisfies the threshold-crossing assumption or not, and our bounds do not require this knowledge. 
Even if the treatment equation does not satisfy this assumption, our bounds are still valid, whereas the ones of \citet{mogstad/santos/torgovitsky:2017} are not (e.g., their bounds can be empty in some cases). In this sense, we extend \citet{mogstad/santos/torgovitsky:2017} to a more general setup with multidimensional unobserved heterogeneity and provide a more robust alternative to practitioners.


Numerical analysis and simulation exercises are used to verify the identification results for various PRTEs and asymptotic properties of the inference method. Importantly, numerical results demonstrate that the PRTE bounds of \citet{mogstad/santos/torgovitsky:2017} can be empty if their key assumption (i.e., one-dimensional unobserved heterogeneity in the treatment selection process) is violated. Moreover, the conservativeness introduced by the convex relaxation method does not lead to information loss in some cases. For example, without shape restriction, our PRTE outer set for the average treatment effect (ATE) is the same as its sharp bounds of \citet{manski1990nonparametric}.

The rest of this paper is organized as follows. Section \ref{section:literature_reviews} discusses related literature. Section \ref{section:model and target} introduces our model setup, target parameters, and some commonly used shape restrictions. Section \ref{section:identification_analysis} presents partial identification results using convex relaxation method and investigates conditions under which our convex-relaxation bounds have no identification power loss. Section \ref{section:inference} outlines the inference procedure in \cite{gafarov2019inference} for constructing uniformly valid PRTE confidence sets. Section \ref{section:numerical_simulation} provides numerical and simulation results. Section \ref{section:conclusion} concludes.

\subsection{Related Literature}\label{section:literature_reviews}


A growing body of work has focused on relaxing the threshold-crossing structure with one-dimensional unobserved heterogeneity to accommodate more flexible forms of unobserved heterogeneity. 
This paper differs from all previous studies because our framework simultaneously incorporates the following features. First, we begin with a general model setup without imposing any threshold-crossing restrictions on the treatment selection, and our target parameters include a large class of PRTEs. Second, we explicitly model the unobserved heterogeneity using a multi-dimensional random vector. Our identification analysis relies on the conditional mean-independence assumption between the potential outcomes and the treatment variable, given the unobserved heterogeneity, which results in a bilinear structure for our target PRTE parameters. 
Third and most importantly, to overcome the computational challenge posed by the bilinear structure, we propose computationally attractive and informative bounds obtained by using the convex relaxation method. 
As a foundational result in the causal effect literature, \citet{vytlacil:2002} establishes the equivalence between the classical monotonicity condition of \citet{imbens/angrist:1994} and the threshold-crossing model with single-dimensional unobserved heterogeneity. In this section of related literature, we cover the relevant papers which relax one of these two conditions.    

This paper is closely related to \citet{yap2022sensitivity}, which also derives a similar bilinear expression for PRTEs when investigating violations of the threshold-crossing structure for the treatment with one-dimensional unobserved heterogeneity. The author uses unobserved compliance types to capture unobserved heterogeneity and a sensitivity parameter to measure the violations of monotonicity. The PRTE bounds are then derived across different values of the sensitivity parameter, and computed by breaking a nonconvex optimization problem into an inner and an outer linear program. 
This paper differs from \citet{yap2022sensitivity} for several reasons. First, our general model setup does not rely on a sensitivity parameter about the violations of monotonicity. Second, by explicitly modeling the unobserved heterogeneity, we can easily impose shape restrictions on the unknown but interpretable functions of the propensity score and MTRs, in the same spirit as 
\citet{heckman/vytlacil:1999,heckman/vytlacil:2005} and \citet{brinch/mogstad/wiswall:2017}. Third, we address the nonconvex problem using convex relaxation, which leads to computationally attractive bounds and easy-to-implement inference.


This paper also relates to studies that explore various relaxations of the threshold-crossing structure for the treatment with one-dimensional unobserved heterogeneity. \citet{klein2010heterogeneous} delves into a selection model with local violations of the treshold-crossing structure. \cite{dinardo/lee:2011} and \citet{small2017instrumental} impose the stochastic monotonicity on the propensity score across all strata of unobserved heterogeneity.  In cases with multiple IVs, \cite*{mogstad2020policy} propose a partial monotonicity condition for each IV separately. In a related vein, \citet{goff2020vector} introduces a vector monotonicity condition that is stronger than partial monotonicity. Other studies that explore various 
weaker monotonicity conditions include \citet{de2017tolerating}, \citet{sloczynski2021should}, \citet{noack2021sensitivity}, \citet{dahl2023never}, \citet{van2023limited}, and \citet{hoff2023identifying}. 

Another body of work develops less restrictive models capable of accommodating multidimensional unobserved heterogeneity. For instance, \citet{lee2018identifying} examine the treatment assignment determined by a vector of threshold-crossing rules. 
\citet{gautier2011triangular} and \citet{gautier2021relaxing} model the unobserved heterogeneity using random coefficient models. \cite{han2024set} propose a generalization of control function approach 
and use a non-monotone treatment assignment as a motivating example. \citet{navjeevan2023identification} establish the point identification of target parameters, where they represent the unobserved heterogeneity using potential outcomes and an unobserved type. Notably, such a point identification result requires the target parameter to be an expectation of a function of unobservables that can be further expressed as a linear function of observables, which may not hold for some commonly studied PRTEs.\footnote{One counterexample that does not satisfy this condition is the ATE in a heterogeneous treatment effect model with an endogenous binary treatment. We know that ATE is an expectation of the difference between two potential outcomes. It is generally impossible to express such a difference as a linear map of functions of observables, and the ATE point identification is also not achievable without additional assumptions.}

There are also studies that consider bounding PRTEs while leaving the treatment selection completely unspecified. 
Some studies focus on treatment effects that have analytical bounds \citep[e.g.,][]{manski1990nonparametric,manski1997monotone,manski2000monotone,huber2017sharp}. Others consider treatment effects whose analytical bounds do not exist or are difficult to derive \citep[e.g.,][]{laffers2019bounding,russell2021sharp}.

\section{Model Setup and Targeting Parameters}\label{section:model and target}
Consider a canonical program evaluation problem with a binary treatment variable $D\in\{0,1\}$ and a scalar outcome variable $Y\in[y^L,y^U]$.\footnote{When there is no restriction on the support of $Y$, we can set $y^L=-\infty$ and $y^U=\infty$.} Define
$$
Y=DY_1+(1-D)Y_0,
$$
where $Y_1$ and $Y_0$ are the two potential outcomes. We assume there exists a vector of observable variables, $Z=(Z_0,X)\in\mathcal{Z}$, that includes a vector of instruments $Z_0\in\mathcal{Z}_0$ and a vector of covariates $X\in\mathcal{X}$. The instruments can be discrete, continuous, or mixed. We introduce a random vector ${V}\in{\mathcal{V}}$ to represent the multidimensional unobserved heterogeneity that affects both the treatment $D$ and the outcome $Y$, so that the individual treatment effect $Y_1-Y_0$ can be heterogeneous depending on the values of the unobserved heterogeneity $V$. For a generic probability distribution $P$, we let $E_{P}$ be its expectation operator. We use $\mathbb{P}$ to denote the true distribution for $(Y_0,Y_1,D,Z,V)$ and use $\mathbb{E}=E_{\mathbb{P}}$ as its expectation operator. We maintain the following key assumptions throughout the paper. 

\begin{assumption}\label{assumption_DGP} ~ There is a random vector $V$ with the following two properties: 
\begin{itemize}
\item[(i)] $\mathbb{E}[Y_d\mid D,Z,{V}]={\mathbb{E}}[Y_d\mid X,{V}]$ and $\mathbb{E}[{\mathbb{E}}[Y_d\mid X,{V}]^2]<\infty$  for each $d=0,1$.
\item[(ii)] Conditional on $Z$, $V$ is uniformly distributed on $\mathcal{V}=[0,1]^{K}$ for a known constant $K$.
\end{itemize}
\end{assumption}
Assumption \ref{assumption_DGP} (i) implies that the treatment endogeneity arises from the unobserved heterogeneity $V$. Consequently, given $(X,{V})$, the treatment and instruments become mean-independent of the potential outcomes. Condition (ii) assumes that the unobserved heterogeneity $V$ is independent of $Z$, and it also introduces a standard normalization to the distribution of $V$.\footnote{
If $V$ is independent of $Z_0$ given $X$ and is continuously distributed given $X$, we can always normalize its distribution to a uniform distribution.} 
This assumption is weaker than the one imposed in \cite{mogstad/santos/torgovitsky:2017}.  We allow $V$ to be multidimensional (including $K=1$ as a special case), whereas \cite{mogstad/santos/torgovitsky:2017} assume $V$  to be one-dimensional.\footnote{Note that the researcher may not know the smallest value of $K$ for which Assumption \ref{assumption_DGP} holds, but may know its upper bound. We can use this upper bound as $K$ in Assumption \ref{assumption_DGP}. For example, if Assumption \ref{assumption_DGP} holds with $K=1$, then it holds with $K=2$.} 

Let $F$ denote the cumulative distribution function. The following lemma shows conditions under which Assumption \ref{assumption_DGP} holds with $K=2$. 

\begin{lemma}\label{lemma_Y0Y1K2}
Suppose $F_{Y_0\mid X}$ and $F_{Y_1\mid Y_0,X}$ are invertible and $Y_d$ is $L^2$-bounded for $d=0,1$.
Define $V=(F_{Y_0\mid X}(Y_0),F_{Y_1\mid Y_0,X}(Y_1))$. 
If $Z$ is independent of $(Y_0,Y_1)$ given $X$, then this $V$ satisfies Assumption \ref{assumption_DGP} with $K=2$.
\end{lemma}

Lemma \ref{lemma_Y0Y1K2} shows that if the potential outcomes are independent of $Z$ given $X$, we can always use the two normalized potential outcomes, $(F_{Y_0\mid X}(Y_0),F_{Y_1\mid Y_0,X}(Y_1))$, to represent the multidimensional $V$. Nonetheless, this particular choice of $V$ only provides one example for the unobserved heterogeneity with $K=2$, and there can be other choices of $V$ that also satisfy Assumption \ref{assumption_DGP}. Note that our model setup is general enough to accommodate various treatment selection models in the literature that assume multidimensional unobserved heterogeneity, including those with random coefficients \citep{gautier2011triangular}, local violations to monotonicity \citep{klein2010heterogeneous}, 
stochastic compliance \citep{small2017instrumental},
and the double hurdle model where treatment is assigned if and only if two criteria are satisfied simultaneously \citep{lee2018identifying}. This is because all these models mentioned above assume the same independence condition between potential outcomes and instruments (conditional on covariates, if any) as in our Lemma \ref{lemma_Y0Y1K2}.

Define the marginal treatment response function $m_{\mathbb{P},d}(v,x)$ with $d=0,1$ and the propensity score $m_{\mathbb{P},D}(v,z)$, conditioning on the unobserved heterogeneity $V=v$, as below:
$$
m_{\mathbb{P},d}(v,x)={\mathbb{E}}[Y_d\mid V=v,X=x],\text{ and }~
m_{\mathbb{P},D}(v,z)={\mathbb{E}}[D\mid V=v,Z=z].
$$
These three functions, $(m_{\mathbb{P},0},m_{\mathbb{P},1},m_{\mathbb{P},D})$, are the key ingredients for various policy relevant treatment effect (PRTE) parameters.  We are going to consider the target parameters of the form
\begin{align}\label{def_target_parameters}
&\mathbb{E}\left[\int_{{\mathcal{V}}} m_{\mathbb{P},0}(v,X)\Big[(1-m_{\mathbb{P},D}(v,Z))\omega_{00}^\star(v,Z)+m_{\mathbb{P},D}(v,Z)\omega_{01}^\star(v,Z)\Big]dv\right]
\nonumber
\\&+
\mathbb{E}\left[\int_{{\mathcal{V}}} m_{\mathbb{P},1}(v,X)\Big[(1-m_{\mathbb{P},D}(v,Z))\omega_{10}^\star(v,Z)+m_{\mathbb{P},D}(v,Z)\omega_{11}^\star(v,Z)\Big]dv\right],
\end{align}
where $\omega_{dd'}^\star$ with $d,d'\in\{0,1\}$ denotes a known or identifiable weight.\footnote{
We consider a general setup that accommodates a large class of PRTE parameters whose known or identifiable weights, $\omega_{dd'}^\star(v,z)$, are allowed to be functions of $v\in\mathcal{V}$, except for Proposition \ref{prop_reduce_dim} and Theorem \ref{coro_equav_MST_CVR}. We assume that the dimension of $V$ is known. For researchers interested only in PRTEs with weights that do not depend on $v$,  however, we can remove the assumption of a known dimension of $V$, as we will discuss in more detail in Proposition \ref{prop_reduce_dim} in later sections.}  Table \ref{table:target_parameters} below collects some important examples of such PRTE parameters and provides expressions for their weights. For example, if our target parameter is the average treatment effect, $\mathbb{E}[Y_1-Y_0]$, then $\omega_{00}^\star(v,z)=\omega_{01}^\star(v,z)=-1$ and $\omega_{10}^\star(v,z)=\omega_{11}^\star(v,z)=1$.

\begin{table}[h!]
\caption{Weights for Commonly Studied Target Parameters}\label{table_weights}
\label{table:target_parameters}
\begin{adjustbox}{width=\linewidth}
\begin{threeparttable}
\begin{tabular}{p{5cm}cccccc}
\hline
\multirow{2}{*}{Target Parameter}&\multirow{2}{*}{Expression}&\multicolumn{4}{c}{Weights}\\
\cline{3-6}\\
&&$\omega_{00}^\star(v,z)$&$\omega_{01}^\star(v,z)$&$\omega_{10}^\star(v,z)$&$\omega_{11}^\star(v,z)$\\
\hline
Average Untreated Outcome&$\mathbb{E}[Y_0]$&1&1&0&0\\[2pt]
Average Treated Outcome&$\mathbb{E}[Y_1]$&0&0&1&1\\[2pt]
Average Treatment Effect (ATE) &$\mathbb{E}[Y_1-Y_0]$&-1&-1&1&1\\[2pt]
ATE given $X\in\mathcal{X}^\star$&$\mathbb{E}[Y_1-Y_0\mid X\in\mathcal{X}^\star]$&$-\frac{1[x\in\mathcal{X}^\star]}{\mathbb{P}(X\in\mathcal{X}^\star)}$&$-\frac{1[x\in\mathcal{X}^\star]}{\mathbb{P}(X\in\mathcal{X}^\star)}$&$\frac{1[x\in\mathcal{X}^\star]}{\mathbb{P}(X\in\mathcal{X}^\star)}$&$\frac{1[x\in\mathcal{X}^\star]}{\mathbb{P}(X\in\mathcal{X}^\star)}$\\[2pt]
Average Treatment Effect on the Treated (ATT) &$\mathbb{E}[Y_1-Y_0\mid D=1]$&0&$-\frac{1}{\mathbb{P}(D=1)}$&0&$\frac{1}{\mathbb{P}(D=1)}$\\[2pt]
Average Treatment Effect on the Untreated (ATU) &$\mathbb{E}[Y_1-Y_0\mid D=0]$&$-\frac{1}{\mathbb{P}(D=0)}$&0&$\frac{1}{\mathbb{P}(D=0)}$&0\\[2pt]
Generalized Local Average Treatment Effect given $V\in\mathcal{V}^\star$& $\mathbb{E}[Y_1-Y_0\mid V\in \mathcal{V}^\star]$&$-\frac{1[v\in \mathcal{V}^\star]}{\mathbb{P}(V\in \mathcal{V}^\star)}$&$-\frac{1[v\in \mathcal{V}^\star]}{\mathbb{P}(V\in \mathcal{V}^\star)}$&$\frac{1[v\in \mathcal{V}^\star]}{\mathbb{P}(V\in \mathcal{V}^\star)}$&$\frac{1[v\in \mathcal{V}^\star]}{\mathbb{P}(V\in \mathcal{V}^\star)}$\\[2pt]
PRTE for a distributional change of $Z$&$\mathbb{E}[Y^*]$ with $Z^\star\sim F_{Z^\star}$ and $Z^\star\in\mathcal{Z}$&$\frac{F_{Z^\star}(dz)}{F_{Z}(dz)}$&0&0&$\frac{F_{Z^\star}(dz)}{F_{Z}(dz)}$\\[5pt]
Average Selection Bias&$\mathbb{E}[Y_0\mid D=1]-\mathbb{E}[Y_0\mid D=0]$&$-\frac{1}{\mathbb{P}(D=0)}$&$\frac{1}{\mathbb{P}(D=1)}$&0&0\\[3pt]
Average Selection on the Gain&$\mathbb{E}[Y_1-Y_0\mid D=1]$&$\frac{1}{\mathbb{P}(D=0)}$&$-\frac{1}{\mathbb{P}(D=1)}$&$-\frac{1}{\mathbb{P}(D=0)}$&$\frac{1}{\mathbb{P}(D=1)}$\\
&$~~-\mathbb{E}[Y_1-Y_0\mid D=0]$\\
\hline
\end{tabular}
Note: See Appendix \ref{appendix_weights} for detailed derivations for the expressions of the weights. The target parameter `PRTE for a distributional change of $Z$' is defined as $\mathbb{E}[Y^*]-\mathbb{E}[Y]$ with $Z^\star\sim F_{Z^\star}$ and $Z^\star\in\mathcal{Z}$, which is the unconditional version of the PRTE in (3.6) of \citet{heckman2007econometric}. Since $\mathbb{E}[Y]$ can be point identified from the observables, we only present the weights for $\mathbb{E}[Y^*]$. 
\end{threeparttable}
\end{adjustbox}
\end{table}

We can see that $(m_{\mathbb{P},0},m_{\mathbb{P},1},m_{\mathbb{P},D})$ appears in bilinear forms in the expression of the target PRTEs given in \eqref{def_target_parameters}. Therefore, for the rest of the paper, define 
\begin{equation}
m_{\mathbb{P},0D}(v,z)=m_{\mathbb{P},0}(v,x)(1-m_{\mathbb{P},D}(v,z))
\mbox{ and }
m_{\mathbb{P},1D}(v,z)=m_{\mathbb{P},1}(v,x)m_{\mathbb{P},D}(v,z).
\end{equation}
Then, we can replace the bilinear functions in \eqref{def_target_parameters} using $m_{\mathbb{P},0D}(v,z)$ and $m_{\mathbb{P},1D}(v,z)$. Denote 
$$m_\mathbb{P}=(m_{\mathbb{P},0},m_{\mathbb{P},1},m_{\mathbb{P},D},m_{\mathbb{P},0D},m_{\mathbb{P},1D})'\in\mathcal{M},$$
where $\mathcal{M}\subseteq (\mathbf{L}^2(\mathcal{V}\times\mathcal{X},[y^L,y^U]))^2\times\mathbf{L}^2(\mathcal{V}\times\mathcal{Z},[0,1])\times(\mathbf{L}^2(\mathcal{V}\times\mathcal{Z},[y^L,y^U]))^2$ is the parameter space for $m_\mathbb{P}$.\footnote{We define $\mathbf{L}^2(\mathcal{A},\mathcal{B})$ to be the set of all the $L^2$ functions from $\mathcal{A}$ to $\mathcal{B}$.}  We  define $\Gamma^\star$ to be the linear map from $\mathcal{M}$ to $\mathbb{R}$ with  
\begin{align}\label{def_target_parameters_linear}
\Gamma^\star(m)&=\mathbb{E}\left[\int_{{\mathcal{V}}} \Big[m_{0D}(v,Z)\omega_{00}^\star(v,Z)+(m_0(v,X)-m_{0D}(v,Z))\omega_{01}^\star(v,Z)\Big]dv\right]
\nonumber\\&\quad+
\mathbb{E}\left[\int_{{\mathcal{V}}} \Big[(m_1(v,X)-m_{1D}(v,Z))\omega_{10}^\star(v,Z)+m_{1D}(v,Z)\omega_{11}^\star(v,Z)\Big]dv\right]
\end{align}
for a generic function $m=(m_{0},m_{1},m_{D},m_{0D},m_{1D})'\in\mathcal{M}$.
Given $\Gamma^\star$, we can express the target parameter in \eqref{def_target_parameters} as $\Gamma^\star(m_{\mathbb{P}})$, which is a linear functional of $m_{\mathbb{P}}$. The functional $\Gamma^\star$ is identifiable from the distribution of $Z$ and thus is known to the researchers. In this sense, our partial identification analysis focuses on how we can bound the unknown function  $m_{\mathbb{P}}$. Note that under the threshold-crossing structure for the treatment with one-dimensional unobserved heterogeneity, \citet{heckman/vytlacil:2005,heckman2007econometric}  express the target parameters in Table \ref{table:target_parameters} as linear functionals of the unknown functions $m_{\mathbb{P},0}$ and $m_{\mathbb{P},1}$ only. This is because $m_{\mathbb{P},D}$ is identified by the propensity score $\mathbb{P}(D=1|Z)$, and it is also the main difference between our model setup and those in the existing literature studying PRTEs under the Heckman-Vytlacil  framework.

\subsection{Shape Restrictions}\label{section:shape_restriction}
Shape restrictions on the unknown function $m_{\mathbb{P}}$ have often been employed in the literature to tighten the PRTE bounds.
In this section, we discuss a few commonly used shape restrictions.

\subsubsection{Restrictions on $m_{\mathbb{P},D}$}

We begin by examining restrictions on the treatment variable and their implied restrictions on $m_{\mathbb{P},D}$.  


\begin{condition}[{Stochastic Monotonicity}]\label{Stochastic_Monotonicity}
Recall $Z=(Z_0,X)$ and suppose $Z_0$ is one-dimensional. 
For any $z_0,z_0'\in\mathcal{Z}_0$ and $(x,v)\in\mathcal{X}\times\mathcal{V}$, if $z_0\geq z_0'$, then $\mathbb{P}(D=1\mid V=v,Z_0=z_0,X=x)\geq \mathbb{P}(D=1\mid V=v,Z_0=z_0',X=x)$.\footnote{When $Z_0$ is multidimensional, we can extend this condition by interpreting $z_0\geq z_0'$, e.g., as each element of $z_0$ being greater than or equal to that of $z_0'$.}
\end{condition}

Condition \ref{Stochastic_Monotonicity}, introduced by \cite{dinardo/lee:2011}, is employed by \cite{small2017instrumental} to study the identification of the Wald estimand. It imposes a monotone restriction on the propensity score at all strata of the unobserved heterogeneity $V$, requiring
$m_{\mathbb{P},D}(v,z)=\mathbb{P}(D=1\mid V,Z_0=z_0,X)$ to be weakly increasing in $z_0$.

\begin{condition}[{Classical Monotonicity}]\label{condition:determinisMonoo}
Recall $Z=(Z_0,X)$ and suppose $Z_0$ is one-dimensional. 
Consider
$
D=D_Z,
$
where $Z_0$ and the potential treatments $\{D_z\}_{z\in\mathcal{Z}}$ are independent given $X$.
For any $z_0,z_0'\in\mathcal{Z}_0$ and $x\in\mathcal{X}$, if $z_0\geq z_0'$, then
$
D_{(z_0,x)}\geq D_{(z'_0,x)}\mbox{ almost surely}.
$
\end{condition}

Condition \ref{condition:determinisMonoo}, introduced by \citet{imbens/angrist:1994}, is a widely adopted monotonicity condition and is stronger than stochastic monotonicity.\footnote{In \citet{dinardo/lee:2011}, they refer to Condition \ref{condition:determinisMonoo} as ``deterministic monotonicity'' because it asserts that the treatment status of an individual with a given value of the unobserved heterogeneity $V$ is a deterministic function of $Z$.} It requires all individuals to respond to the same shift in the instrument value in the same direction. Specifically, in the case of a binary IV, the classical monotonicity rules out ``defiers'' and is the key to point identify the local average treatment effect. \citet{vytlacil:2002} shows that imposing the classical monotonicity is equivalent to assuming that the treatment variable is determined by a threshold-crossing model with one-dimensional unobserved heterogeneity. In other words, under Condition \ref{condition:determinisMonoo}, there exists a one-dimensional random variable $V$ such that
$D=1\{\mathbb{P}(D=1\mid Z)\geq V\}\mbox{ almost surely}$, where $V$ given $Z$ is uniformly distributed over $[0,1]$, 
and we can set $m_{\mathbb{P},D}(V,Z)=1\{\mathbb{P}(D=1\mid Z)\geq V\}$.

\begin{condition}[Random Coefficient Model with Identifiable Distribution]\label{condition:random_coeff}
    Suppose there are no covariates $X$ for simplicity, and $Z=(Z_1,...,Z_K)$ contains $K$ instruments. Consider a random coefficient model for the treatment variable 
$$D=1[\pi_1Z_1+...+\pi_KZ_K\geq 0],$$
where $\{\pi_k\}_{k=1}^K$ represent continuous random coefficients that are independent of $Z$. 
Define $V=(V_1,V_2,...,V_K)=[F_{\pi_1}(\pi_1),F_{\pi_2|\pi_1}(\pi_2),...,F_{\pi_K|\pi_1,...,\pi_{K-1}}(\pi_K)]$. \citet{gautier2013nonparametric} provides sufficient conditions under which the joint distribution function $F_{\pi}(\cdot)$ is point identified. Under those conditions, we can set $$m_{\mathbb{P},D}(V,Z)=1[\pi_1Z_1+...+\pi_KZ_K\geq 0],$$
where $\pi_1=F^{-1}_{\pi_1}(V_1)$ and $\pi_k=F^{-1}_{\pi_k|\pi_1,...,\pi_{k-1}}(V_k)$ for all $k=2,...,K$. The discussions above also hold when conditioning on $X$.
\end{condition}

\begin{condition}[Double hurdle model]\label{condition:double_hurdle} Suppose there are no
covariates $X$ for simplicity. Consider a double hurdle model where a treatment is assigned if and only if two criteria are satisfied simultaneously:
$$
D=1[U_1<Q_1(Z)\text{ and }U_2<Q_2(Z)],
$$
where $U_1,U_2\in\mathbb{R}$ stand for some continuous unobserved characteristics that are independent of $Z$, and functions $Q_1$ and $Q_2$ are a priori unknown. 
Define $V=(V_1,V_2)=(F_{U_1}(U_1),F_{U_2|U_1}(U_2))$. Theorem 4.2 of \citet{lee2018identifying} provides sufficient conditions under which the joint distribution function $F_{U_1,U_2}(\cdot)$, $Q_1$, and $Q_2$ are all point identified. Under those conditions, we can set
$$m_{\mathbb{P},D}(V,Z)=1[U_1<Q_1(Z)\text{ and }U_2<Q_2(Z)],$$
where $U_1=F^{-1}_{U_1}(V_1)$ and $U_2=F^{-1}_{U_2\mid U_1}(V_2)$. The discussions above also hold when conditioning on $X$.
\end{condition}

\subsubsection{Restrictions on $(m_{\mathbb{P},0},m_{\mathbb{P},1})$}


Next, we consider restrictions on the marginal treatment response functions $(m_{\mathbb{P},0},m_{\mathbb{P},1})$ implied by variants of the monotone treatment response condition. 
\begin{condition}[{Monotone Treatment Response}]\label{condition:mono_resp}
With probability one, treatment responses are weakly increasing $Y_1\geq Y_0$.
\end{condition}
Condition \ref{condition:mono_resp} is the monotone treatment response condition introduced by \citet{manski1997monotone} and \citet{manski2000monotone}. It implies 
$\mathbb{E}[Y_0|Z=z] \leq \mathbb{E}[Y|Z=z] \leq \mathbb{E}[Y_1|Z=z] \text{ for all }z=(z_0,x)\in\mathcal{Z}, $
which further implies restrictions on $(m_{\mathbb{P},0},m_{\mathbb{P},1})$ 
\begin{equation}\label{MTR_implied_restrictions}\begin{aligned}
   \mathbb{E}[m_{\mathbb{P},0}(V,x)] \leq \inf_{z_0\in\mathcal{Z}_0}&\mathbb{E}[Y|X=x,Z_0=z_0], \text{ and }\\
   &\sup_{z_0\in\mathcal{Z}_0}\mathbb{E}[Y|X=x,Z_0=z_0]\leq\mathbb{E}[m_{\mathbb{P},1}(V,x)].
\end{aligned}
\end{equation}
\begin{condition}[{Monotone Treatment Response at Mean for given $(V,X)$}]\label{condition:mono_resp_meanV}
The mean treatment responses are weakly increasing given $(V,X)=(v,x)$ for any $(v,x)\in{\mathcal{V}}\times\mathcal{X}$, i.e., $\mathbb{E}[Y_1\mid V=v,X=x]\geq \mathbb{E}[Y_0\mid V=v,X=x]$.
\end{condition}

\begin{condition}[{Monotone Treatment Response at Mean for given $X$}]\label{condition:mono_resp_relax}The mean treatment responses are weakly increasing given $X=x$ for any $x\in\mathcal{X}$, i.e.,
$\mathbb{E}[Y_1\mid X=x]\geq \mathbb{E}[Y_0\mid X=x]$.
\end{condition}

 Conditions \ref{condition:mono_resp_meanV} and \ref{condition:mono_resp_relax} only require monotone treatment response at the mean value and are weaker versions of Condition \ref{condition:mono_resp}. Specifically, Condition \ref{condition:mono_resp_meanV} implies restrictions on $(m_{\mathbb{P},0},m_{\mathbb{P},1})$ that $$m_{\mathbb{P},1}(v,x)-m_{\mathbb{P},0}(v,x)=\mathbb{E}[Y_1-Y_0\mid V=v,X=x]\geq 0,$$ 
 allowing for heterogeneous treatment responses across individuals for some random reasons.

For Condition \ref{condition:mono_resp_relax}, since $\mathbb{E}[Y_1-Y_0\mid X=x]=\mathbb{E}\{\mathbb{E}[Y_1-Y_0\mid V,X=x]\mid X=x\}=\mathbb{E}[m_{\mathbb{P},1}(V,x)-m_{\mathbb{P},0}(V,x)]$, it imposes the sign restriction  on $(m_{\mathbb{P},0},m_{\mathbb{P},1})$ that $$\int_{{\mathcal{V}}}(m_{\mathbb{P},1}(v,x)-m_{\mathbb{P},0}(v,x))dv\geq0,$$
allowing for heterogeneous treatment responses across individuals with different values of $V$ for any given $X=x$. A similar condition is explored by \citet{carlos2019average} to study the ATE bounds.

\subsubsection{Restrictions on $(m_{\mathbb{P},0},m_{\mathbb{P},1},m_{\mathbb{P},0D},m_{\mathbb{P},1D})$}
In this section, we discuss shape restrictions which impose joint restrictions on $(m_{\mathbb{P},0},m_{\mathbb{P},1},m_{\mathbb{P},0D},m_{\mathbb{P},1D})$. \citet{manski2000monotone} consider a monotone treatment selection
assumption, which states that the mean value of the potential outcome is weakly increasing in
the treatment. 
\begin{condition}[{Monotone Treatment Selection}]\label{condition:mono_select}For $d=0,1$ and any $x\in\mathcal{X}$,
$\mathbb{E}[Y_d\mid X=x,D=1]\geq \mathbb{E}[Y_d\mid X=x,D=0]$.
\end{condition}
Given Assumption \ref{assumption_DGP}, Condition \ref{condition:mono_select} implies the following two inequalities\footnote{See detailed proofs in Appendix \ref{appendix_proofs}.}
\begin{equation}\label{MTS_constraints} 
\begin{aligned}
\mathbb{E}[Y\mid X=x,D=1]\mathbb{P}(D=0\mid X=x)~\geq&~\mathbb{E}\Big[\int_{{\mathcal{V}}}(m_{\mathbb{P},1}(v,x)-m_{\mathbb{P},1D}(v,Z))dv~\Big|~ X=x\Big],\\
\mathbb{E}[Y\mid X=x,D=0]\mathbb{P}(D=1\mid X=x)~\leq&~\mathbb{E}\Big[\int_{{\mathcal{V}}}(m_{\mathbb{P},0}(v,x)-m_{\mathbb{P},0D}(v,Z))dv~\Big|~ X=x\Big].
\end{aligned}
\end{equation} 
Another shape restriction is the monotone selection on the gains, which states that the individuals who self-select into the treatment would gain more from being treated compared to those who remain untreated.
\begin{condition}[{Monotone Selection Gain}]\label{condition:mono_select_gain}
For any $x\in\mathcal{X}$, $\mathbb{E}[Y_1-Y_0\mid X=x,D=1]\geq \mathbb{E}[Y_1-Y_0\mid X=x,D=0]$.
\end{condition}
Under Assumption \ref{assumption_DGP}, similar arguments used to show the inequalities in \eqref{MTS_constraints} can be applied to show that Condition \ref{condition:mono_select_gain} implies the following restriction:
\begin{equation}\label{MSG_constraints}
\begin{aligned}
&\mathbb{E}[Y\mid X=x,D=1]+\mathbb{E}[Y\mid X=x,D=0]\\
&~~~\geq~\mathbb{E}\Big[\int_{{\mathcal{V}}}(m_{\mathbb{P},1}(v,x)-m_{\mathbb{P},1D}(v,Z))dv\Big|~ X=x\Big]/\mathbb{P}(D=0\mid X=x)\\
&~~~~~~+\mathbb{E}\Big[\int_{{\mathcal{V}}}(m_{\mathbb{P},0}(v,x)-m_{\mathbb{P},0D}(v,Z))dv\Big|~ X=x\Big]/\mathbb{P}(D=1\mid X=x).
\end{aligned}
\end{equation}

\section{Identification Analysis}\label{section:identification_analysis}

In this section, we investigate the identification of the target PRTE $\Gamma^\star(m_{\mathbb{P}})$. Recall that the functional
$\Gamma^\star$ is identifiable, and we only need to bound the unknown function $m_{\mathbb{P}}$. To achieve this goal, we first derive constraints that characterize the identified set for $m_{\mathbb{P}}$. Next, we apply a convex relaxation method to obtain computationally feasible PRTE bounds, which we will refer to as the convex-relaxation bounds. Then, we provide sufficient conditions under which the convex-relaxation bounds result in no identification power loss.

\subsection{Identified set for $m_{\mathbb{P}}$}
We begin by deriving a set of conditional moment equalities that are satisfied by $m_{\mathbb{P}}$. Under Assumption \ref{assumption_DGP}, by the law of iterated expectation and the conditional mean independence between $Y_d$ and $D$ given $(V,Z)$, we have 
\begin{align*}
\mathbb{E}[YD\mid Z]
=&
\mathbb{E}[\mathbb{E}[Y_1\mid D=1,Z,V]\mathbb{E}[D\mid Z,V]\mid Z]
\\
=&
\mathbb{E}[\mathbb{E}[Y_1\mid X,V]\mathbb{E}[D\mid Z,V]\mid Z]
\\
=&
\mathbb{E}[m_{\mathbb{P},1}(V,X)m_{\mathbb{P},D}(V,Z)\mid Z]
\\
=&
\mathbb{E}[m_{\mathbb{P},1D}(V,Z)\mid Z].
\end{align*}
Similarly, we can show that $\mathbb{E}[Y(1-D)\mid Z]=\mathbb{E}[m_{\mathbb{P},0D}(V,Z)\mid Z]$ and $\mathbb{E}[D\mid Z]=\mathbb{E}[m_{\mathbb{P},D}(V,Z)\mid Z]$. For any generic function $m=(m_{0},m_{1},m_{D},m_{0D},m_{1D})'\in\mathcal{M}$, let us introduce the following restrictions on $m$: 
\begin{align}
\label{sharp_1}
\mathbb{E}[YD\mid Z]=&\int_{{\mathcal{V}}}m_{1D}(v,Z)dv,\\
\label{sharp_2}
\mathbb{E}[Y(1-D)\mid Z]=&
\int_{{\mathcal{V}}}m_{0D}(v,Z)dv,\\
\label{sharp_3}
\mathbb{E}[D\mid Z]=&\int_{{\mathcal{V}}}m_{D}(v,Z)dv.
\end{align}

\begin{theorem}\label{eq:theorem:sharpness}Under Assumption \ref{assumption_DGP},
Eqs. \eqref{sharp_1}\text{ to }\eqref{sharp_3} hold for  $m_{\mathbb{P}}$.
\end{theorem}

The conditional moment conditions in Eqs. \eqref{sharp_1}\text{ to }\eqref{sharp_3} may not be appealing for estimation and inference when $Z$ is continuously distributed. Following \cite{mogstad/santos/torgovitsky:2017}, we use IV-like functions to transform the conditional moment conditions into unconditional ones. 
Let $\mathcal{S}$ denote a (user-specified) subset of $\mathbf{L}^2(\mathcal{Z},\mathbb{R})$.
Then, we use an IV-like function $s\in\mathcal{S}$ as a weighting function to integrate the conditional moments in Eqs. \eqref{sharp_1}\text{ to }\eqref{sharp_3}.
Namely, Theorem \ref{eq:theorem:sharpness} implies that $m_{\mathbb{P}}$ satisfies the following unconditional moment conditions for every $s\in\mathcal{S}$:\footnote{The IV-like function $s\in\mathcal{S}$ plays the same role as an instrumental variable. For example, in a simple case where $Z_0$ is one-dimensional, if we set $s(z)=\frac{z_0-\mathbb{E}[Z_0]}{\mathbb{C}ov(D,Z_0)}$, then Eqs. \eqref{intro_contraint_Gamma_s0} and \eqref{intro_contraint_Gamma_s1} together imply that $\mathbb{E}\Big[\int_{{\mathcal{V}}}m_{\mathbb{P},1D}(v,Z)s(Z)dv\Big]+\mathbb{E}\Big[\int_{{\mathcal{V}}}m_{\mathbb{P},0D}(v,Z)s(Z)dv\Big]=\frac{\mathbb{C}ov(Y,Z_0)}{\mathbb{C}ov(D,Z_0)}$, where the right-hand side is the identifiable IV estimand.} 
\begin{align}
\label{intro_contraint_Gamma_s0}
\mathbb{E}[YDs(Z)]=&\mathbb{E}\Big[\int_{{\mathcal{V}}}m_{1D}(v,Z)s(Z)dv\Big],\\
\label{intro_contraint_Gamma_s1}
\mathbb{E}[Y(1-D)s(Z)]=&\mathbb{E}\Big[\int_{{\mathcal{V}}}m_{0D}(v,Z)s(Z)dv\Big],\\
\label{intro_contraint_Lambda_s}
\mathbb{E}[Ds(Z)]=&\mathbb{E}\Big[\int_{{\mathcal{V}}}m_{D}(v,Z)s(Z)dv\Big].
\end{align}

\begin{theorem}\label{theorem_sharp}
Under Assumption \ref{assumption_DGP}, Eqs. \eqref{intro_contraint_Gamma_s0} to \eqref{intro_contraint_Lambda_s} hold for $m_{\mathbb{P}}$ and $\forall s\in\mathcal{S}\subseteq\mathbf{L}^2(\mathcal{Z},\mathbb{R})$.
\end{theorem}

Recall that we denote $m=(m_0,m_1,m_D,m_{0D},m_{1D})'$ as some generic function in the space $\mathcal{M}$. To characterize the constraints on $m$ that is compatible with $m_{\mathbb{P}}$, let us introduce the following notations. 
For any $s\in\mathcal{S}\subseteq \mathbf{L}^2(\mathcal{Z},\mathbb{R})$ and $m\in\mathcal{M}$, let us
define 
$$
\Gamma_{s}(m)=\mathbb{E}\Big[s(Z)\int_{{\mathcal{V}}}(m_{1D}(v,Z),m_{0D}(v,Z),m_D(v,Z))'dv\Big].
$$
Then, there are two sets of constraints under which $m\in \mathcal{M}$ is compatible with $m_{\mathbb{P}}$. The first set includes the following nonlinear constraints that come from the definition of $m_{\mathbb{P},dD}$:
\begin{equation}\label{cond:bilinear}
\begin{aligned}
m_{0D}(v,z)=&m_0(v,x)(1-m_D(v,z))\\
m_{1D}(v,z)=&m_1(v,x)m_D(v,z),~~\qquad\text{ for }  \forall (v,z)\in{\mathcal{V}}\times \mathcal{Z}.
\end{aligned}
\end{equation}
The second set includes the linear constraints in Eqs. \eqref{intro_contraint_Gamma_s0} to \eqref{intro_contraint_Lambda_s} that are imposed by identifiable estimands using observed data:
\begin{equation}\label{cond:Gamma}
\mathbb{E}[s(Z)(YD,Y(1-D),D)']=\Gamma_{s}(m),\mbox{ for every }s\in\mathcal{S}.
\end{equation}
For a given $\mathcal{S}$, let us define $\mathcal{M}_{\mathcal{S}}$ as a collection of  functions in $\mathcal{M}$ that satisfy both the above linear and nonlinear constraints
\begin{align}\label{def_MS}
\mathcal{M}_{\mathcal{S}}=\{m\in\mathcal{M}:~m\text{ satisfies Eqs. \eqref{cond:bilinear} and \eqref{cond:Gamma} }\mbox{ for every }s\in\mathcal{S}\}.
\end{align}
We can see from Theorem \ref{theorem_sharp} that  $m_\mathbb{P}\in\mathcal{M}_{\mathcal{S}}$. However,  for any given $\mathcal{S}$, $\mathcal{M}_{\mathcal{S}}$ is only an outer identified set for all $m\in\mathcal{M}$ that are compatible with $m_{\mathbb{P}}$, as it may not fully utilize all information from the observed data. If $\mathcal{S}$ expands to include more IV-like functions, then more restrictions are imposed on
$m$, and $\mathcal{M_S}$ will become tighter. We formalize this result in the following corollary.

\begin{corollary}\label{cor:bilinar_reprentation}
Under Assumption \ref{assumption_DGP}, we have $m_\mathbb{P}\in\mathcal{M}_{\mathcal{S}}$ for any $\mathcal{S}\subseteq\mathbf{L}^2(\mathcal{Z},\mathbb{R})$. In addition, if $\mathcal{S}=\mathbf{L}^2(\mathcal{Z},\mathbb{R})$, then $\mathcal{M}_{\mathcal{S}}$ satisfies
$$\mathcal{M}_{\mathcal{S}}=\{m\in\mathcal{M}:~m\mbox{ satisifes Eqs. }\eqref{sharp_1}\text{ to \eqref{sharp_3} and \eqref{cond:bilinear}}\}.
$$
\end{corollary}
Corollary \ref{cor:bilinar_reprentation} provides a sufficient condition under which $\mathcal{M}_{\mathcal{S}}$ reduces to the identified set of $m_{\mathbb{P}}$ characterized by the conditional moment constraints in Eqs. \eqref{sharp_1}\text{ to }\eqref{sharp_3} and the constraints imposed by the definition in Eq. \eqref{cond:bilinear}. Specifically, it requires the choice of the IV-like functions in $\mathcal{S}$ to exhaust all available information in the data, which is not restrictive. For example, if $Z\in\{0,1\}$ is a binary scalar, then $\mathcal{S}=\{1[z=0],1[z=1]\}$ satisfies this condition. If $Z$ is continuous, then $\mathcal{S}=\{1[z\leq z']:~z'\in\mathcal{Z}\}$ needs to consist of indicators for all half-spaces of $\mathcal{Z}$. 
For any given $\mathcal{M_S}$, we define the identified set for the target PRTE parameter in Eq. \eqref{def_target_parameters_linear} as 
$$\{\Gamma^\star(m): m\in\mathcal{M_S}\}.$$

\subsubsection{Sharpness}\label{section:sharp}
In this section, we define the sharp identified set for $m_{\mathbb{P}}$ and investigate sufficient conditions under which $\mathcal{M_S}$ is sharp.  The sharp identified set for the target PRTE is then defined as the set of all possible values of $\Gamma^\star(m)$, where $m$ is in the sharp identified set for $m_\mathbb{P}$. 


\vspace{5mm}

\noindent\textbf{Definition of Sharp Identified Set}. \hspace{3mm} Let us first introduce some useful notations. For any generic distribution $P$ for  $(Y_0,Y_1,D,Z,V)$, we denote $$m_P=(m_{P,0},m_{P,1},m_{P,D},m_{P,0D},m_{P,1D})',$$ where $m_{P,d}$ with $d=0,1$ and $m_{P,D}$ are the marginal treatment response function and the propensity score function induced by $P$, respectively, and we define $m_{P,0D}=m_{P,0}(1-m_{P,D})$ and $m_{P,1D}=m_{P,1}m_{P,D}$. Let $\mathcal{P}$ be the set of distributions $P$ for $(Y_0,Y_1,D,Z,V)$ that satisfy the following conditions: (i) $P_{Y_D,D,Z}=\mathbb{P}_{Y,D,Z}$, (ii) $E_P[Y_d\mid D,Z,{V}]=E_P[Y_d\mid X,{V}]$ and $E_P[{E_P}[Y_d\mid X,{V}]^2]<\infty$ for each $d=0,1$, and (iii) $V$ conditional on $Z$ is uniformly distributed on $\mathcal{V}$ with a known dimension $K$.\footnote{We use $P_{Y_D,D,Z}$ to denote the distribution for $(Y_D,D,Z)$ induced by $P$, and we use $\mathbb{P}_{Y,D,Z}$ to denote the distribution for $(Y,D,Z)$ induced by the true distribution $\mathbb{P}$.} In other words, $\mathcal{P}$ is the set of distributions that satisfy Assumption \ref{assumption_DGP} and are observationally equivalent to $\mathbb{P}$. Then, the sharp identified set for $m_{\mathbb{P}}$ can be defined by 
$$\{m_{P}\in\mathcal{M}: P\in\mathcal{P}\}.$$ 
If $m_\mathbb{P}$ satisfies certain shape restriction of the form $m_\mathbb{P}\in\mathcal{R}$ for a known set $\mathcal{R}\subseteq\mathcal{M}$, such as those discussed in Section \ref{section:shape_restriction}, then the sharp identified set for $m_{\mathbb{P}}$ can be defined as
$$
\{m_{P}\in\mathcal{R}: P\in\mathcal{P}\}.
$$

\vspace{5mm}

\noindent\textbf{Sharpness of $\mathcal{M_S}$ with Binary Outcome}. \hspace{3mm} Next, we establish the sharpness of the identified set $\mathcal{M_S}$ defined in \eqref{def_MS} in an empirically relevant case with a binary outcome variable.\footnote{It is interesting to explore sharpness in cases with general outcome variables. However, we leave a rigorous exploration for future research.} In the theorem below, we show that if we utilize all available information from observed data by choosing $\mathcal{S}=\mathbf{L}^2(\mathcal{Z},\mathbb{R})$, then $\mathcal{M_S}$ is the sharp identified set for $m_\mathbb{P}$.
\begin{theorem}\label{theorem:sharp_binary_P_m}
Under Assumption \ref{assumption_DGP}, if $Y\in\{0,1\}$ and $m_\mathbb{P}\in\mathcal{R}$ for some known set $\mathcal{R}\subseteq\mathcal{M}$, then
$$
\{m\in\mathcal{R}: ~m\in \mathcal{M_S}\mbox{ with }\mathcal{S}=\mathbf{L}^2(\mathcal{Z},\mathbb{R})\}=\{m_{P}\in\mathcal{R}: P\in\mathcal{P}\}.
$$
\end{theorem}
When $\mathcal{M_S}$ is the sharp identified set for $m_\mathbb{P}$, the identified set for PRTE based on $\mathcal{M_S}$ is also sharp. Despite the sharpness demonstrated in Theorem \ref{theorem:sharp_binary_P_m}, the PRTE identified set can be computationally challenging to obtain, as it requires solving nonconvex optimization problems due to the nonlinear constraints in \eqref{cond:bilinear}. In what follows, we address this challenge by introducing a convex relaxation method.

\subsection{A Convex Relaxation Method}

In the objective function $\Gamma^\star(m)$ and the constraints on $m\in\mathcal{M_S}$, everything is linear in $m$ except for the nonlinear (in particular, bilinear) constraints in Eq. \eqref{cond:bilinear}. 
In this section, we propose to use the convex relaxation method proposed by \cite{mccormick1976computability} to relax the bilinear constraints in  Eq. \eqref{cond:bilinear} into linear constraints. This then allows us to bound $\Gamma^\star(m)$ by solving linear optimization problems. Let us first introduce the convex-relaxation presentation of the bilinear constraints on $m_{\mathbb{P},d}$. 
\begin{lemma}[\citealp{mccormick1976computability}, p.153]\label{lemma:mccormick}
Suppose there are known (or identifiable) functions $m_d^L$, $m_d^U$, $m_D^L$, and $m_D^U$, such that $m_{\mathbb{P}}$ satisfies the following inequalities for all $(v,z)\in\mathcal{V}\times\mathcal{Z}$
\begin{align}\label{eq:mccormick_bounds_m}
\begin{cases}
m_d^L(v,x)\leq m_{d}(v,x)\leq m_d^U(v,x), \text{ for }d=0,1,\\
m_D^L(v,z)\leq m_{D}(v,z)\leq m_D^U(v,z).\end{cases}
\end{align}
Then, any $m_\mathbb{P}$ that satisfies Eq. \eqref{cond:bilinear} also satisfies 
\begin{flalign}\hspace*{-0.2cm}\label{eq:mccormick_relax1_0}
\begin{cases}
m_{0D}(v,z)\geq m_0^L(v,x)\left(1-m_{D}(v,z)\right)+m_{0}(v,x)\left(1-m_D^U(v,z)\right)-m_0^L(v,x)\left(1-m_D^U(v,z)\right)\\
m_{0D}(v,z)\geq m_0^U(v,x)\left(1-m_{D}(v,z)\right)+m_{0}(v,x)\left(1-m_D^L(v,z)\right)-m_0^U(v,x)\left(1-m_D^L(v,z)\right)\\
m_{0D}(v,z)\leq m_0^U(v,x)\left(1-m_{D}(v,z)\right)+m_{0}(v,x)\left(1-m_D^U(v,z)\right)-m_0^U(v,x)\left(1-m_D^U(v,z)\right)\\
m_{0D}(v,z)\leq m_0^L(v,x)\left(1-m_{D}(v,z)\right)+m_{0}(v,x)\left(1-m_D^L(v,z)\right)-m_0^L(v,x)\left(1-m_D^L(v,z)\right),
\end{cases}
\end{flalign}
and
\begin{align}\label{eq:mccormick_relax4_1}
\begin{cases}
m_{1D}(v,z)\geq m_1^L(v,x)m_{D}(v,z)+m_{1}(v,x)m_D^L(v,z)-m_1^L(v,x)m_D^L(v,z)\\
m_{1D}(v,z)\geq m_1^U(v,x)m_{D}(v,z)+m_{1}(v,x)m_D^U(v,z)-m_1^U(v,x)m_D^U(v,z)\\
m_{1D}(v,z)\leq m_1^U(v,x)m_{D}(v,z)+m_{1}(v,x)m_D^L(v,z)-m_1^U(v,x)m_D^L(v,z)\\
m_{1D}(v,z)\leq m_1^L(v,x)m_{D}(v,z)+m_{1}(v,x)m_D^U(v,z)-m_1^L(v,x)m_D^U(v,z),
\end{cases}
\end{align}\mbox{ for all }$(v,z)\in\mathcal{V}\times\mathcal{Z}$.
\end{lemma}
\begin{proof}[Proof of Lemma \ref{lemma:mccormick}] 
The proofs are similar for all eight inequalities, so we focus on the first inequality in \eqref{eq:mccormick_relax1_0}. By definition, $m_{\mathbb{P},0D}(v,z)=m_{\mathbb{P},0}(v,x)(1-m_{\mathbb{P},D}(v,z))$ and $m_{\mathbb{P},1D}(v,z)=m_{\mathbb{P},1}(v,x)m_{\mathbb{P},D}(v,z)$.
Then, we have that $m_{\mathbb{P}}$ satisfies
\begin{align*}
m_{0D}(v,z)
&=
m_0(v,x)(1-m_D(v,z))
\\
&=
m_0^L(v,x)(1-m_D(v,z))
+m_0(v,x)(1-m_D^U(v,z))
-m_0^L(v,x)(1-m_D^U(v,z))
\\&\quad+(m_0(v,x)-m_0^L(v,x))(m_D^U(v,z)-m_D(v,z)).
\end{align*}
The last term is always non-negative, which results in the desirable inequality. 
\end{proof}

In Lemma \ref{lemma:mccormick},  we assume the bounds $m^L_d(v,x)$, $m^U_d(v,x)$, $m^L_D(v,z)$, and $m^U_D(v,z)$ are either known or identifiable. If no prior information about these upper and lower bounds is available, we
can set $m^L_D(v,z)$ and $m^U_D (v,z)$ to 0 and 1, respectively, since $D$ is a binary variable, and set $m_d^L(v,x)$ and $m_d^U(v,x)$ to $y^L$ and $y^U$, respectively, if they are finite. From Lemma \ref{lemma:mccormick}, we can see that after applying the convex relaxation, the equality constraints in Eq. \eqref{cond:bilinear}, which are nonlinear in $m_{\mathbb{P}}$, can be relaxed and replaced by the inequality constraints in Eqs. \eqref{eq:mccormick_bounds_m} to \eqref{eq:mccormick_relax4_1}, which are all linear in $m_{\mathbb{P}}$. The boundaries of the set for $m_{\mathbb{P},dD}$, characterized by the inequalities in \eqref{eq:mccormick_relax1_0} and \eqref{eq:mccormick_relax4_1}, are referred to as the McCormick envelopes. 

Under convex relaxation, for any given user-specified $\mathcal{S}$, let us define  
$$\mathcal{M}^r_\mathcal{S}=\{m\in \mathcal{M}:~m\text{ satisfies Eq. \eqref{cond:Gamma} for any $s\in \mathcal{S}$, and Eqs. \eqref{eq:mccormick_bounds_m} to \eqref{eq:mccormick_relax4_1}}\}.$$
Then, $\mathcal{M}^r_\mathcal{S}$ is the collection of functions $m\in \mathcal{M}$ characterized by linear constraints from observed data and convex relaxation. Thus, it is an outer set for all $m\in\mathcal{M}$ that are compatible with $m_\mathbb{P}$ even if $S$ is a set of sufficiently rich class of IV-like functions. 

\begin{theorem}\label{theorem:main_result_relax}
Under Assumption \ref{assumption_DGP} and the assumptions in Lemma \ref{lemma:mccormick}, $m_\mathbb{P}\in \mathcal{M}^r_\mathcal{S}$ and $\mathcal{M}_\mathcal{S}\subseteq\mathcal{M}^r_\mathcal{S}$ for any $\mathcal{S}\subseteq\mathbf{L}^2(\mathcal{Z},\mathbb{R})$.
\end{theorem}
Based on $\mathcal{M}^r_\mathcal{S}$, we can define an outer set for the target PRTE parameter as below:
\begin{equation}\label{eq:set_representation+relax}
\{\Gamma^\star(m): m\in\mathcal{M}^r_\mathcal{S}\}.
\end{equation}
Since both $\Gamma^\star(m)$ and the constraints that characterize $\mathcal{M}^r_\mathcal{S}$ are linear (therefore, convex) in $m$, under Assumption \ref{assumption_DGP}, the PRTE outer set in \eqref{eq:set_representation+relax} is equivalent to an interval. Let us denote this interval as $[\underline\beta^\star,\overline\beta^\star]$. Then, we can characterize the two ending points $\underline\beta^\star$ and $\overline\beta^\star$ using the following linear optimization problems:\footnote{Note that we omit the dependence of $\underline\beta^\star$ and $\overline\beta^\star$ on $\mathcal{S}$ for simplicity. }
\begin{equation}\label{opt_relax}
\underline\beta^\star=\inf_{m\in\mathcal{M}^r_\mathcal{S}}\Gamma^\star(m) ~\text{ and } ~\overline\beta^\star=\sup_{m\in\mathcal{M}^r_\mathcal{S}}\Gamma^\star(m).
\end{equation}

There are a few advantages of considering the outer set $\{\Gamma^\star(m): m\in\mathcal{M}^r_\mathcal{S}\}$. First, it gives us computationally tractable PRTE bounds because $\underline\beta^\star$ and $\overline\beta^\star$ can be easily and reliably solved by linear programming techniques that are routinely used in many empirical works.  Second, shape restrictions that are linear in $m$ can be easily incorporated into the linear programming problems to further shrink the bounds. Third, under some conditions, we can show that $\underline\beta^\star$ and $\overline\beta^\star$ can be computed without knowing the true dimension of the unobserved heterogeneity $V$. To illustrate this, let us denote
\begin{align}\label{prop_reduce_dim_fun}
\tilde{m}_d(x)=&\int_{\mathcal{V}}m_d(v,x)dv,~~
\tilde{m}_D(z)=\int_{\mathcal{V}}m_D(v,z)dv,~~
\tilde{m}_{dD}(z)=\int_{\mathcal{V}}m_{dD}(v,z)dv.
\end{align}
Let $\tilde{m}=(\tilde{m}_0,\tilde{m}_1,\tilde{m}_D,\tilde{m}_{0D},\tilde{m}_{1D})'\in \tilde{\mathcal{M}}$, where
$
\tilde{\mathcal{M}}=\{\tilde{m}:~\tilde{m}\mbox{ defined  in }\eqref{prop_reduce_dim_fun}\mbox{ for some }m\in\mathcal{M}\}$.

\begin{proposition}\label{prop_reduce_dim}Suppose $\omega_{dd'}^\star(V,Z)$, $m^L_d(V,X)$, $m^U_d(V,X)$, $m^L_{D}(V,Z)$, and $m^U_{D}(V,Z)$ do not depend on $V$ for all $d,d'=0,1$. For any $\mathcal{S}$, if there exists a $m^\star\in \mathcal{M}$ such that $m^\star=\arg\inf_{m\in\mathcal{M}^r_\mathcal{S}}\Gamma^\star(m)$, then we have
$$\tilde{m}^\star=\arg\inf_{\tilde{m}\in\tilde{\mathcal{M}}^r_\mathcal{S}}\Gamma^\star(\tilde{m}),$$
where $\tilde{m}^\star\in \tilde{\mathcal{M}}$ is obtained from substituting $m^\star$ into \eqref{prop_reduce_dim_fun}, and we define $\tilde{\mathcal{M}}^r_\mathcal{S}=\{\tilde{m}\in \tilde{\mathcal{M}}:~\tilde{m}\mbox{ satisfies  
\eqref{cond:Gamma} for any $s\in\mathcal{S}$ and \eqref{eq:mccormick_bounds_m} to \eqref{eq:mccormick_relax4_1}}\}$. 
Moreover, we have $\underline\beta^\star=\Gamma^\star(m^\star)=\Gamma^\star(\tilde{m}^\star)$.
The same result holds for the supremum. 
\end{proposition}

Proposition \ref{prop_reduce_dim} provides sufficient conditions under which our proposed convex-relaxation bounds for PRTEs can be obtained by solving linear optimization problems over the space $\tilde{\mathcal{M}}$. Given that any $\tilde{m}\in\tilde{\mathcal{M}}$ is degenerate in $V$, the solutions to these linear optimization problems over $\tilde{\mathcal{M}}$ do not rely on knowledge of the true dimension of $V$. This notable result is contingent upon a crucial condition: the weights $\omega_{dd'}^\star(V,Z)$ and the known bounds for $m_{\mathbb{P},d}$ and $m_{\mathbb{P},D}$ must be degenerate in $V$. Specifically, the former, $\omega_{dd'}^\star(V,Z)=\omega_{dd'}^\star(Z)$, holds for many target PRTE estimands listed in Table \ref{table_weights}. The latter holds if we set constant upper and lower bounds for $m_{\mathbb{P},d}$ and $m_{\mathbb{P},D}$. 

\subsection{No Identification Power Loss of Convex-Relaxation Bounds}\label{section:no_id_power_loss} 

Despite the advantages of the convex relaxation method mentioned above, its conservativeness may raise concerns. This is because the resulting identified set for $m_{\mathbb{P}}$, $\mathcal{M}^r_\mathcal{S}$, is only an outer set, which is often larger and less informative than its identified set $\mathcal{M_S}$ that does not rely on convex relaxation. This section investigates this concern. 
We consider two cases, and in each case, we provide sufficient conditions under which our convex-relaxation bounds result in no identification power loss. 

We first consider the case where the known or identifiable lower and upper bounds of $m_D(V,Z)$ used in the convex relaxation coincide, implying that $m_D(V,Z)$ is identifiable. If researchers are aware of this and use it properly by setting $m_D^L(V,Z)=m_D^U(V,Z)=m_D(V,Z)$ when applying the convex relaxation method, we can show that the outer set for $\Gamma^\star(m)$ derived from $\mathcal{M}^r_\mathcal{S}$ reduces to its identified set derived from $\mathcal{M_S}$.

\begin{theorem}\label{theorem:sharp_binary_determonistic_monot}
Suppose Assumption \ref{assumption_DGP}, the assumptions in Lemma \ref{lemma:mccormick}, and $m_D^L(V,Z)=m_D^U(V,Z)$ hold. For any $\mathcal{S}\subseteq\mathbf{L}^2(\mathcal{Z},\mathbb{R})$ and any $\mathcal{R}\subseteq\mathcal{M}$, if we set $m_d^L(V,X)=y^L$ and $m_d^U(V,X)=y^U$ for $d=0,1$, and $m_D^L(V,Z)=m_D^U(V,Z)$ in \eqref{eq:mccormick_bounds_m}, then $$\{\Gamma^\star(m): m\in\mathcal{R}\mbox{ and }m\in \mathcal{M}^r_\mathcal{S}\}= \{\Gamma^\star(m): m\in\mathcal{R} \mbox{ and }m\in\mathcal{M_S}\}.$$
\end{theorem}
The result in Theorem \ref{theorem:sharp_binary_determonistic_monot} is straightforward because under the condition $m_D^L(V,Z)=m_D^U(V,Z)$, the inequality
constraints in \eqref{eq:mccormick_bounds_m} to \eqref{eq:mccormick_relax4_1} all reduce to the equality constraints in \eqref{cond:bilinear}. We can see that $m_D^L(V,Z)=m_D^U(V,Z)$ holds, if $m_{\mathbb{P},D}(V,Z)$ is point identified. This is the case if we impose the shape restriction in Condition \ref{condition:random_coeff} that $D$ is generated by a random coefficient model, where $m_{\mathbb{P},D}(V,Z)=1[\pi_1Z_1+\pi_2Z_2+...+\pi_KZ_K>0]$ is point identified. Another example for the point identified $m_{\mathbb{P},D}(V,Z)$ is when we impose the shape restriction in Condition \ref{condition:double_hurdle} that $D$ is generated by a double hurdle model, where $m_{\mathbb{P},D}(V,Z)=1[U_1<Q_1(Z)\text{ and }U_2<Q_2(Z)]$ is point identified. Besides, if we impose the classical monotonicity restriction on $D$, which is equivalent to assuming that $D$ is generated by a threshold-crossing model with a single-dimensional $V\sim \mathrm{Uniform}[0,1]$, then $m_{\mathbb{P},D}(V,Z)=1[\mathbb{P}[D=1\mid Z]\geq V]$ is also point identified.

By combining Theorems \ref{theorem:sharp_binary_P_m} and \ref{theorem:sharp_binary_determonistic_monot}, we can see that in the case of a binary outcome, if $m_{\mathbb{P},D}(V,Z)$ is point identified, the convex-relaxation bounds for PRTEs are sharp as long as we choose IV-like functions to exhaust all the information from observed data.
Such a sharpness result directly relates to the existing literature. For example, \cite{heckman/vylacil:2001:book} propose the sharp ATE bound assuming bounded outcomes and threshold crossing in the treatment equation. 
Below, we show that 
our convex-relaxation bounds for ATE without shape restrictions coincide with the sharp ATE bounds of \cite{heckman/vylacil:2001:book}.

\begin{corollary}\label{coro_equiv_CvR_HV}
Consider the assumptions in Theorem \ref{theorem:sharp_binary_determonistic_monot} and assume $Y\in\{0,1\}$. 
Suppose there are no covariates $X$ and we are interested in the ATE, that is, $\omega_{00}^\star(v,z)=\omega_{01}^\star(v,z)=-1$ and $\omega_{10}^\star(v,z)=\omega_{11}^\star(v,z)=1$. 
Assume $K=1$ and
$$
m_{\mathbb{P},D}(V,Z)=1\{\mathbb{P}(D=1\mid Z)\geq V\},\text{ with }V\sim \mathrm{Uniform}[0,1].
$$
If we set $\mathcal{S}=\mathbf{L}^2(\mathcal{Z},\mathbb{R})$,  $m_D^L(V,Z)=m_D^U(V,Z)=1\{\mathbb{P}(D=1\mid Z)\geq V\}$, and $m_d^L(V,Z)=0$ and $m_d^U(V,Z)=1$ for $d=0,1$, 
then our convex relaxation bounds $[\underline\beta^\star,\overline\beta^\star]$ defined in \eqref{opt_relax} are the same as the ATE bounds of \cite{heckman/vylacil:2001:book}:
\begin{align*}
\underline\beta^\star=&p^U\mathbb{E}[Y\mid p(Z)=p^U,D=1]-(1-p^L)\mathbb{E}[Y\mid p(Z)=p^L,D=0]-p^L, \\
\overline\beta^\star=&p^U\mathbb{E}[Y\mid p(Z)=p^U,D=1]+(1-p^U)-(1-p^L)\mathbb{E}[Y\mid p(Z)=p^L,D=0],
\end{align*}
where we denote $p^L=\inf_{z\in\mathcal{Z}}\mathbb{P}(D=1\mid Z=z)\text{ and }p^U=\sup_{z\in\mathcal{Z}}\mathbb{P}(D=1\mid Z=z)$.
\end{corollary}

Next, we consider another case in which the convex-relaxation bounds are simplified to those derived by \citet*{mogstad/santos/torgovitsky:2017} (hereafter, MST) if the true data satisfy the threshold-crossing model for the treatment with one-dimensional unobserved heterogeneity, even though the researcher neither knows nor uses this information. We use the propensity score given $Z$, defined by $p(z)=\mathbb{P}(D=1|Z=z)$. Suppose the treatment variable is generated by a threshold-crossing model in the true data:
\begin{align}\label{model_MST}
D=1\{p(z)\geq U\}\mbox{ almost surely},
\end{align} 
where $U$ conditional on $Z$ is uniformly distributed on $[0,1]$, $\mathbb{E}[Y_d\mid D,Z,U]=\mathbb{E}[Y_d\mid X,U]$, and $\mathbb{E}[Y^2_d]<\infty$ for $d\in\{0,1\}$. Here, we use a different notation $U$ to represent the unobserved heterogeneity, in order to emphasize that the researcher may not know whether the true data satisfy the threshold-crossing assumption for the treatment, and instead, a multidimensional $V$ can be assumed when computing our convex-relaxation bounds. 

Let us first introduce the MST bounds assuming the threshold-crossing structure in \eqref{model_MST}. With notation abuse, we denote the marginal treatment response functions as $\tilde{m}_{\mathbb{P},d}(u,x)=\mathbb{E}[Y_d\mid U=u, X=x]$ for $d=0,1$. The propensity score given $(U,Z)$, i.e., $\mathbb{E}[D|U=u,Z=z]=1\{p(z)\geq u\}$, is point identified for any  $(u,z)\in[0,1]\times\mathcal{Z}$. Therefore, different from the expression of our target parameter given in \eqref{def_target_parameters}, the MST target parameter depends only on the marginal treatment response functions:
\begin{align}\label{def_MST_target}\mathbb{E}\left[\int_0^1\tilde{m}_{\mathbb{P},0}(u,X)\tau^\star_0(u,Z)du\right]+\mathbb{E}\left[\int_0^1\tilde{m}_{\mathbb{P},1}(u,X)\tau^\star_1(u,Z)du\right],
\end{align}
where $\tau^\star_d(u,z)$ for $d=0,1$ denotes some known or identifiable weights. 
Under the the threshold-crossing structure in \eqref{model_MST}, the expression in \eqref{def_MST_target} is equivalent to that of our target parameter in \eqref{def_target_parameters} if we set
\begin{equation}\label{weights_MST}\begin{aligned}
    \tau^\star_0(u,Z)=&1[u>p(Z)]\omega^\star_{00}(u,Z)+1[u\leq p(Z)]\omega^\star_{01}(u,Z),\\
    \tau^\star_1(u,Z)=&1[u>p(Z)]\omega^\star_{10}(u,Z)+1[u\leq p(Z)]\omega^\star_{11}(u,Z),
\end{aligned}
\end{equation}
and all PRTE parameters listed in Table \ref{table_weights} can be equivalently expressed as in \eqref{def_MST_target}, using the weights given in \eqref{weights_MST}. For any generic functions $(\tilde{m}_0,\tilde{m}_1)$ in some parameter space $\tilde{\mathcal{M}}_{01}\subseteq(\mathbf{L}^2([0,1]\times\mathcal{X},[y^L,y^U]))^2$, define a map from $\tilde{\mathcal{M}}_{01}$ to $\mathbb{R}$ as below:
$$\tilde{\Gamma}^\star(\tilde{m}_0,\tilde{m}_1)=\mathbb{E}\left[\int_0^1\tilde{m}_0(u,X)\tau^\star_0(u,Z)du\right]+\mathbb{E}\left[\int_0^1\tilde{m}_1(u,X)\tau^\star_1(u,Z)du\right].$$
In addition, define $\tilde{\mathcal{S}}$ to be a collection of $\tilde{s}(d,z)$, where $\tilde{s}:\{0,1\}\times\mathcal{Z}\mapsto\mathbb{R}$ denotes the identifiable (or known) IV-like function used by MST to construct identified sets for $\tilde{m}_{\mathbb{P},0}$ and $\tilde{m}_{\mathbb{P},1}$. 
For any generic functions $\tilde{m}_0$ and $\tilde{m}_1$, they consider the following restriction
\begin{align}\label{MST_restrictions}\mathbb{E}[\tilde{s}(D,Z)Y]=\mathbb{E}\left[\int_0^1\tilde{m}_0(u,X)\tilde{s}(0,Z)1[u>p(Z)]du+\int_0^1\tilde{m}_1(u,X)\tilde{s}(1,Z)1[u\leq p(Z)]du\right].
\end{align}
MST has shown that $\tilde{m}_{\mathbb{P},0}$ and $\tilde{m}_{\mathbb{P},1}$ must satisfy the restriction in \eqref{MST_restrictions} for any $\tilde{s}\in\tilde{\mathcal{S}}$. Define $$\tilde{\mathcal{M}}_{01,\tilde{\mathcal{S}}}=\left\{(\tilde{m}_0,\tilde{m}_1)\in \tilde{\mathcal{M}}_{01}:~\tilde{m}_0\text{ and }\tilde{m}_1\text{ satisfy \eqref{MST_restrictions} for all }\tilde{s}\in \tilde{\mathcal{S}}\right\}.$$  
Then, the MST bounds for the  target PRTE parameter $\tilde{\Gamma}^\star(\tilde{m}_0,\tilde{m}_1)$ is
\begin{equation}\label{MST_bounds}
\left[\inf_{(\tilde{m}_0,\tilde{m}_1)\in\tilde{\mathcal{M}}_{01,\tilde{\mathcal{S}}}}\tilde{\Gamma}^\star(\tilde{m}_0,\tilde{m}_1),~\sup_{(\tilde{m}_0,\tilde{m}_1)\in\tilde{\mathcal{M}}_{01,\tilde{\mathcal{S}}}}\tilde{\Gamma}^\star(\tilde{m}_0,\tilde{m}_1)\right].
\end{equation}

Even if the true data generating process satisfies the threshold-crossing assumption in \eqref{model_MST}, the researcher may not be aware of it and may instead compute our convex relaxation bounds. The theorem below provides sufficient conditions under which our convex relaxation bounds for $\Gamma^\star(m)$ coincide
with the MST bounds for $\tilde{\Gamma}^\star(\tilde{m}_0,\tilde{m}_1)$. 

\begin{theorem}\label{coro_equav_MST_CVR}
   Suppose Assumption \ref{assumption_DGP}, the assumptions in Lemma \ref{lemma:mccormick}, and the threshold-crossing structure in \eqref{model_MST} hold. If we further assume the following conditions
  \begin{itemize}
      \item[(i)] $\mathcal{S}=\mathbf{L}^2(\mathcal{Z},\mathbb{R})$ and $\tilde{\mathcal{S}}=\mathbf{L}^2(\{0,1\}\times\mathcal{Z},\mathbb{R})$,
       \item[(ii)]$\omega_{dd'}^\star(V,Z)=\omega_{dd'}^\star(Z)$ does not depend on $V$  for $d,d'=0,1$,
        \item[(iii)]$\tau^\star_0(u,Z)=1[u>p(Z)]\omega^\star_{00}(Z)+1[u\leq p(Z)]\omega^\star_{01}(Z)$ and $\tau^\star_1(u,Z)=1[u>p(Z)]\omega^\star_{10}(Z)+1[u\leq p(Z)]\omega^\star_{11}(Z)$,
         \item[(iv)]we set $m_d^L(v,x)=y^L$, $m_d^U(v,x)=y^U$, $m_D^L(v,x)=0$, $m_D^U(v,x)=1$ in Lemma \ref{lemma:mccormick},
  \end{itemize}
then our convex relaxation bounds $[\underline\beta^\star,\overline\beta^\star]$ defined in \eqref{opt_relax} are the same as MST bounds:
   $$[\underline\beta^\star,\overline\beta^\star]=\left[\inf_{(\tilde{m}_0,\tilde{m}_1)\in\tilde{\mathcal{M}}_{01,\tilde{\mathcal{S}}}}\tilde{\Gamma}^\star(\tilde{m}_0,\tilde{m}_1),~\sup_{(\tilde{m}_0,\tilde{m}_1)\in\tilde{\mathcal{M}}_{01,\tilde{\mathcal{S}}}}\tilde{\Gamma}^\star(\tilde{m}_0,\tilde{m}_1)\right].$$ 
 This result also holds if we impose the same linear shape restrictions on 
$(\int_{0}^1\tilde{m}_0(u,x)du,\int_{0}^1\tilde{m}_1(u,x)du)$ and $(\int_{\mathcal{V}}m_0(v,x)dv,\int_{\mathcal{V}}m_1(v,x)dv)$.
\end{theorem}

Theorem \ref{coro_equav_MST_CVR} provides sufficient conditions under which our convex relaxation method does not result in identification power loss compared to the MST bounds. That is, if the true data generating process satisfies the threshold-crossing assumption in \eqref{model_MST}, our convex-relaxation bounds—despite not utilizing this information—are just as tight as the MST bounds for a large class of PRTE parameters whose weights do not depend on $V$. Moreover, this result also holds under a particular set of linear shape restrictions that are imposed on the integrated treatment response functions. 
In short, our method generalizes MST, and remains applicable and more robust in settings where practitioners are uncertain whether the threshold-crossing structure in \eqref{model_MST} holds.

Theorem \ref{coro_equav_MST_CVR} may resemble Theorem 2 of \cite{heckman/vylacil:2001:book}, but they are substantively different results. In their Theorem 2, if the threshold-crossing assumption in \eqref{model_MST} holds, the \cite{manski1990nonparametric} IV-mean-independence sharp bounds, which do not impose the threshold crossing structure for the treatment, are equivalent to the  Heckman-Vytlacil sharp bounds. In our Theorem \ref{coro_equav_MST_CVR}, we do not consider the equivalence between the sharp bounds with and without the threshold-crossing assumption in \eqref{model_MST}. Instead, we consider the equivalence between our convex-relaxation bounds and the MST bounds.

\section{Inference Procedure and Algorithm}\label{section:inference}

In this section, we present our computation and inference procedure for PRTE bounds defined by the linear optimization problems in \eqref{opt_relax}. To solve the optimization problems using linear programming techniques, we first replace the potentially infinite-dimensional parameter space for $m$ with a finite-dimensional space defined using basis functions. Then, we demonstrate that the computationally tractable inference method of \citet{gafarov2019inference} can be directly applied in our analysis to construct uniformly valid PRTE confidence sets.

\subsection{Finite-Dimensional Approximation}\label{section:computation}
In a similar vein to \citet{mogstad/santos/torgovitsky:2017}, we approximate the unknown function $m$ using a finite number of known basis
functions. Let
$\mathbf{B}_d:(\mathcal{V,X})\mapsto\mathbb{R}^{K_d}$ with $d=0,1$ be a vector of known basis functions for the unknown function $m_d$. Similarly, let
$\mathbf{B}_{D}:(\mathcal{V,Z})\mapsto\mathbb{R}^{K_D}$ be a vector of known basis functions for $m_D$, and $\mathbf{B}_{dD}:(\mathcal{V,Z})\mapsto\mathbb{R}^{K_{dD}}$ with $d=0,1$ be a vector of known basis functions for $m_{dD}$. Then, we can approximate $m=(m_0,m_1,m_D,m_{0D},m_{1D})'$ by $\mathbf{B}'\eta_2$ with a $d_{\eta_2}\times 1$ vector of unknown parameters $\eta_2$ and a matrix of known basis functions $\mathbf{B}$, where 
$$
\mathbf{B}'\eta_2=\begin{bmatrix}\mathbf{B}'_0&0&0&0&0\\
0&\mathbf{B}'_1&0&0&0\\
0&0&\mathbf{B}'_D&0&0\\
0&0&0&\mathbf{B}'_{0D}&0\\
0&0&0&0&\mathbf{B}'_{1D}\\
\end{bmatrix}\eta_2.
$$
We denote  $\eta_1=\Gamma^\star(\mathbf{B}'\eta_2)$ and $\eta=(\eta_1,\eta_2')'$.  Define the parameter space for $\eta$ as below
$$\Upsilon=\{\eta=(\eta_1,\eta_2')'\in\mathbb{R}^{d_\eta}:~\eta_1=\Gamma^\star(\mathbf{B}'\eta_2)\in\mathbb{R}\text{ and }\mathbf{B}'\eta_2\in\mathcal{M}\}.$$
Then, the dimension of $\eta$ is $d_\eta=d_{\eta_2}+1=K_0+K_1+K_D+K_{0D}+K_{1D}$. Let $e_1=(1,0,...,0)'\in\mathbb{R}^{d_\eta}$. With this approximation, we can modify the optimization problems in \eqref{opt_relax} into 
\begin{equation}\label{linear_prom_orig_finite}
\underline\beta^\star_a=\inf_{\eta\in\Upsilon}e_1'\eta \quad \text{ and } \quad \overline\beta^\star_a=\sup_{\eta\in\Upsilon}
e_1'\eta
\end{equation}
$\eta$ subject to 
\begin{equation}\label{eq:constrainted_finite_relaxed}
\eta_1=\Gamma^\star(\mathbf{B}'\eta_2)~\mbox{ and }~\mathbf{B}'\eta_2\in \mathcal{M}^r_\mathcal{S}.
\end{equation}
Note that any linear shape restrictions on $m$ (e.g., those discussed in Section \ref{section:shape_restriction}) can be easily transformed into linear constraints on the unknown parameters in $\eta$.\footnote{For a known matrix $\mathbf{C}$, we can use $\mathbf{C}\eta\leq 0$ to describe the linear shape restrictions on $m$. For example, if there are no covariates and we impose the restriction that $\mathbb{E}[m_1(V)-m_0(V)]\geq0$, then we can set $\mathbf{C}=[0,\mathbb{E}[\mathbf{B}'_0],-\mathbb{E}[\mathbf{B}'_1],\mathbf{0}'_{K_D+K_{0D}+K_{1D}}]$, where $\mathbf{0}_l$ denotes a $l\times 1$ vector of zeros.} Thus, we do not consider shape restrictions in this section for notation simplicity.

Depending on the support of $Z$, the above finite-dimensional approximation of the unknown function $m$ may lead to narrower PRTE bounds compared to the nonparametric bounds $[\underline\beta^\star,\overline\beta^\star]$ in \eqref{opt_relax}. In other words, $[\underline\beta^\star_a,\overline\beta^\star_a]$ may be a proper subset of  $[\underline\beta^\star,\overline\beta^\star]$ or even a subset of $\{\Gamma^\star(m):~m\in\mathcal{M_S}\}$. 
Below, we show that when $Z$ has a finite support, computing the bounds using finite-dimensional constant splines does not alter the nonparametric bounds. Without loss of generality, assume $X\in\{x_1,...,x_{K_X}\}$ and $Z\in\{z_1,...,z_{K_Z}\}$. For any user-specified partition $\mathcal{V}=\bigcup_{k=1}^{K_V}\mathcal{V}_k$ and $d\in\{0,1\}$, let us consider basis functions $\mathbf{B}_d=\{b_{d,kj}\}$ and $\mathbf{B}_D=\mathbf{B}_{dD}=\{b_{D,kl}\}$, where
\begin{align}\label{constant_splines_B}
b_{d,kj}=1[v\in\mathcal{V}_k,x=x_j],\text{ and }b_{D,kl}=1[v\in\mathcal{V}_k,z=z_l].\end{align}
We refer to any finite-dimensional approximation $\mathbf{B}'\eta_2$ as constant splines if $\mathbf{B}$ are constructed using basis functions in \eqref{constant_splines_B}. For any $m\in\mathcal{M}$, one special constant spline that provides the best mean squared error
approximation for $m$, denoted by $\Lambda m=(\Lambda m_0,\Lambda m_1,\Lambda m_D,\Lambda m_{0D},\Lambda m_{1D})$, is given by
\begin{equation}\label{constant_spline_appx}
\begin{aligned}
\Lambda m_d(v,x)=&\sum_{k=1}^{K_V}\sum_{j=1}^{K_X}\mathbb{E}[m_d(V,X)\mid v\in\mathcal{V}_k,X=x_j]b_{d,kj},\\
\Lambda m_D(v,z)=&\sum_{k=1}^{K_V}\sum_{l=1}^{K_Z}\mathbb{E}[m_D(V,Z)\mid v\in\mathcal{V}_k,Z=z_{l}]b_{D,kl},\\
\Lambda m_{dD}(v,x)=&\sum_{k=1}^{K_V}\sum_{l=1}^{K_Z}\mathbb{E}[m_{dD}(V,Z)\mid v\in\mathcal{V}_k,Z=z_{l}]b_{D,kl},
\end{aligned}
\end{equation}
which corresponds to setting the unknown coefficients in $\eta_2$ to the conditional means of $m_d(V,X)$ given $(V,X)$, and the conditional means of $m_D(V,Z)$ and $m_{dD}(V,Z)$ given $(V,Z)$. The proposition below provides conditions under which the bound $[\underline\beta^\star_a,\overline\beta^\star_a]$, obtained by solving the linear programs defined by \eqref{linear_prom_orig_finite} and \eqref{eq:constrainted_finite_relaxed} over the space of constant splines, is the same as the nonparametric bound $[\underline\beta^\star,\overline\beta^\star]$ defined in \eqref{opt_relax}. This result extends the exact computational approach demonstrated by Proposition 4 of \citet{mogstad/santos/torgovitsky:2017} to more general settings with multidimensional unobserved heterogeneity.
\begin{proposition}\label{prop_finite_dim_approx}
Suppose the support of $Z$ is finite and $\Lambda m\in \mathcal{M}$ for any $m\in\mathcal{M}$. Let $\{\mathcal{V}_k\}_{k=1}^{K_V}$ be a partition of $\mathcal{V}$ such that $\omega^\star_{dd'}(v,z)$ is a constant on $v\in\mathcal{V}_k$ for all $z\in\mathcal{Z}$ and $d,d'=0,1$. Then, we have $\underline\beta^\star=\underline\beta^\star_a$ and $\overline\beta^\star=\overline\beta^\star_a$, where $\underline\beta^\star_a$ and $ \overline\beta^\star_a$ defined in \eqref{linear_prom_orig_finite} and \eqref{eq:constrainted_finite_relaxed} are obtained using finite-dimensional constant spline approximation of $m$. 
\end{proposition}

In general, the above finite-dimensional approximation approach can be applied to compute the bounds as long as we approximate $m$ using a linear-in-parameter specification, regardless of the support of $Z$. For example, if $Z$ is a vector of possibly continuous IVs and covariates, we can employ additive linear structures to approximate the unknown functions: $m_d(v,x)=g_d(v)+x'\eta_{d,X}$, $m_D(v,z)=g_D(v)+z'\eta_{D,Z}$, and $m_{dD}(v,z)=g_{dD}(v)+z'\eta_{dD,Z}$ for $d=0,1$, where $\eta_{d,X}$, $\eta_{D,Z}$, and $\eta_{dD,Z}$ are finite-dimensional vectors of unknown parameters. Then, we further approximate the unknown function $g(v)=(g_0(v),g_1(v),g_D(v),g_{0D}(v),g_{1D}(v))$ using finite-dimensional basis functions.\footnote{However, unlike the case with a finite support of $Z$, if $Z$ is continuous, we may need to increase the dimension of the basis functions as the sample size increases, to ensure that the resulting PRTE bounds converge to the nonparametric bounds $[\underline\beta^\star,\overline\beta^\star]$.}

\subsection{Inference}\label{subsection:inference}

In this section, we denote $\mathcal{S}=\{s_1,...,s_{|\mathcal{S}|}\}$ by assuming its cardinality is finite, i.e., $|\mathcal{S}|<\infty$. For illustrative purposes, we do not consider shape restrictions in the presentation of the inference procedure. Nonetheless, these restrictions can be easily incorporated with additional notations.  Below, we demonstrate that the regularized support function method of \citet{gafarov2019inference} can be directly applied in our analysis to obtain uniformly valid confidence sets for the PRTEs. In this section, we closely follow the notations and assumptions from  \citet{gafarov2019inference}. 

Let us denote by $\Upsilon(\mathbb{P})$ the set of $\eta$ that satisfies all the constraints in \eqref{eq:constrainted_finite_relaxed}: 
$$\Upsilon(\mathbb{P})=\left\{\eta=(\eta_1,\eta_2')'\in\Upsilon:~\eta_1=\Gamma^\star(\mathbf{B}'\eta_2)\text{ and }\mathbf{B}'\eta_2\in \mathcal{M}^r_\mathcal{S}\mbox{ for any given }\mathcal{S}\right\}.$$
Without loss of generality, let us assume that the set $\Upsilon(\mathbb{P})$ is compact. Denote $\mathcal{J}^{eq}$ as a set of $d_{eq}$ equality constraints which include $\eta_1=\Gamma^\star(\mathbf{B}'\eta_2)$ and those given in \eqref{cond:Gamma}. Denote $\mathcal{J}^{ineq}$ as the set of $d_{ineq}$ inequality constraints in \eqref{eq:mccormick_bounds_m} to \eqref{eq:mccormick_relax4_1} that describe the convex relaxation. Since $\Upsilon(\mathbb{P})$ only consists of linear constraints in $\eta$, we can find a random matrix $\mathbf{W}\in\mathbb{R}^{(d_{eq}+d_{ineq})\times (d_\eta+1)}$ that depends only on observed variables, such that $\Upsilon(\mathbb{P})$ is equivalently expressed as a set of $\eta$ that satisfies
\begin{eqnarray}\label{constraints_g}
&&\mathbb{E}[g_j(\mathbf{W},\eta)]=0,~~j\in\mathcal{J}^{eq},\nonumber\\
&&\mathbb{E}[g_j(\mathbf{W},\eta)]\leq0,~~j\in\mathcal{J}^{ineq},\nonumber
\end{eqnarray}
where $g_j(\mathbf{W},\eta)=\sum_{l=1}^{d_\eta}W_{j,l}*\eta_l-W_{j,d_\eta+1}$ for all $j\in\mathcal{J}^{eq}\cup\mathcal{J}^{ineq}$, and $W_{j,l}$ denotes the $(j,l)$-th element in $\mathbf{W}$. The detailed expression of $\mathbf{W}$ is given in Appendix \ref{appendix_W}.
Let us rewrite $$\mathbf{W}=\begin{bmatrix}\mathbf{W}_1^{eq}&\mathbf{W}_2^{eq}\\\mathbf{W}^{ineq}_1&\mathbf{W}^{ineq}_2\end{bmatrix},$$ 
where the four blocks $\mathbf{W}^{eq}_1\in\mathbb{R}^{d_{eq}\times d_\eta}$, $\mathbf{W}^{eq}_2\in\mathbb{R}^{d_{eq}\times 1}$, $\mathbf{W}^{ineq}_1\in\mathbb{R}^{d_{ineq}\times d_\eta}$, and $\mathbf{W}^{ineq}_2\in\mathbb{R}^{d_{ineq}\times 1}$ are compatible with the dimensions of $\mathcal{J}^{eq}$ and $\mathcal{J}^{ineq}$.
Denote $A_{\mathbb{P}}=\mathbb{E}\begin{bmatrix}\mathbf{W}^{eq}_1\\\mathbf{W}^{ineq}_1\end{bmatrix}$ and $\pi_{\mathbb{P}}=\mathbb{E}\begin{bmatrix}\mathbf{W}^{eq}_2\\\mathbf{W}^{ineq}_2\end{bmatrix}$.
Then the linear programs in \eqref{linear_prom_orig_finite} and \eqref{eq:constrainted_finite_relaxed} are equivalent to
\begin{eqnarray}\label{test_stat}
\underline{\beta}_a^\star&=&\min_{\eta\in\Upsilon}e_1'\eta~~~~\text{and}~~~~\overline{\beta}_a^\star~=~\max_{\eta\in\Upsilon}e_1'\eta\\
\text{s.t.}&&~~~~e_j'A_{\mathbb{P}}\eta=e_j'\pi_{\mathbb{P}},~~j\in\mathcal{J}^{eq}\nonumber\\
&&~~~~e_j'A_{\mathbb{P}}\eta\leq e_j'\pi_{\mathbb{P}},~~j\in\mathcal{J}^{ineq},\nonumber
\end{eqnarray}
where $e_j$ is a vector whose $j$-th entry is one and all other entries are zero. The expression of the linear programs in \eqref{test_stat} is the same as that in \citet{gafarov2019inference}. Therefore, we can directly apply his method of the regularized support function to construct uniformly valid confidence intervals for $\left[\underline{\beta}_a^\star,\overline{\beta}_a^\star\right]$. The idea of the method is to add a small regularization term to the objective functions
and consider the following regularized programs which are strictly convex:
\begin{align}\label{regularized_primal}
\underline{\beta}_a^\star(\mu_n)=\min_{\eta\in\Upsilon(\mathbb{P})}\left(e_1'\eta+\mu_n\|\eta\|^2\right)~~\text{and}~~\overline{\beta}_a^\star(\mu_n)=-\min_{\eta\in\Upsilon(\mathbb{P})}\left(-e_1'\eta+\mu_n\|\eta\|^2\right),
\end{align}
where $\|\cdot\|$ denotes the Euclidian norm, and $\mu_n>0$ with $\mu_n\rightarrow0$ as sample size $n\rightarrow\infty$ is a tuning parameter chosen by the researcher.\footnote{Following \citet{gafarov2019inference}, we use the tuning parameter $\overline{\mu}_n=\sqrt{\frac{\log(\log(n))}{n}}\hat{\overline{\mu}}_1\text{ and }\underline{\mu}_n=\sqrt{\frac{\log(\log(n))}{n}}\hat{\underline{\mu}}_1,$ where   $\hat{\overline{\mu}}_1=\frac{(tr\{\hat{var}[\hat{\overline{\lambda}}_n(0)'\mathbf{W}_i]\})^{1/2}}{d_\eta+1}$ and $\hat{\underline{\mu}}_1=\frac{(tr\{\hat{var}[\hat{\underline{\lambda}}_n(0)'\mathbf{W}_i]\})^{1/2}}{d_\eta+1}$, and $\hat{\overline{\lambda}}_n(0)$ and $\hat{\underline{\lambda}}_n(0)$ are the Lagrange multiplier of the nonregularized maximization and minimization problem, respectively.} 

\citet{gafarov2019inference} implements the inference procedure in three steps. 

\textbf{Step 1}.  Construct the sample analog of the regularized programs. Denote $\mathbf{W}_i=\begin{bmatrix}\mathbf{W}_{1i}^{eq}&\mathbf{W}_{2i}^{eq}\\\mathbf{W}^{ineq}_{1i}&\mathbf{W}^{ineq}_{2i}\end{bmatrix}$ as a sample observation of $\mathbf{W}$ for $i=1,...,n$. Define $\hat{A}_{n}=\frac{1}{n}\sum_{i=1}^{n}\begin{bmatrix}\mathbf{W}^{eq}_{1i}\\\mathbf{W}^{ineq}_{1i}\end{bmatrix}$ and $\hat{\pi}_{n}=\frac{1}{n}\sum_{i=1}^{n}\begin{bmatrix}\mathbf{W}^{eq}_{2i}\\\mathbf{W}^{ineq}_{2i}\end{bmatrix}$.
The sample analog of \eqref{regularized_primal} with some sequence $\mu_n$ is given by
\begin{align}\label{regularized_primal_sample}
\hat{\underline{\beta}}_a^\star(\mu_n)=\min_{\eta\in\Upsilon}&\left(e_1'\eta+\mu_n\|\eta\|^2\right)~~\text{and}~~\hat{\overline{\beta}}_a^\star(\mu_n)=-\min_{\eta\in\Upsilon}\left(-e_1'\eta+\mu_n\|\eta\|^2\right)\\
\text{s.t. }&~~~~~~~~~~~e_j'\hat{A}_{n}\eta=e_j'\hat{\pi}_{n},~~j\in\mathcal{J}^{eq}\nonumber\\
&~~~~~~~~~~~e_j'\hat{A}_{n}\eta\leq e_j'\hat{\pi}_{n},~~j\in\mathcal{J}^{ineq}.\nonumber
\end{align}

\textbf{Step 2}. Estimate the asymptotic variance of the regularized support function estimators, $\hat{\underline{\beta}}_a^\star(\mu_n)$ and $\hat{\overline{\beta}}_a^\star(\mu_n)$, denoted by $
\hat{\underline{\sigma}}^2_n(\mu_n)$ and $\hat{\overline{\sigma}}_n^2(\mu_n)$, respectively. For any given $\mu_n$, denote $\hat{\underline{\eta}}^\star_n(\mu_n)$ and $\hat{\overline{\eta}}^\star_n(\mu_n)$ as the solutions to the  regularized programs in \eqref{regularized_primal_sample}. Then, the asymptotic variance can be computed as below:
\begin{align}\label{def_asym_variance_beta}
\hat{\underline{\sigma}}^2_n(\mu_n)=\frac{1}{n}\sum_{i=1}^n\left[\hat{\underline{\lambda}}_n'(\mu_n)g\left(\mathbf{W}_i,\hat{\underline{\eta}}^\star_n(\mu_n)\right)\right]^2\text{ and }\hat{\overline{\sigma}}_n^2(\mu_n)=\frac{1}{n}\sum_{i=1}^n\left[\hat{\overline{\lambda}}_n'(\mu_n)g\left(\mathbf{W}_i,\hat{\overline{\eta}}^\star_n(\mu_n)\right)\right]^2,
\end{align}
where $\hat{\underline{\lambda}}_n(\mu_n)$ and $\hat{\overline{\lambda}}_n(\mu_n)$ are the estimated vectors of Lagrange multipliers of the programs in \eqref{regularized_primal_sample}. 

\textbf{Step 3}.  Compute the bias correction term, the outer bound estimator, and the uniformly valid confidence set. Recall that compared to $\underline{\beta}_a^\star$ and $\overline{\beta}_a^\star$, the regularization term $\mu_n\|\eta\|^2$ introduces an upward and a downward bias in $\hat{\underline{\beta}}_a^\star(\mu_n)$ and $\hat{\overline{\beta}}_a^\star(\mu_n)$, respectively. Below, we introduce two bias correction terms $\mu_n\|\hat{\underline{\eta}}^{out}_n\|^2$ and $\mu_n\|\hat{\overline{\eta}}^{out}_n\|^2$ in Step 3.1, so that the outer bound estimators defined in Step 3.2 are asymptotically unbiased or biased towards the desirable directions.

\begin{itemize}
\item \textbf{Step 3.1}.  Let $\hat{\Upsilon}(\mathbb{P})$ be the sample analog of set $\Upsilon(\mathbb{P})$ defined using $\hat{A}_n$ and $\hat\pi_n$. 
Define $$\hat{\underline{\eta}}^{out}_n=(\hat{\underline{\eta}}^{out}_{1,n},...,\hat{\underline{\eta}}^{out}_{d_\eta,n})'\text{ and }\hat{\overline{\eta}}^{out}_n=(\hat{\overline{\eta}}^{out}_{1,n},...,\hat{\overline{\eta}}^{out}_{d_\eta,n})',$$ where we denote $\hat{\underline{\eta}}^{out}_{k,n}=\max\{\hat{\underline{\eta}}_{k,n}^+,\hat{\underline{\eta}}_{k,n}^-\}$ and $\hat{\overline{\eta}}^{out}_{k,n}=\max\{\hat{\overline{\eta}}_{k,n}^+,\hat{\overline{\eta}}_{k,n}^-\}$ for $k=1,...,d_\eta$, with
\begin{align*}
&\hat{\underline{\eta}}_{k,n}^+=\Big|\min_{\eta\in\hat{\Upsilon}(\mathbb{P}),\eta_{1}\leq \hat{\underline{\beta}}_a^{\star}+\mu_n}\{\eta_{k}\}\Big|,~~\hat{\underline{\eta}}_{k,n}^-=\Big|\min_{\eta\in\hat{\Upsilon}(\mathbb{P}),\eta_{1}\leq \hat{\underline{\beta}}_a^{\star}+\mu_n}\{-\eta_{k}\}\Big|,\\
&\hat{\overline{\eta}}_{k,n}^+=\Big|\min_{\eta\in\hat{\Upsilon}(\mathbb{P}),\eta_{1}\geq \hat{\overline{\beta}}_a^{\star}-\mu_n}\{\eta_{k}\}\Big|,~~\hat{\overline{\eta}}_{k,n}^-=\Big|\min_{\eta\in\hat{\Upsilon}(\mathbb{P}),\eta_{1}\geq \hat{\overline{\beta}}_a^{\star}-\mu_n}\{-\eta_{k}\}\Big|,
\end{align*}
where we denote $\hat{\underline{\beta}}_a^\star$ and $\hat{\overline{\beta}}_a^\star$ as the estimators of $\underline{\beta}_a^\star$ and $\overline{\beta}_a^\star$ defined in \eqref{test_stat}. We can see that, with probability approaching one, $\|\hat{\underline{\eta}}^{out}_n\|\geq\|\eta\|$ for any $\eta$ that solves $\min\limits_{\eta\in\hat\Upsilon(\mathbb{P})}e_1'\eta$, and $\|\hat{\overline{\eta}}^{out}_n\|\geq\|\eta\|$ for any $\eta$ that solves $\max\limits_{\eta\in\hat\Upsilon(\mathbb{P})}e_1'\eta$. 
\item \textbf{Step 3.2}. Given the two bias correction terms, $\mu_n\|\hat{\underline{\eta}}^{out}_n\|^2$ and $\mu_n\|\hat{\overline{\eta}}^{out}_n\|^2$, the outer bound estimators can be defined as $\hat{\underline{\beta}}^{out}(\mu_n)=\hat{\underline{\beta}}_a^\star(\mu_n)-\mu_n\|\hat{\underline{\eta}}^{out}_n\|^2$ and  $\hat{\overline{\beta}}^{out}(\mu_n)=\hat{\overline{\beta}}_a^\star(\mu_n)+\mu_n\|\hat{\overline{\eta}}^{out}_n\|^2$.
\item \textbf{Step 3.3}. Let $c_{1-\alpha}$ be the $1-\alpha$ quantile of $N(0,1)$. Then, a $(1-\alpha)$-confidence interval for $\left[\underline{\beta}_a^\star,\overline{\beta}_a^\star\right]$ can be defined as  
\begin{align}\label{def_conf_interval}\left[\hat{\underline{\beta}}^{out}(\mu_n)-c_{1-\frac{\alpha}{2}}\frac{\underline{\hat{\sigma}}^+_n}{\sqrt{n}},~~\hat{\overline{\beta}}^{out}(\mu_n)+c_{1-\frac{\alpha}{2}}\frac{\hat{\overline{\sigma}}^+_n}{\sqrt{n}}\right],
\end{align}
where $\underline{\hat{\sigma}}^+_n=\hat{\underline{\sigma}}_n(\mu_n)$ and $\hat{\overline{\sigma}}^+_n=\hat{\overline{\sigma}}_n(\mu_n)$ denote the standard deviation of the  asymptotic variance given in \eqref{def_asym_variance_beta}.
\end{itemize}

The following assumptions are employed by \citet{gafarov2019inference} to show the asymptotic properties of the confidence interval in \eqref{def_conf_interval}. We present them here for completeness.

\begin{assumption}\label{assumption_sample}~
\begin{itemize}
\item[(i)] $\Upsilon(\mathbb{P})$ is non-empty and compact. 
\item[(ii)] The IV-like functions in $\mathcal{S}$ are chosen so that $\min\limits_{\mathcal{J}\subseteq\mathcal{J}^{ineq},|\mathcal{J}|=d_\eta-d_{eq}}\kappa_1(\mathbf{J}(A_{\mathbb{P}},\pi_{\mathbb{P}}))>0$ and 
$\min\limits_{\mathcal{J}\subseteq\mathcal{J}^{ineq},|\mathcal{J}|=d_\eta-d_{eq}+1,~\eta\in\Upsilon(\mathbb{P})}\|\mathbf{J}(A_{\mathbb{P}}\eta-\pi_{\mathbb{P}})\|>0$, where $\mathbf{J}=(e_{j_1},...,e_{j_k})'$ is a selection matrix that selects active constraints $\mathcal{J}\cup \mathcal{J}^{eq}=\{j_1,...,j_k\}$ for some subset $\mathcal{J}\subseteq\mathcal{J}^{ineq}$, and $\kappa_1(B)$ is the smallest left-singular-value of a matrix $B$.
\item[(iii)] We observe i.i.d. samples with bounded higher-order moments: $\{\mathbf{W}_i\in\mathbb{R}^{(d_{eq}+d_{ineq})\times (d_\eta+1)},~i=1,...,n\}$ is an i.i.d. sample of $\mathbf{W}$ and 
there exists $\epsilon>0$ and $\overline{M}<\infty$ such that $\mathbb{E}[\|\mathbf{W}_i\|^{2+\epsilon}]<\overline{M}$.
\end{itemize}
\end{assumption}


Let $\mathcal{P}^\star$ be a collection of $P$ that satisfies Assumptions \ref{assumption_DGP} and \ref{assumption_sample}. The theorem below demonstrates that the $(1-\alpha)$-confidence interval in \eqref{def_conf_interval} covers $\left[\underline{\beta}_a^\star,\overline{\beta}_a^\star\right]$ uniformly across $\mathcal{P}^\star$ with probability at least $1-\alpha$.  

\begin{theorem}[\citealp{gafarov2019inference}]\label{inference_gafarov} Under Assumptions \ref{assumption_DGP} and \ref{assumption_sample}, for $0<\alpha<1/2$, if $\mu_n\rightarrow0$ and $\mu_n\sqrt{n}\rightarrow\infty$, then we have $$\liminf_{n\rightarrow\infty}\inf_{P\in\mathcal{P}^\star}P\left(\left[\underline{\beta}_a^\star,\overline{\beta}_a^\star\right]\subset \left[\hat{\underline{\beta}}^{out}(\mu_n)-c_{1-\frac{\alpha}{2}}\frac{\underline{\hat{\sigma}}^+_n}{\sqrt{n}},~~\hat{\overline{\beta}}^{out}(\mu_n)+c_{1-\frac{\alpha}{2}}\frac{\hat{\overline{\sigma}}^+_n}{\sqrt{n}}\right]\right)\geq 1-\alpha.$$
\end{theorem}

\section{Numerical Illustration and Monte Carlo Simulation}\label{section:numerical_simulation}


In this section, we numerically illustrate the performance of our proposed convex-relaxation bounds and verify the asymptotic properties in finite sample using Monte Carlo simulations. We omit the covariates $X$ to ease the computation.

\subsection{Data Generating Process and Target Parameters}\label{section:num_DGP}
Consider the data generating process (DGP) for a binary outcome variable:
$$Y_1=1[m_1(V)-\varepsilon_1>0]\mbox{ and }Y_0=1[m_0(V)-\varepsilon_0>0],$$
where $(\varepsilon_0,\varepsilon_1)\sim \mathrm{Uniform}([0,1]^2)$. We consider two different dimensions of the unobserved heterogeneity: (i) $V\sim \mathrm{Uniform}[0,1]$; and (ii) $V=(V_1,V_2)\sim \mathrm{Uniform}([0,1]^2)$. The specifications for $(m_0,m_1)$ are summarized in Table \ref{DGP_m0m1}.

\begin{table}[h!]
\centering
\caption{DGP Designs for Outcome Variable}\label{DGP_m0m1}
\adjustbox{max width=\textwidth}{\begin{tabular}{lllcccc}
\toprule
\multirow{4}{*}{\;\;$V\sim \mathrm{Uniform}[0,1]$} & $m_0(v)=\mathbf{Ber}(v)'\theta_0,$&\;\;\;\;$\theta_0=(0.6,0.4,0.3)'$ \smallskip\\
&$m_1(v)=\mathbf{Ber}(v)'\theta_1,$&\;\;\;\;$\theta_1=(0.75,0.5,0.3)'$\medskip\\
&\multicolumn{2}{l}{$\mathbf{Ber}(v)=(ber_0^{2}(v),ber_1^{2}(v),ber_2^{2}(v))'$ }\smallskip\\
&\multicolumn{2}{l}{\hspace{2cm}with $ber_k^{p}(v)$ the $k$-th Bernstein polynomial of degree $p$}\smallskip\\
\midrule
\multirow{5}{*}{\minitab[l]{$V=(V_1,V_2)$\\\;\;\;\;$\sim \mathrm{Uniform}([0,1]^2$)}} &$m_0(v_1,v_2)=\mathbf{Ber}(v_1,v_2)'\theta_0,$&\;\;\;\;$\theta_0=(0.7,0.5,0.5,0.3,0.2,0.1,0.1,0,0)'$\smallskip\\
&$m_1(v_1,v_2)=\mathbf{Ber}(v_1,v_2)'\theta_1,$&\;\;\;\;$\theta_1=(0.85,0.65,0.5,0.5,0.45,0.3,0.2,0.1,0.1)'$\medskip\\
&\multicolumn{2}{l}{$\mathbf{Ber}(v_1,v_2)=[ber_{00}^{2}(v_1,v_2),...,ber_{02}^{2}(v_1,v_2),...,ber_{20}^{2}(v_1,v_2),...,ber_{22}^{2}(v_1,v_2)]'$}\smallskip\\
&\multicolumn{2}{l}{\hspace{2.5cm}with
$ber_{kk'}^{p}(v_1,v_2)=ber_{k}^{p}(v_1)ber_{k'}^{p}(v_2)$}\smallskip\\
&\multicolumn{2}{l}{\hspace{2.5cm} and $ber_{k}^{p}(v)$ the $k$-th Bernstein polynomial of degree $p$}\smallskip\\
\bottomrule
\end{tabular}}
\end{table}

Consider two mutually independent instruments $Z=(Z_1,Z_2)\in\mathcal{Z}$, where $\mathcal{Z}=\{0,1,2\}\times\{0.5,1\}$ with $\mathbb{P}(Z_1=0)=0.5$ and $  \mathbb{P}(Z_1=1)=0.4$, and $\mathbb{P}(Z_2=0.5)=0.7$. We generate the treatment variable using a random coefficient model:
\begin{align}\label{simulation_DGP_treatment_random}
D=1[Q(Z,U,V)\geq 0],
\end{align}
where $U\sim N(0,var(U))$ with $sd(U)\in\{0.1,0.5,0.9\}$. In addition, $Z$, $(\varepsilon_0,\varepsilon_1)$, $U$ and $V$ are mutually independent. Let
\begin{align*}
\begin{cases}
Q(z,u,v)=0.2z_1+uz_2-v, & \mbox{if $v$ is  single-dimensional } \\
Q(z,u,v)=(0.6-v_1)z_1+uz_2-0.5v_2, & \mbox{if  $v=(v_1,v_2)$ is two-dimensional}.
\end{cases}
\end{align*}
In the numerical illustration and the simulation studies, we consider three target parameters: (i) ATE$=\mathbb{E}[Y_1-Y_0]$, (ii) ATT$=\mathbb{E}[Y_1-Y_0\mid D=1]$, and (iii) the average selection bias ASB$=\mathbb{E}[Y_0\mid D=1]-\mathbb{E}[Y_0\mid D=0]$.


\subsection{Numerical Illustration}\label{section:numerical}
In this section, we numerically illustrate the performance of our convex-relaxation bounds, starting from the cases without any shape restrictions. 
For all three target parameters, we compare our method with  three existing bounding methods:
\begin{itemize}
\item[] ``\textbf{Manski}'' -- the nonparametric bounds with instruments \citep{manski1990nonparametric};
\item[] ``\textbf{HV}'' -- the nonparametric bounds with instruments assuming the threshold-crossing structure in the treatment variable \citep{heckman/vytlacil:2001};
\item[]``\textbf{MST}'' -- the bounds assuming the threshold-crossing structure in the treatment variable, obtained via a linear programming method \citep*{mogstad/santos/torgovitsky:2017};
\item[]
``\textbf{CvR}'' -- the bounds proposed in this paper, obtained by applying convex relaxation and a linear programming method.
\end{itemize} 
Specifically, both Manski and HV derive analytical bounds for $\mathbb{E}[Y_0]$ and $\mathbb{E}[Y_1]$, which only depend on the distributions of observed variables, $(Y,D,Z)$. With these bounds, we can compute the Manski and HV bounds for the three target parameters, as all of them can be expressed as some linear combinations of $\mathbb{E}[Y_0]$ and $\mathbb{E}[Y_1]$:
\begin{align*}
    ATE=&\mathbb{E}[Y_1]-\mathbb{E}[Y_0],\\
    ATT=&\mathbb{E}[Y_1-Y_0\mid D=1]=\mathbb{E}[Y\mid D=1] -(\mathbb{E}[Y_0]-\mathbb{E}[Y(1-D)])/\mathbb{P}(D=1),\\
    ASB=&\mathbb{E}[Y_0\mid D=1]-\mathbb{E}[Y_0\mid D=0]=(\mathbb{E}[Y_0]-\mathbb{E}[Y(1-D)])/\mathbb{P}(D=1)-\mathbb{E}[Y\mid D=0].
\end{align*}
On the other hand, MST and CvR bounds are both computed using linear programming methods. In particular, we calculate MST bounds using all the IV-like functions in $\tilde{\mathcal{S}}=\mathbf{L}^2(\{0,1\}\times\mathcal{Z},\mathbb{R})$.

For CvR bounds, we implement the linear programming computation outlined in Section \ref{section:computation}, and we set $m_d^L(v,x)=m_D^L(v,z)=0$ and $m_d^U(v,x)=m_D^U(v,z)=1$ for the McCormick convex relaxation. 
We adopt the IV-like functions in $\mathcal{S}=\mathbf{L}^2(\mathcal{Z},\mathbb{R})$,  which exhausts all available information from the observed data.\footnote{In the data generating process of this section, the support $\mathcal{Z}$ is finite, so we only need to use the indicator functions for each point in $\mathcal{Z}$ as the IV-like function.}
Because $\omega^\star_{dd'}(v,z)$, $m_d^L(v,x)$, $m_d^U(v,x)$, $m_D^L(v,z)$ and $m_D^U(v,z)$ all do not depend on $v$ for all three target parameters, according to Proposition \ref{prop_reduce_dim}, we do not need to know the true dimension of $V$ and only need to bound the integral of the function $m$ in Eq.\eqref{prop_reduce_dim_fun},  instead of the function $m$ itself. 
Thus, for both the single-dimensional and two-dimensional $V$ cases,  we use  $\mathbf{B}_d=1$ for $\int_{\mathcal{V}}m_d(v)dv$ with $d=0,1$, and $\mathbf{B}_D=\mathbf{B}_{dD}=\{1[z=l]\}_{l\in\mathcal{Z}}$ for $\int_{\mathcal{V}}m_D(v,z)dv$ and $\int_{\mathcal{V}}m_{dD}(v,z)dv$.

\bigskip

\noindent\textbf{Bounds without Shape Restrictions}. \quad
Numerical results of the four bounding methods are summarized in Table \ref{tab:num_random_coef}.
There are a few points worth mentioning. 
First, we know that the assumptions required by HV bounds, which include the threshold-crossing structure, are stronger than those required by Manski bounds.
However, for ATE, its Manski bounds are narrower than or equal to HV bounds across all DGP designs, as the threshold-crossing structure does not hold. This highlights that imposing more structural assumptions does not always lead to narrower bounds. Whereas for ATT and ABS, Manski bounds and HV bounds coincide with each other. Second, due to model misspecification, there are no solutions to the linear programming problems of the MST method, resulting in empty MST bounds for all three target parameters in all DGP designs.\footnote{The empty MST bounds are an indication of the failure of their assumption (i.e., the threshold-crossing structure). However, it is not vice versa, that is, the MST bounds can be non-empty under other DGPs even if the threshold-crossing structure fails.} 
Third, our CvR bounds for all three target parameters are the same as their Manski bounds across all DGPs designs. These results demonstrate the consequences of model misspecification in the MST method and the informativeness of our CvR bounds when the threshold-crossing structure fails to hold.


\bigskip

\noindent\textbf{Bounds with Shape Restrictions}. \quad Next, we evaluate our CvR method under shape restrictions. First, we consider the restrictions in \eqref{MTR_implied_restrictions}, which are implied by the monotone treatment response (MTR) in Condition \ref{condition:mono_resp}. Define 
\begin{align*}
(MTR)~~\mathcal{R}^1=\Big\{m\in\mathcal{M}:~\mathbb{E}[m_0(V,x)]& \leq \inf_{z_0\in\mathcal{Z}_0}\mathbb{E}[Y\mid X=x,Z_0=z_0],\text{ and }\\
   \sup_{z_0\in\mathcal{Z}_0}&\mathbb{E}[Y\mid X=x,Z_0=z_0]\leq\mathbb{E}[m_1(V,x)]\Big\}.
\end{align*}
Second, we consider the shape restrictions in \eqref{MTS_constraints}, which are implied by the monotone treatment selection (MTS) in Condition \ref{condition:mono_select}. Define
\begin{align*}
(MTS)~~\mathcal{R}^2=\Big\{m\in\mathcal{M}:~\mathbb{E}[Y\mid D=1]\mathbb{P}(D=0)&\geq \mathbb{E}[m_1(V)-m_{1D}(V,Z))],\text{ and }\\
\mathbb{E}[Y\mid D=0]\mathbb{P}(D=1)&\leq \mathbb{E}[m_0(V)-m_{0D}(V,Z)]\Big\}.
\end{align*}
Third, let us consider the shape restriction in \eqref{MSG_constraints}, which is implied by the monotone selection on the gains (MSG) in Condition \ref{condition:mono_select_gain}. Define
\begin{align*}
(MSG)~~\mathcal{R}^3=\Big\{m\in\mathcal{M}:~\mathbb{E}[Y\mid D=1]&+\mathbb{E}[Y\mid D=0]\\
~\geq~&\mathbb{E}\Big[\int_{{\mathcal{V}}}(m_{1}(v)-m_{1D}(v,Z))dv\Big]/\mathbb{P}(D=0)\\
&+\mathbb{E}\Big[\int_{{\mathcal{V}}}(m_{0}(v)-m_{0D}(v,Z))dv\Big]/\mathbb{P}(D=1)\Big\}.
\end{align*}
The last shape restriction is a combination of MTR, MTS, and MSG, defind as $$\mathcal{R}^4=\mathcal{R}^1\cap\mathcal{R}^2\cap\mathcal{R}^3.$$ 
Note that all these four sets of shape restrictions are satisfied by our DGPs.

The resulting CvR bounds for the three target parameters under these shape restrictions are presented in Tables \ref{tab:num_ATE_random_shape}, \ref{tab:num_ATT_random_shape}, and  \ref{tab:num_ABS_random_shape}. Notably, compared to the CvR bounds without restrictions, the restrictions implied by MTR, as in $\mathcal{R}^1$, significantly reduce the CvR bound size for all three target parameters and identify the signs of both ATE and ATT. Furthermore, the restrictions implied by MTS, as in $\mathcal{R}^2$, improve the upper CvR bounds for ATE and ATT and identify the sign of ASB. Compared to MTR and MTS, the restrictions imposed by MSG, as in $\mathcal{R}^3$, yield relatively smaller improvements to the CvR bounds for all three target parameters. Finally, by combining all three shape restrictions,  as in $\mathcal{R}^4$, we obtain the tightest CvR bounds, which coincide with the intersection of the bounds under $\mathcal{R}^1$ to $\mathcal{R}^3$.

\begin{table}[htbp!]
\begin{center}
\caption{Comparison of Bounds without Shape Restrictions}\label{tab:num_random_coef}
\subcaption{Target parameter: ATE}
\adjustbox{max width=\textwidth}{
\begin{tabular}{lcccccccc}
\toprule
&\multicolumn{3}{c}{\textbf{Single-dimensional $V$}}&&\multicolumn{3}{c}{\textbf{Two-dimensional $V$}}\\
\cline{2-4}\cline{6-8}
&$sd(U)=0.1$&$sd(U)=0.5$&$sd(U)=0.9$&&$sd(U)=0.1$&$sd(U)=0.5$&$sd(U)=0.9$\\
True value& 0.083&0.083&0.083&&0.139&0.139&0.139\\
\toprule
Manski &[-0.181, 0.437]&[-0.228, 0.439] &[-0.269, 0.451]&&[-0.031, 0.534]&[-0.145, 0.564]& [-0.231, 0.570]\\
HV&[-0.183, 0.437] &[-0.229, 0.439]&[-0.269, 0.451]&&[-0.031, 0.534] &[-0.155, 0.564]&[-0.247, 0.570]\\
MST &$\emptyset$&$\emptyset$&$\emptyset$&& $\emptyset$& $\emptyset$& $\emptyset$\\
CvR& [-0.181, 0.437]&[-0.228, 0.439]&[-0.269, 0.451]&& [-0.031, 0.534] &[-0.145, 0.564]&[-0.231, 0.570]\\
\bottomrule
\end{tabular}}
\medskip
\subcaption{Target parameter: ATT}
\adjustbox{max width=\textwidth}{
\begin{tabular}{lcccccccc}
\toprule
&\multicolumn{3}{c}{\textbf{Single-dimensional $V$}}&&\multicolumn{3}{c}{\textbf{Two-dimensional $V$}}\\
\cline{2-4}\cline{6-8}
&$sd(U)=0.1$&$sd(U)=0.5$&$sd(U)=0.9$&&$sd(U)=0.1$&$sd(U)=0.5$&$sd(U)=0.9$\\
True value&0.135&0.120&0.107&&0.143&0.144&0.142\\
\toprule
Manski & [0.073, 0.223]& [-0.102, 0.381]&[-0.212, 0.429]&&[0.021, 0.208]& [-0.291,    0.332]&  [-0.429, 0.369]\\
HV & [0.073, 0.223]& [-0.102, 0.381]&[-0.212, 0.429]&&[0.021, 0.208]& [-0.291,    0.332]&  [-0.429, 0.369]\\
MST & $\emptyset$ & $\emptyset$ & $\emptyset$ &&$\emptyset$ & $\emptyset$ & $\emptyset$\\
CvR& [0.073, 0.223]& [-0.102, 0.381]&[-0.212, 0.429]&&[0.021, 0.208]& [-0.291,    0.332]&  [-0.429, 0.369]\\
\bottomrule
\end{tabular}}
\medskip
\subcaption{Target parameter: ABS}
\adjustbox{max width=\textwidth}{
\begin{tabular}{lcccccccc}
\toprule
&\multicolumn{3}{c}{\textbf{Single-dimensional $V$}}&&\multicolumn{3}{c}{\textbf{Two-dimensional $V$}}\\
\cline{2-4}\cline{6-8}
&$sd(U)=0.1$&$sd(U)=0.5$&$sd(U)=0.9$&&$sd(U)=0.1$&$sd(U)=0.5$&$sd(U)=0.9$\\
True value&0.134&0.097&0.067&&0.203&0.130&0.087\\
\toprule
Manski & [0.046, 0.196]& [-0.165, 0.319]& [-0.256, 0.385]&& [0.138, 0.324]&[-0.059,    0.564]&[-0.140, 0.659]\\
HV & [0.046, 0.196]& [-0.165, 0.319]& [-0.256, 0.385]&&[0.138, 0.324]&[-0.059,    0.564]&[-0.140, 0.659]\\
MST & $\emptyset$ & $\emptyset$ & $\emptyset$ &&$\emptyset$ & $\emptyset$ & $\emptyset$\\
CvR&[0.046, 0.196]& [-0.165, 0.319]& [-0.256, 0.385]&&[0.138, 0.324]&[-0.059,    0.564]&[-0.140, 0.659]\\
\bottomrule
\end{tabular}}
\end{center}
\footnotesize
Note: This table compares four bounds (Manski, HV, MST, and CvR) for three target parameters: ATE, ATT, and ABS. In the true DGP, the treatment variable is generated by a random coefficient model with either a single- or two-dimensional unobserved heterogeneity $V$. The `True value' in the second row represents the value of each target parameter under the true DGP. The four bounds are calculated as described in Section \ref{section:numerical}. 
$\emptyset$ represents an empty set.
\end{table}

\begin{table}[htbp!]
\begin{center}
\caption{CvR Bounds with Shape Restrictions for ATE}\label{tab:num_ATE_random_shape}
\subcaption{\textbf{Single-dimensional $V$}}
\adjustbox{max width=\textwidth}{\begin{tabular}{cccccccccc}
\toprule
&  &  &$\mathcal{R}^1$&$\mathcal{R}^2$&$\mathcal{R}^3$&$\mathcal{R}^4$\\
CvR&True value&No restriction&  (MTR)& (MTS)& (MSG)&($\mathcal{R}^1\cap\mathcal{R}^2\cap\mathcal{R}^3$)\\
\toprule
$sd(U)= 0.1$&0.083 &[-0.181, 0.437]&[0.048, 0.437]&[-0.181, 0.263]&[-0.181, 0.223]&[0.048, 0.223]\\
$sd(U)= 0.5$&0.083&[-0.228, 0.439] & [0.034, 0.439]&[-0.228, 0.217]&    [-0.228, 0.381]&[0.034, 0.217]\\
$sd(U)= 0.9$&0.083&[-0.269, 0.451]&[0.025, 0.451]&[-0.269, 0.174]&  [-0.269, 0.429]&[0.025, 0.174]\\
\bottomrule
\end{tabular}
}
\medskip
\subcaption{\textbf{Two-dimensional $V$}}
\adjustbox{max width=\textwidth}{\begin{tabular}{cccccccccc}
\toprule
&  &  &$\mathcal{R}^1$&$\mathcal{R}^2$&$\mathcal{R}^3$&$\mathcal{R}^4$\\
CvR&True value&No restriction&  (MTR)& (MTS)& (MSG)&($\mathcal{R}^1\cap\mathcal{R}^2\cap\mathcal{R}^3$)\\
\toprule
$sd(U)= 0.1$&0.139 &[-0.031, 0.534]&[0.060, 0.534]&[-0.031, 0.316]&[-0.031, 0.208]& [0.060, 0.208]\\
$sd(U)= 0.5$& 0.139&[-0.145, 0.564]& [0.038, 0.564]&[-0.145, 0.273]&[-0.078, 0.332]& [0.038, 0.273]\\
$sd(U)= 0.9$&0.139&[-0.231, 0.570]&[0.025, 0.570]& [-0.231, 0.230]&[-0.112, 0.369]&[0.025, 0.230]\\
\bottomrule
\end{tabular}
}
\end{center}
\footnotesize
Note: This table displays the CvR bounds for ATE under four shape restrictions. Column “True value” presents the ATE value under the true DGP. Column “No restrictions” gives the CvR bounds without any shape restrictions. Columns $\mathcal{R}^1$ to $\mathcal{R}^4$ list the CvR bounds under the corresponding shape restrictions described in Section \ref{section:numerical}.
\end{table}

\begin{table}[htbp!]
\begin{center}
\caption{CvR Bounds with Shape Restrictions for ATT}\label{tab:num_ATT_random_shape}
\subcaption{\textbf{Single-dimensional $V$}}
\adjustbox{max width=\textwidth}{\begin{tabular}{ccccccccc}
\toprule
&  &  &$\mathcal{R}^1$&$\mathcal{R}^2$&$\mathcal{R}^3$&$\mathcal{R}^4$\\
CvR&True value&No restriction&  (MTR)& (MTS)& (MSG)&($\mathcal{R}^1\cap\mathcal{R}^2\cap\mathcal{R}^3$)\\
\toprule
$sd(U)= 0.1$& 0.135&[0.073, 0.223]&[0.113, 0.223]&[0.073, 0.223]&[0.073, 0.223]&[0.113, 0.223]\\
$sd(U)= 0.5$&0.120&[-0.102, 0.381] &  [0.056, 0.381]&[-0.102, 0.217]&    [-0.102, 0.381]&[0.056, 0.217]\\
$sd(U)= 0.9$& 0.107&[-0.212, 0.429]&[0.032, 0.429]&[-0.212, 0.174]&  [-0.212, 0.429]&[0.032, 0.174]\\
\bottomrule
\end{tabular}
}
\medskip
\subcaption{\textbf{Two-dimensional $V$}}
\adjustbox{max width=\textwidth}{\begin{tabular}{ccccccccc}
\toprule
&  &  &$\mathcal{R}^1$&$\mathcal{R}^2$&$\mathcal{R}^3$&$\mathcal{R}^4$\\
CvR&True value&No restriction&  (MTR)& (MTS)& (MSG)&($\mathcal{R}^1\cap\mathcal{R}^2\cap\mathcal{R}^3$)\\
\toprule
$sd(U)= 0.1$ &0.143&[0.021, 0.208] &[0.114, 0.208]&[0.021, 0.208]&[0.021, 0.208]&[0.114, 0.208]\\
$sd(U)= 0.5$ & 0.144&[-0.291, 0.332]&[0.051, 0.332]& [-0.291, 0.273]&[-0.078, 0.332]& [0.051, 0.273]\\
$sd(U)= 0.9$&0.142&[-0.429, 0.369]&[0.027, 0.369]& [-0.429, 0.230]&[-0.112, 0.369]& [0.027, 0.230]\\
\bottomrule
\end{tabular}
}
\end{center}
\footnotesize
Note: This table displays the CvR bounds for ATT under four shape restrictions. Column “True value” presents the ATE value under the true DGP. Column “No restrictions” gives the CvR bounds without any shape restrictions. Columns $\mathcal{R}^1$ to $\mathcal{R}^4$ list the CvR bounds under the corresponding shape restrictions described in Section \ref{section:numerical}.
\end{table}

\begin{table}[htbp!]
\begin{center}
\caption{CvR Bounds with Shape Restrictions for ASB}\label{tab:num_ABS_random_shape}
\subcaption{\textbf{Single-dimensional $V$}}
\adjustbox{max width=\textwidth}{\begin{tabular}{ccccccccc}
\toprule
&  &  &$\mathcal{R}^1$&$\mathcal{R}^2$&$\mathcal{R}^3$&$\mathcal{R}^4$\\
CvR&True value&No restriction&  (MTR)& (MTS)& (MSG)&($\mathcal{R}^1\cap\mathcal{R}^2\cap\mathcal{R}^3$)\\
\toprule
$sd(U)= 0.1$& 0.134&[0.046, 0.196]&[0.046, 0.156]&[0.046, 0.196]&[0.046, 0.196]&[0.046, 0.156]\\
$sd(U)= 0.5$& 0.097&[-0.165, 0.319]& [-0.165, 0.161]&[0.000, 0.319]&    [-0.165, 0.319]&[0.000, 0.161]\\
$sd(U)= 0.9$& 0.067&[-0.256, 0.385]&[-0.256, 0.142]&[0.000, 0.385]&  [-0.256, 0.385]&[0.000, 0.142]\\
\bottomrule
\end{tabular}
}
\medskip
\subcaption{\textbf{Two-dimensional $V$}}
\adjustbox{max width=\textwidth}{\begin{tabular}{ccccccccc}
\toprule
&  &  &$\mathcal{R}^1$&$\mathcal{R}^2$&$\mathcal{R}^3$&$\mathcal{R}^4$\\
CvR&True value&No restriction&  (MTR)& (MTS)& (MSG)&($\mathcal{R}^1\cap\mathcal{R}^2\cap\mathcal{R}^3$)\\
\toprule
$sd(U)= 0.1$& 0.203& [0.138, 0.324]&[0.138, 0.231]&[0.138, 0.324]&[0.138, 0.324]&[0.138, 0.231]\\
$sd(U)= 0.5$& 0.130& [-0.059, 0.564]& [-0.059, 0.222]&[0.000, 0.564]& [-0.059, 0.351]&[0.000, 0.222]\\
$sd(U)= 0.9$& 0.087& [-0.140, 0.659]&[-0.140, 0.202]& [0.000, 0.659]& [-0.140, 0.341]&[0.000, 0.202]\\
\bottomrule
\end{tabular}
}
\end{center}
\footnotesize
Note: This table displays the CvR bounds for ASB under four shape restrictions. Column “True value” presents the ATE value under the true DGP. Column “No restrictions” gives the CvR bounds without any shape restrictions. Columns $\mathcal{R}^1$ to $\mathcal{R}^4$ list the CvR bounds under the corresponding shape restrictions described in Section \ref{section:numerical}.
\end{table}

\newpage
\subsection{Monte Carlo Simulations}

To check the finite sample performance of the CvR bounds, we conduct Monte Carlo simulations following the \citet{gafarov2019inference} inference method described in Sections \ref{subsection:inference}. The DGPs, target parameters, IV-like functions, and known boundaries $m^L_d$, $m^U_d$, $m^L_D$, and $m^U_D$ in the convex relaxation are the same as those used to generate the numerical results summarized in Table \ref{tab:num_random_coef}. In this section, we do not consider shape restrictions. 
We present the coverage rates of the 95\% confidence interval for ATE, ATT, and ASB in Figures \ref{fig:random_coverage_ATE}, \ref{fig:random_coverage_ATT}, and \ref{fig:random_coverage_ABS}, respectively. These results demonstrate that the asymptotic properties of the inference method, as described in Theorem \ref{inference_gafarov}, hold true in finite samples. Importantly, the confidence intervals exhibit reasonable width and cover the true CvR bounds with a probability of at least 95\%. As the sample size increases, the confidence intervals become narrower, and they exclude values outside the true CvR bounds with an increasing probability.

\begin{figure}[htbp!]
\caption{Coverage Rates of the 95\% Confidence Interval for ATE}\label{fig:random_coverage_ATE}
\subcaption{Single-dimensional $V$}
\includegraphics[width=\linewidth]{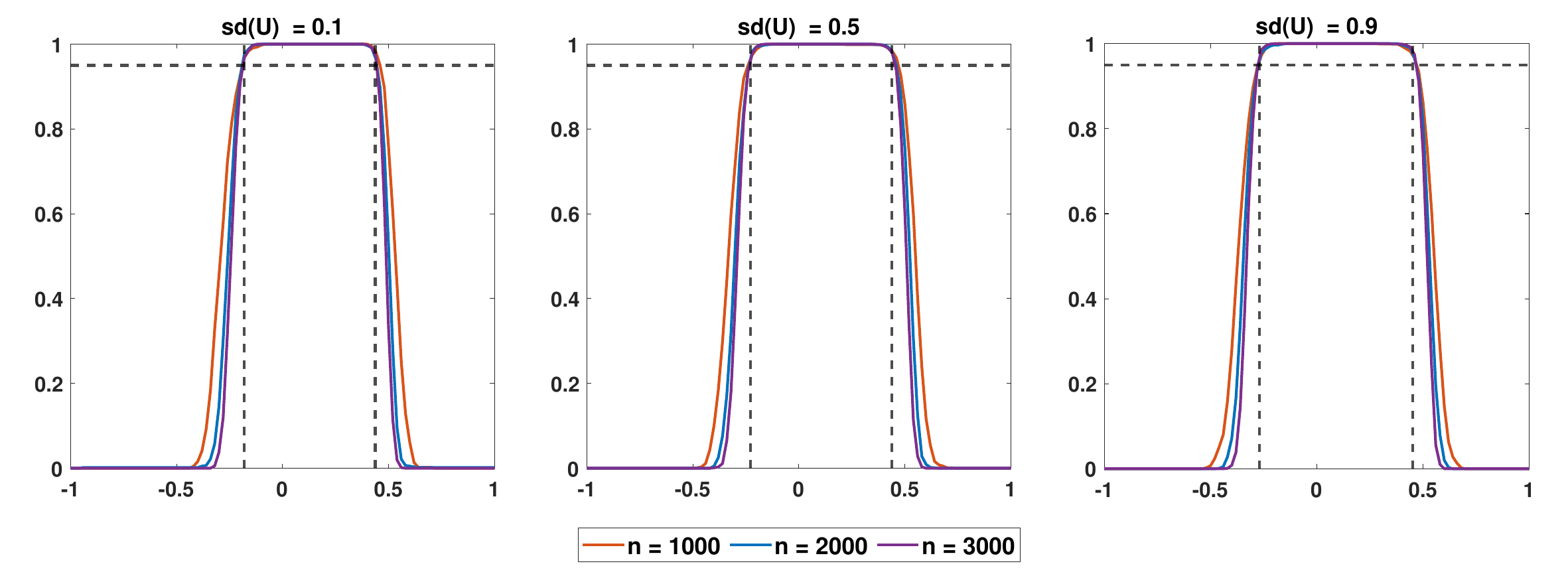}
\subcaption{Two-dimensional $V$}
\includegraphics[width=\linewidth]{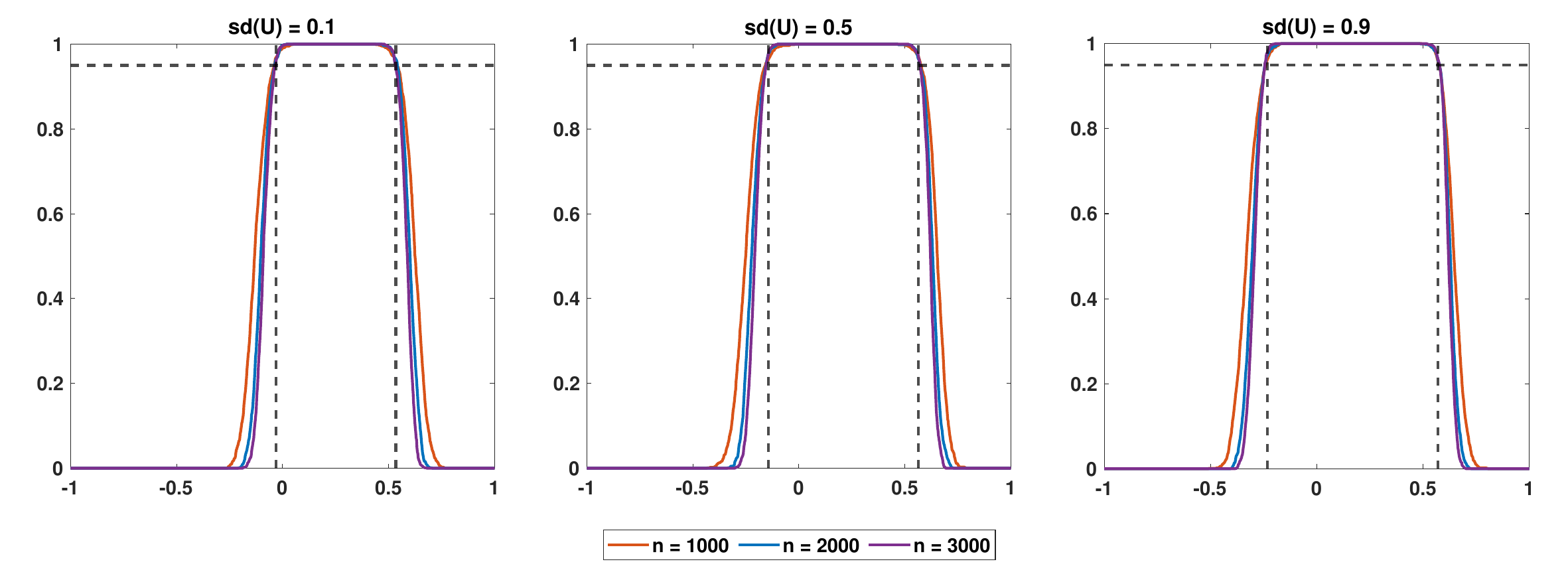}
\footnotesize Note: This figure plots the coverage rates of the 95\% confidence interval for the target parameter ATE, obtained using the regularized support function estimation method of \citet{gafarov2019inference}. We plot the coverage rates using solid curves and the true CvR bounds using dashed vertical lines. The dashed horizontal line marks the 95\% confidence level.
\end{figure}

\begin{figure}[htbp!]
\caption{Coverage Rates of the 95\% Confidence Set for ATT}\label{fig:random_coverage_ATT}
\subcaption{Single-dimensional $V$}
\includegraphics[width=\linewidth]{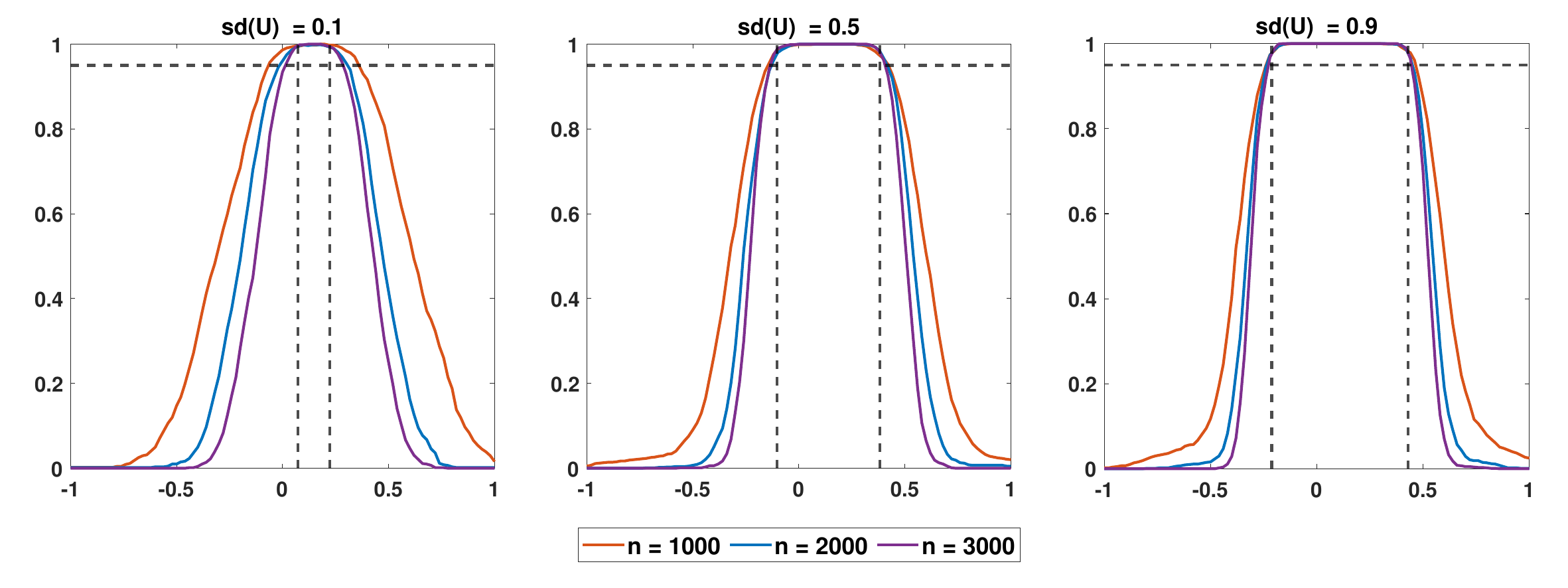}
\subcaption{Two-dimensional $V$}
\includegraphics[width=\linewidth]{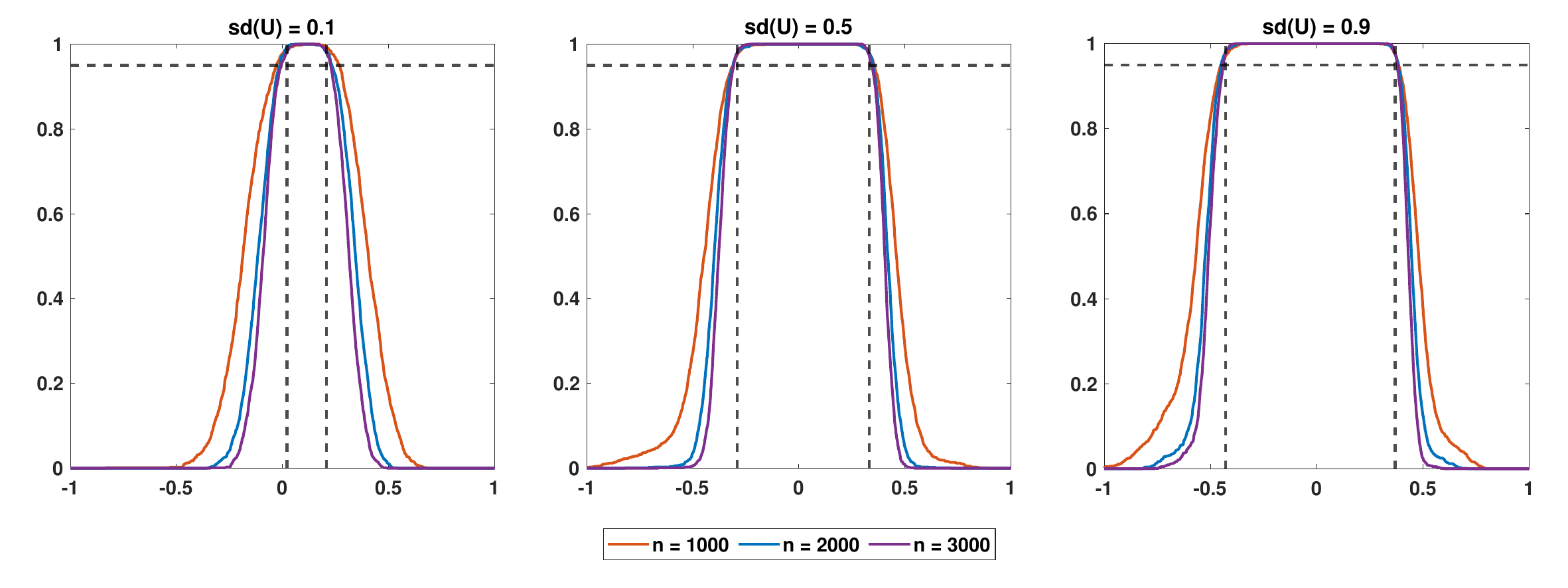}
\footnotesize Note: This figure plots the coverage rates of the 95\% confidence interval for the target parameter ATT, obtained using the regularized support function estimation method of \citet{gafarov2019inference}. We plot the coverage rates using solid curves and the true CvR bounds using dashed vertical lines. The dashed horizontal line marks the 95\% confidence level.
\end{figure}

\begin{figure}[htbp!]
\caption{Coverage Rates of the 95\% Confidence Set for ASB}\label{fig:random_coverage_ABS}
\subcaption{Single-dimensional $V$}
\includegraphics[width=\linewidth]{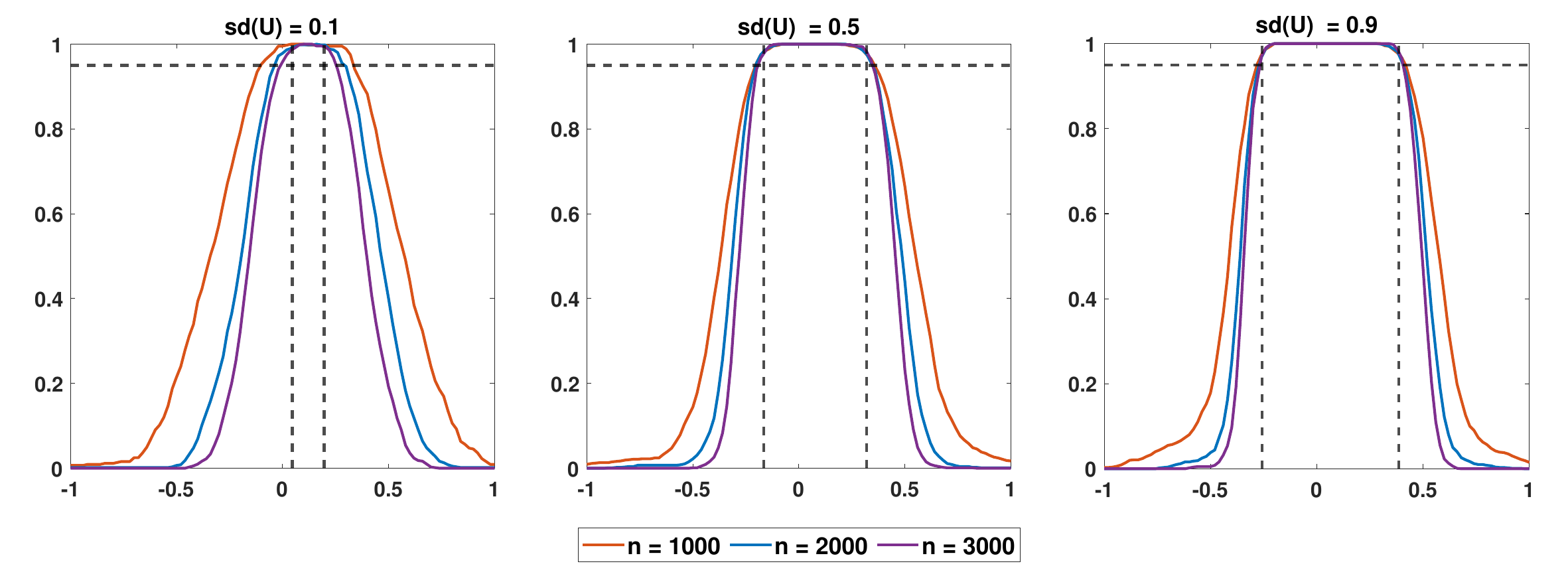}
\subcaption{Two-dimensional $V$}
\includegraphics[width=\linewidth]{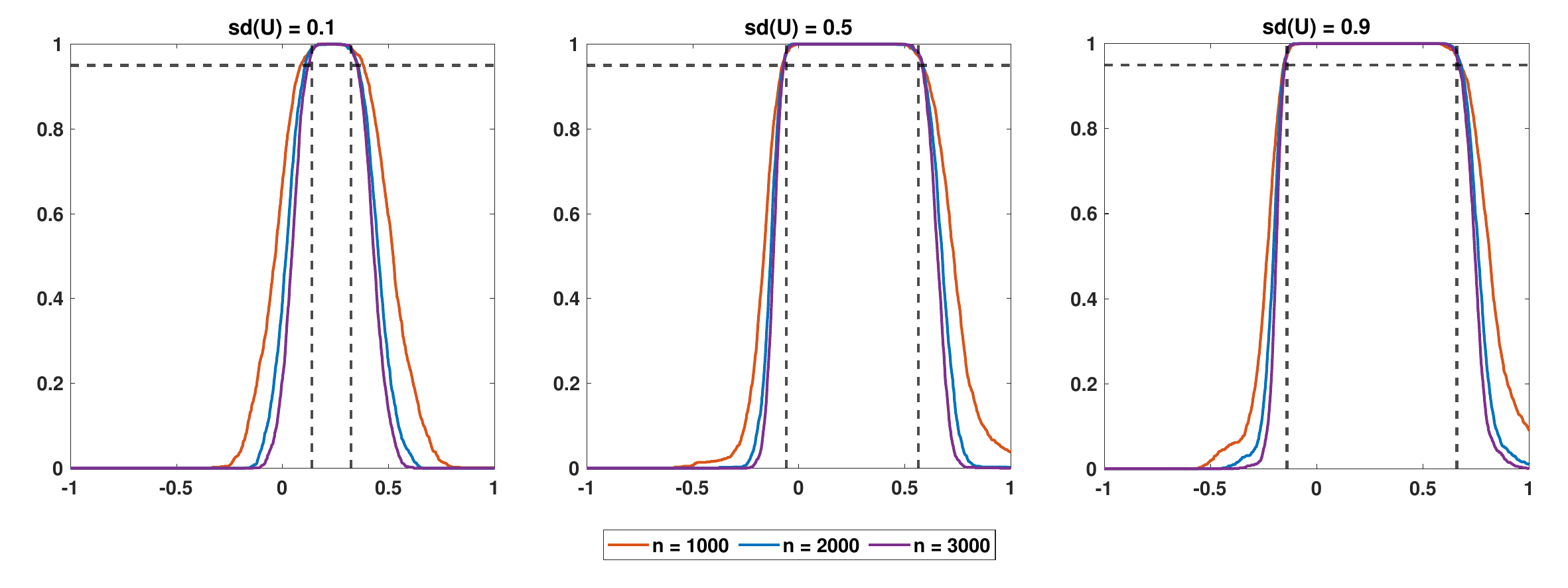}
\footnotesize Note: This figure plots the coverage rates of the 95\% confidence interval for the target parameter ASB, obtained using the regularized support function estimation method of \citet{gafarov2019inference}. We plot the coverage rates using solid curves and the true CvR bounds using dashed vertical lines. The dashed horizontal line marks the 95\% confidence level.
\end{figure}

\section{Conclusion}\label{section:conclusion}

This paper proposes a unified framework to partially identify various policy relevant treatment effects (PRTEs) using a convex relaxation method, when the threshold-crossing structure for the treatment may not hold. Our proposed method robustifies the one of \citet{mogstad/santos/torgovitsky:2017} in the following sense: When the true data generating process satisfies the threshold-crossing assumption for the treatment our convex-relaxation bounds for a large class of PRTE parameters are simplified to the bounds of \citet{mogstad/santos/torgovitsky:2017}, even if we do not impose this condition. This robustness makes our method a more attractive option for practitioners, as it remains applicable regardless of whether the threshold-crossing structure for the treatment holds or not.

We focus on a flexible model that accommodates multidimensional unobserved heterogeneity. We assume that the potential outcomes and the treatment variable are mean-independent given the unobserved heterogeneity. Under this assumption, both the target parameters and the identifiable estimands can be expressed as weighted averages of bilinear functions of the unknown marginal treatment responses and propensity scores given the unobserved heterogeneity. The PRTE bounds are then characterized by nonconvex optimization problems subject to restrictions imposed by the identifiable estimands. 

To address the nonconvexity problem, we introduce a convex relaxation method which is attractive for several reasons. First, it provides a computationally feasible identified set for PRTEs. This method replaces each bilinear function with a new function and four linear inequality constraints, converting a bilinear program to a linear program. Second, we show that our proposed convex-relaxation bounds do not lose identification power under some conditions as long as we fully exploit the available information on observable data. Third, linear shape restrictions and their combinations can be easily incorporated to further improve the bounds. At last, we can directly apply the existing inference method of \citet{gafarov2019inference} to construct uniformly valid confidence intervals for the PRTE bounds. 
Numerical and simulation analyses demonstrate that our convex-relaxation bounds are informative under violations of the threshold-crossing structure.

%

\bibliography{mybib}

\begin{thebibliography}{46}
\newcommand{\enquote}[1]{``#1''}
\expandafter\ifx\csname natexlab\endcsname\relax\def\natexlab#1{#1}\fi

\bibitem[\protect\citeauthoryear{Brinch, Mogstad, and Wiswall}{Brinch
  et~al.}{2017}]{brinch/mogstad/wiswall:2017}
\textsc{Brinch, C.~N., M.~Mogstad, and M.~Wiswall} (2017): \enquote{Beyond LATE
  with a Discrete Instrument,} \emph{Journal of Political Economy}, 125,
  985--1039.

\bibitem[\protect\citeauthoryear{Carr and Kitagawa}{Carr and
  Kitagawa}{2023}]{carr2023testing}
\textsc{Carr, T. and T.~Kitagawa} (2023): \enquote{Testing instrument validity
  with covariates,} \emph{arXiv preprint arXiv:2112.08092}.

\bibitem[\protect\citeauthoryear{Dahl, Huber, and Mellace}{Dahl
  et~al.}{2023}]{dahl2023never}
\textsc{Dahl, C.~M., M.~Huber, and G.~Mellace} (2023): \enquote{It is never too
  LATE: a new look at local average treatment effects with or without defiers,}
  \emph{The Econometrics Journal}, 26, 378--404.

\bibitem[\protect\citeauthoryear{de~Chaisemartin}{de~Chaisemartin}{2017}]{de2017tolerating}
\textsc{de~Chaisemartin, C.} (2017): \enquote{Tolerating defiance? Local
  average treatment effects without monotonicity,} \emph{Quantitative
  Economics}, 8, 367--396.

\bibitem[\protect\citeauthoryear{DiNardo and Lee}{DiNardo and
  Lee}{2011}]{dinardo/lee:2011}
\textsc{DiNardo, J. and D.~S. Lee} (2011): \enquote{Program Evaluation and
  Research Designs,} in \emph{Handbook of Labor Economics}, ed. by
  O.~Ashenfelter and D.~Card, Elsevier, vol. 4, Part A, chap.~5, 463--536.

\bibitem[\protect\citeauthoryear{Flores and Chen}{Flores and
  Chen}{2019}]{carlos2019average}
\textsc{Flores, C.~A. and X.~Chen} (2019): \emph{Average Treatment Effect
  Bounds with an Instrumental Variable: Theory and Practice}, SPRINGER.

\bibitem[\protect\citeauthoryear{Gafarov}{Gafarov}{2025}]{gafarov2019inference}
\textsc{Gafarov, B.} (2025): \enquote{Simple subvector inference on sharp
  identified set in affine models,} \emph{Journal of Econometrics}, 249,
  105952.

\bibitem[\protect\citeauthoryear{Gautier}{Gautier}{2021}]{gautier2021relaxing}
\textsc{Gautier, E.} (2021): \enquote{Relaxing monotonicity in endogenous
  selection models and application to surveys,} in \emph{Advances in
  Contemporary Statistics and Econometrics: Festschrift in Honor of Christine
  Thomas-Agnan}, Springer, 59--78.

\bibitem[\protect\citeauthoryear{Gautier and Hoderlein}{Gautier and
  Hoderlein}{2011}]{gautier2011triangular}
\textsc{Gautier, E. and S.~Hoderlein} (2011): \enquote{A triangular treatment
  effect model with random coefficients in the selection equation,} \emph{arXiv
  preprint arXiv:1109.0362}.

\bibitem[\protect\citeauthoryear{Gautier and Kitamura}{Gautier and
  Kitamura}{2013}]{gautier2013nonparametric}
\textsc{Gautier, E. and Y.~Kitamura} (2013): \enquote{Nonparametric estimation
  in random coefficients binary choice models,} \emph{Econometrica}, 81,
  581--607.

\bibitem[\protect\citeauthoryear{Goff}{Goff}{2024}]{goff2020vector}
\textsc{Goff, L.} (2024): \enquote{A vector monotonicity assumption for
  multiple instruments,} \emph{Journal of Econometrics}, 241, 105735.

\bibitem[\protect\citeauthoryear{Han and Kaido}{Han and
  Kaido}{2024}]{han2024set}
\textsc{Han, S. and H.~Kaido} (2024): \enquote{Set-Valued Control Functions,}
  \emph{arXiv preprint arXiv:2403.00347}.

\bibitem[\protect\citeauthoryear{Heckman}{Heckman}{2001}]{heckman2001micro}
\textsc{Heckman, J.~J.} (2001): \enquote{Micro data, heterogeneity, and the
  evaluation of public policy: Nobel lecture,} \emph{Journal of political
  Economy}, 109, 673--748.

\bibitem[\protect\citeauthoryear{Heckman and Vytlacil}{Heckman and
  Vytlacil}{1999}]{heckman/vytlacil:1999}
\textsc{Heckman, J.~J. and E.~J. Vytlacil} (1999): \enquote{Local instrumental
  variables and latent variable models for identifying and bounding treatment
  effects,} \emph{Proceedings of the national Academy of Sciences}, 96,
  4730--4734.

\bibitem[\protect\citeauthoryear{Heckman and Vytlacil}{Heckman and
  Vytlacil}{2001{\natexlab{a}}}]{heckman/vylacil:2001:book}
---\hspace{-.1pt}---\hspace{-.1pt}--- (2001{\natexlab{a}}):
  \enquote{Instrumental variables, selection models, and tight bounds on the
  average treatment effect,} in \emph{Econometric Evaluation of Labour Market
  Policies}, ed. by M.~Lechner and F.~Pfeiffer, Heidelberg: Physica-Verlag HD,
  1--15.

\bibitem[\protect\citeauthoryear{Heckman and Vytlacil}{Heckman and
  Vytlacil}{2001{\natexlab{b}}}]{heckman/vytlacil:2001}
---\hspace{-.1pt}---\hspace{-.1pt}--- (2001{\natexlab{b}}):
  \enquote{Policy-Relevant Treatment Effects,} \emph{American Economic Review},
  91, 107--111.

\bibitem[\protect\citeauthoryear{Heckman and Vytlacil}{Heckman and
  Vytlacil}{2005}]{heckman/vytlacil:2005}
---\hspace{-.1pt}---\hspace{-.1pt}--- (2005): \enquote{Structural Equations,
  Treatment Effects, and Econometric Policy Evaluation,} \emph{Econometrica},
  73, 669--738.

\bibitem[\protect\citeauthoryear{Heckman and Vytlacil}{Heckman and
  Vytlacil}{2007}]{heckman2007econometric}
---\hspace{-.1pt}---\hspace{-.1pt}--- (2007): \enquote{Econometric evaluation
  of social programs, part II: Using the marginal treatment effect to organize
  alternative econometric estimators to evaluate social programs, and to
  forecast their effects in new environments,} \emph{Handbook of Econometrics},
  6, 4875--5143.

\bibitem[\protect\citeauthoryear{Huber, Laff{\'e}rs, and Mellace}{Huber
  et~al.}{2017}]{huber2017sharp}
\textsc{Huber, M., L.~Laff{\'e}rs, and G.~Mellace} (2017): \enquote{Sharp IV
  bounds on average treatment effects on the treated and other populations
  under endogeneity and noncompliance,} \emph{Journal of Applied Econometrics},
  32, 56--79.

\bibitem[\protect\citeauthoryear{Huber and Mellace}{Huber and
  Mellace}{2015}]{huber/mellace:2015}
\textsc{Huber, M. and G.~Mellace} (2015): \enquote{Testing Instrument Validity
  for LATE Identification Based on Inequality Moment Constraints,} \emph{The
  Review of Economics and Statistics}, 97, 398--411.

\bibitem[\protect\citeauthoryear{Imbens}{Imbens}{2014}]{imbens2014instrumental}
\textsc{Imbens, G.~W.} (2014): \enquote{Instrumental variables: an
  econometrician's perspective,} \emph{Statistical science}, 29, 323--358.

\bibitem[\protect\citeauthoryear{Imbens and Angrist}{Imbens and
  Angrist}{1994}]{imbens/angrist:1994}
\textsc{Imbens, G.~W. and J.~D. Angrist} (1994): \enquote{Identification and
  Estimation of Local Average Treatment Effects,} \emph{Econometrica}, 62,
  467--75.

\bibitem[\protect\citeauthoryear{K{\'e}dagni and Mourifi{\'e}}{K{\'e}dagni and
  Mourifi{\'e}}{2020}]{kedagni2020generalized}
\textsc{K{\'e}dagni, D. and I.~Mourifi{\'e}} (2020): \enquote{Generalized
  instrumental inequalities: testing the instrumental variable independence
  assumption,} \emph{Biometrika}, 107, 661--675.

\bibitem[\protect\citeauthoryear{Kitagawa}{Kitagawa}{2015}]{kitagawa_2015}
\textsc{Kitagawa, T.} (2015): \enquote{A Test for Instrument Validity,}
  \emph{Econometrica}, 83, 2043--2063.

\bibitem[\protect\citeauthoryear{Klein}{Klein}{2010}]{klein2010heterogeneous}
\textsc{Klein, T.~J.} (2010): \enquote{Heterogeneous treatment effects:
  Instrumental variables without monotonicity?} \emph{Journal of Econometrics},
  155, 99--116.

\bibitem[\protect\citeauthoryear{Laff{\'e}rs}{Laff{\'e}rs}{2019}]{laffers2019bounding}
\textsc{Laff{\'e}rs, L.} (2019): \enquote{Bounding average treatment effects
  using linear programming,} \emph{Empirical economics}, 57, 727--767.

\bibitem[\protect\citeauthoryear{Laff{\'e}rs and Mellace}{Laff{\'e}rs and
  Mellace}{2017}]{laffers2017note}
\textsc{Laff{\'e}rs, L. and G.~Mellace} (2017): \enquote{A note on testing
  instrument validity for the identification of LATE,} \emph{Empirical
  Economics}, 53, 1281--1286.

\bibitem[\protect\citeauthoryear{Lee and Salani{\'e}}{Lee and
  Salani{\'e}}{2018}]{lee2018identifying}
\textsc{Lee, S. and B.~Salani{\'e}} (2018): \enquote{Identifying effects of
  multivalued treatments,} \emph{Econometrica}, 86, 1939--1963.

\bibitem[\protect\citeauthoryear{Manski}{Manski}{1990}]{manski1990nonparametric}
\textsc{Manski, C.~F.} (1990): \enquote{Nonparametric bounds on treatment
  effects,} \emph{The American Economic Review}, 80, 319--323.

\bibitem[\protect\citeauthoryear{Manski}{Manski}{1997}]{manski1997monotone}
---\hspace{-.1pt}---\hspace{-.1pt}--- (1997): \enquote{Monotone Treatment
  Response,} \emph{Econometrica}, 65, 1311--1334.

\bibitem[\protect\citeauthoryear{Manski and Pepper}{Manski and
  Pepper}{2000}]{manski2000monotone}
\textsc{Manski, C.~F. and J.~V. Pepper} (2000): \enquote{Monotone instrumental
  variables: {W}ith an application to the returns to schooling,}
  \emph{Econometrica}, 68, 997--1010.

\bibitem[\protect\citeauthoryear{McCormick}{McCormick}{1976}]{mccormick1976computability}
\textsc{McCormick, G.~P.} (1976): \enquote{Computability of global solutions to
  factorable nonconvex programs: Part I—Convex underestimating problems,}
  \emph{Mathematical Programming}, 10, 147--175.

\bibitem[\protect\citeauthoryear{Mogstad, Santos, and Torgovitsky}{Mogstad
  et~al.}{2018}]{mogstad/santos/torgovitsky:2017}
\textsc{Mogstad, M., A.~Santos, and A.~Torgovitsky} (2018): \enquote{Using
  Instrumental Variables for Inference About Policy Relevant Treatment
  Parameters,} \emph{Econometrica}, 86, 1589--1619.

\bibitem[\protect\citeauthoryear{Mogstad and Torgovitsky}{Mogstad and
  Torgovitsky}{2024}]{mogstad2024instrumental}
\textsc{Mogstad, M. and A.~Torgovitsky} (2024): \enquote{Instrumental variables
  with unobserved heterogeneity in treatment effects,} in \emph{Handbook of
  Labor Economics}, Elsevier, vol.~5, 1--114.

\bibitem[\protect\citeauthoryear{Mogstad, Torgovitsky, and Walters}{Mogstad
  et~al.}{2021}]{mogstad2020policy}
\textsc{Mogstad, M., A.~Torgovitsky, and C.~R. Walters} (2021): \enquote{The
  causal interpretation of two-stage least squares with multiple instrumental
  variables,} \emph{American Economic Review}, 111, 3663--3698.

\bibitem[\protect\citeauthoryear{Mourifi{\'e} and Wan}{Mourifi{\'e} and
  Wan}{2017}]{mourifie2017testing}
\textsc{Mourifi{\'e}, I. and Y.~Wan} (2017): \enquote{Testing local average
  treatment effect assumptions,} \emph{Review of Economics and Statistics}, 99,
  305--313.

\bibitem[\protect\citeauthoryear{Navjeevan, Pinto, and Santos}{Navjeevan
  et~al.}{2023}]{navjeevan2023identification}
\textsc{Navjeevan, M., R.~Pinto, and A.~Santos} (2023): \enquote{Identification
  and estimation in a class of potential outcomes models,} \emph{arXiv preprint
  arXiv:2310.05311}.

\bibitem[\protect\citeauthoryear{Noack}{Noack}{2021}]{noack2021sensitivity}
\textsc{Noack, C.} (2021): \enquote{Sensitivity of LATE estimates to violations
  of the monotonicity assumption,} \emph{arXiv preprint arXiv:2106.06421}.

\bibitem[\protect\citeauthoryear{Russell}{Russell}{2021}]{russell2021sharp}
\textsc{Russell, T.~M.} (2021): \enquote{Sharp bounds on functionals of the
  joint distribution in the analysis of treatment effects,} \emph{Journal of
  Business \& Economic Statistics}, 39, 532--546.

\bibitem[\protect\citeauthoryear{Sloczynski}{Sloczynski}{2021}]{sloczynski2021should}
\textsc{Sloczynski, T.} (2021): \enquote{When Should We (Not) Interpret Linear
  IV Estimands as LATE?} Working Paper.

\bibitem[\protect\citeauthoryear{Small, Tan, Ramsahai, Lorch, and
  Brookhart}{Small et~al.}{2017}]{small2017instrumental}
\textsc{Small, D.~S., Z.~Tan, R.~R. Ramsahai, S.~A. Lorch, and M.~A. Brookhart}
  (2017): \enquote{Instrumental variable estimation with a stochastic
  monotonicity assumption,} \emph{Statistical Science}, 32, 561--579.

\bibitem[\protect\citeauthoryear{Sun}{Sun}{2023}]{SUN2023105523}
\textsc{Sun, Z.} (2023): \enquote{Instrument validity for heterogeneous causal
  effects,} \emph{Journal of Econometrics}, 237, 105523.

\bibitem[\protect\citeauthoryear{van't Hoff}{van't
  Hoff}{2023}]{hoff2023identifying}
\textsc{van't Hoff, N.} (2023): \enquote{Identifying Causal Effects of
  Nonbinary, Ordered Treatments using Multiple Instrumental Variables,}
  \emph{arXiv preprint arXiv:2311.17575}.

\bibitem[\protect\citeauthoryear{van’t Hoff, Lewbel, and Mellace}{van’t
  Hoff et~al.}{2024}]{van2023limited}
\textsc{van’t Hoff, N., A.~Lewbel, and G.~Mellace} (2024): \enquote{Limited
  Monotonicity and the Combined Compliers LATE,} Tech. rep., Boston College
  Department of Economics.

\bibitem[\protect\citeauthoryear{Vytlacil}{Vytlacil}{2002}]{vytlacil:2002}
\textsc{Vytlacil, E.~J.} (2002): \enquote{Independence, Monotonicity, and
  Latent Index Models: An Equivalence Result,} \emph{Econometrica}, 70,
  331--341.

\bibitem[\protect\citeauthoryear{Yap}{Yap}{2022}]{yap2022sensitivity}
\textsc{Yap, L.} (2022): \enquote{Sensitivity of Policy Relevant Treatment
  Parameters to Violations of Monotonicity,} \emph{Working Papers 655,
  Princeton University, Department of Economics, Industrial Relations Section.}

\end{thebibliography}

\clearpage
\appendix
\section*{Supplemental Appendix}
The appendix is organized as follows. Appendix \ref{appendix_weights} provides the derivations for the weights of various target parameters given in Table \ref{table:target_parameters}. Appendix \ref{appendix_W} presents the expression of $\mathbf{W}$ introduced in Section \ref{subsection:inference}. Appendix \ref{appendix_proofs} provides proofs for results in the main text.

\section{Derivations for Table \ref{table:target_parameters}}\label{appendix_weights}
In this section, all the expressions for the PRTE weights are derived under Assumption \ref{assumption_DGP}.\medskip

\noindent\textbf{Weights for $\mathbb{E}[Y_0]$, $\mathbb{E}[Y_1]$, and ATE}.
\begin{align*}
\mathbb{E}[Y_d]= &\mathbb{E}\{ \mathbb{E}[Y_d\mid V,Z]\}=\mathbb{E}[m_{\mathbb{P},d}(V,X)]=\mathbb{E}\Big[\int_{{\mathcal{V}}}m_{\mathbb{P},d}(v,X)dv\Big],\\
\mathbb{E}[Y_1-Y_0]=&\mathbb{E}\Big[\int_{{\mathcal{V}}}m_{\mathbb{P},1}(v,X)-m_{\mathbb{P},0}(v,X)dv\Big].
\end{align*}

\noindent\textbf{Weights for ATE given $X\in\mathcal{X}^\star$}.
\begin{align*}
\mathbb{E}[Y_d\mid X\in\mathcal{X}^\star]= &\frac{\mathbb{E}\{Y_d1[X\in\mathcal{X}^\star]\}}{\mathbb{P}(X\in\mathcal{X}^\star)}=\frac{\mathbb{E}\{\mathbb{E}[Y_d\mid V,Z]1[X\in\mathcal{X}^\star]\}}{\mathbb{P}(X\in\mathcal{X}^\star)}\\=&\frac{\mathbb{E}\{m_{\mathbb{P},d}( V,X)1[X\in\mathcal{X}^\star]\}}{\mathbb{P}(X\in\mathcal{X}^\star)}
=\frac{\mathbb{E}\{\int_{{\mathcal{V}}}m_{\mathbb{P},d}( v,X)1[X\in\mathcal{X}^\star]dv\}}{\mathbb{P}(X\in\mathcal{X}^\star)},\\
\mathbb{E}[Y_1-Y_0\mid X\in\mathcal{X}^\star]= &\frac{\mathbb{E}\{\int_{{\mathcal{V}}}[m_{\mathbb{P},1}( v,X)-m_{\mathbb{P},0}(v,X)]1[X\in\mathcal{X}^\star]dv\}}{\mathbb{P}(X\in\mathcal{X}^\star)}.
\end{align*}

\noindent\textbf{Weights for ATT and ATU}.
\begin{align*}
\mathbb{E}[Y_d|D=1]=&\frac{\mathbb{E}[Y_dD]}{\mathbb{P}(D=1)}
=\frac{\mathbb{E}\{\mathbb{E}[Y_dD\mid V,Z]\}}{\mathbb{P}(D=1)}
=\frac{\mathbb{E}[m_{\mathbb{P},d}( V,X)m_{\mathbb{P},D}(V,Z)]}{\mathbb{P}(D=1)}\\
=&\frac{\mathbb{E}[\int_{{\mathcal{V}}}m_{\mathbb{P},d}(v,X)m_{\mathbb{P},D}(v,Z)dv]}{\mathbb{P}(D=1)},\\
\mathbb{E}[Y_d|D=0]=&\frac{\mathbb{E}[Y_d(1-D)]}{\mathbb{P}(D=0)}
=\frac{\mathbb{E}[\int_{{\mathcal{V}}}m_{\mathbb{P},d}(v,X)(1-m_{\mathbb{P},D}(v,Z))dv]}{\mathbb{P}(D=0)}.
\end{align*}
Thus,
\begin{align*}
\mathbb{E}[Y_1-Y_0\mid D=1]=&\frac{\mathbb{E}[\int_{{\mathcal{V}}}(m_{\mathbb{P},1}(v,X)-m_{\mathbb{P},0}(v,X))m_{\mathbb{P},D}(v,Z)dv]}{\mathbb{P}(D=1)},\\
\mathbb{E}[Y_1-Y_0\mid D=0]=&\frac{\mathbb{E}[\int_{{\mathcal{V}}}(m_{\mathbb{P},1}(v,X)-m_{\mathbb{P},0}(v,X))(1-m_{\mathbb{P},D}(v,Z))dv]}{\mathbb{P}(D=0)}.
\end{align*}

\noindent\textbf{Weights for Generalized Local Average Treatment Effect given $V\in\mathcal{V}^\star$}.
\begin{align*}
\mathbb{E}[Y_d\mid V\in\mathcal{V}^\star]=&\frac{\mathbb{E}[Y_d1[V\in\mathcal{V}^\star]]}{\mathbb{P}(V\in\mathcal{V}^\star)}
=\frac{\mathbb{E}\{\mathbb{E}[Y_d1[V\in\mathcal{V}^\star]\mid V,Z]\}}{\mathbb{P}(V\in\mathcal{V}^\star)}\\
=&\frac{\mathbb{E}\{1[V\in\mathcal{V}^\star]m_{\mathbb{P},d}(V,X)\}}{\mathbb{P}(V\in\mathcal{V}^\star)}.
\end{align*}
Then we can derive that
\begin{align*}
\mathbb{E}[Y_1-Y_0\mid V\in\mathcal{V}^\star]=&\frac{\mathbb{E}\{1[V\in\mathcal{V}^\star][m_{\mathbb{P},1}(V,X)-m_{\mathbb{P},0}(V,X)]\}}{\mathbb{P}(V\in\mathcal{V}^\star)}\\
=&\mathbb{E}\Big[\int_{{\mathcal{V}}}\frac{1[V\in\mathcal{V}^\star]}{\mathbb{P}(V\in\mathcal{V}^\star)}m_{\mathbb{P},1}(v,X)-\frac{1[V\in\mathcal{V}^\star]}{\mathbb{P}(V\in\mathcal{V}^\star)}m_{\mathbb{P},0}(v,X)dv\Big].
\end{align*}

\noindent\textbf{Weights for PRTE}. The new policy is denoted as $Z^\star=(X,Z_0^\star)\in\mathcal{Z}^\star$, where $X$ is the vector of covariates whose distribution is unchanged in the new policy. Note that $(Y^\star,D^\star)$ denotes the outcome and treatment variable generated by $Z^\star$ and keeping $(m_{\mathbb{P}},V)$ the same, such that $(Y^\star,D^\star,Z^\star)$ also satisfies Assumption \ref{assumption_DGP}. By definition, we have
\begin{align}\label{EYstar0}
\mathbb{E}[Y^\star]=&\mathbb{E}[Y_0(1-D^\star)+Y_1D^\star]\nonumber\\
=&\int_\mathcal{Z^\star}\int_{{\mathcal{V}}}\left[m_{\mathbb{P},0}(v,x)(1-m_{\mathbb{P},D}(v,z))+m_{\mathbb{P},1}(v,x)m_{\mathbb{P},D}(v,z)\right]dvF_{Z^\star}(dz),
\end{align}

Consider a distributional change in the instrument variable with its new support being a subset of the original support, i.e., $Z^\star=(X,Z_0^\star)\sim F_{Z^\star}$ for some known distribution $F_{Z^\star}$ and $\mathcal{Z}^\star\subseteq\mathcal{Z}$. Since $F_{Z}(dz)\neq0$ for all $z\in\mathcal{Z}$,  \eqref{EYstar0} becomes
\begin{align*}
\mathbb{E}[Y^\star]=&\int_\mathcal{Z}\int_{{\mathcal{V}}}\left[m_{\mathbb{P},0}(v,x)(1-m_{\mathbb{P},D}(v,z))+m_{\mathbb{P},1}(v,x)m_{\mathbb{P},D}(v,z)\right]dvF_{Z^\star}(dz)\\
=&\int_\mathcal{Z}\int_{{\mathcal{V}}}\left[m_{\mathbb{P},0}(v,x)(1-m_{\mathbb{P},D}(v,z))+m_{\mathbb{P},1}(v,x)m_{\mathbb{P},D}(v,z)\right]\frac{F_{Z^\star}(dz)}{F_{Z}(dz)}dvF_{Z}(dz).
\end{align*}

\noindent\textbf{Weights for Average Selection Bias}.
\begin{align*}\mathbb{E}[Y_0\mid D=1]-\mathbb{E}[Y_0\mid D=0]=&\frac{\mathbb{E}[\int_{{\mathcal{V}}}m_{\mathbb{P},0}(v,X)m_{\mathbb{P},D}(v,Z)dv]}{\mathbb{P}(D=1)}\\
&~~~~~~-\frac{\mathbb{E}[\int_{{\mathcal{V}}}m_{\mathbb{P},0}(v,X)(1-m_{\mathbb{P},D}(v,Z))dv]}{\mathbb{P}(D=0)},
\end{align*}

\noindent\textbf{Weights for Average Selection on the Gain}.
\begin{align*}
\mathbb{E}[Y_1-Y_0\mid D=1]-\mathbb{E}[Y_1-Y_0\mid D=0]=&\frac{\mathbb{E}[\int_{{\mathcal{V}}}(m_{\mathbb{P},1}(v,X)-m_{\mathbb{P},0}(v,X))m_{\mathbb{P},D}(v,Z)dv]}{\mathbb{P}(D=1)}\\
-&\frac{\mathbb{E}[\int_{{\mathcal{V}}}(m_{\mathbb{P},1}(v,X)-m_{\mathbb{P},0}(v,X))(1-m_{\mathbb{P},D}(v,Z))dv]}{\mathbb{P}(D=0)}.
\end{align*}



\section{Expression of $\mathbf{W}$}\label{appendix_W}
In this section, we provide the expression of $\mathbf{W}$ introduced in Section \ref{subsection:inference}. Recall that we rewrite $$\mathbf{W}=\begin{bmatrix}\mathbf{W}_1^{eq}&\mathbf{W}_2^{eq}\\\mathbf{W}^{ineq}_1&\mathbf{W}^{ineq}_2\end{bmatrix},$$ 
where the four blocks $\mathbf{W}^{eq}_1\in\mathbb{R}^{d_{eq}\times d_\eta}$, $\mathbf{W}^{eq}_2\in\mathbb{R}^{d_{eq}\times 1}$, $\mathbf{W}^{ineq}_1\in\mathbb{R}^{d_{ineq}\times d_\eta}$, and $\mathbf{W}^{ineq}_2\in\mathbb{R}^{d_{ineq}\times 1}$ are compatible with the dimensions of $\mathcal{J}^{eq}$ and $\mathcal{J}^{ineq}$. The definition of these blocks is given below. Denote $\mathbf{1}_b$ and $\mathbf{0}_b$ as $b\times 1$ vector of ones and zeros, respectively, and denote $\mathbf{0}_{b\times c}$ as a $b\times c$ matrix of zeros. Define a $ d_{\eta_2}\times1$ vector
\begin{align}
\mathbf{T}^\star(Z)=&\Big[
\int_{{\mathcal{V}}} \mathbf{B}'_{0}(v,X)\omega^\star_{01}(v,Z)dv,~\int_{{\mathcal{V}}} \mathbf{B}'_{1}(v,X)\omega^\star_{10}(v,Z)dv,~\mathbf{0}'_{K_D},\nonumber\\
&\int_{{\mathcal{V}}} \mathbf{B}'_{0D}(v,Z)(\omega^\star_{00}(v,Z)-\omega^\star_{01}(v,Z))dv,~\int_{{\mathcal{V}}} \mathbf{B}'_{1D}(v,Z)(\omega^\star_{11}(v,Z)-\omega^\star_{10}(v,Z))dv
\Big]'.
\end{align}
Then, we have $\eta_1=\mathbb{E}[\mathbf{T}^\star(Z)]'\eta_2$. For  $d=0,1$ and $z=(z_0,x)$, define 
\begin{align*}
\mathbf{u}_{d}(x)=&\int_{{\mathcal{V}}} \mathbf{B}_{d}(v,x)dv,\qquad\mathbf{u}_{D}(z)=\int_{{\mathcal{V}}} \mathbf{B}_{D}(v,z)dv,\qquad
\mathbf{u}_{dD}(z)=\int_{{\mathcal{V}}} \mathbf{B}_{dD}(v,z)dv.\end{align*} 
Define a $K_{D}\times |\mathcal{S}|$ matrix $\mathbf{U}_{D}(z)$ and a $K_{dD}\times |\mathcal{S}|$ matrix $\mathbf{U}_{dD}(z)$ as below
\begin{align*}
\mathbf{U}_{D}(z)=&[s_1(z)\mathbf{u}_{D}(z),...,s_{|\mathcal{S}|}(z)\mathbf{u}_{D}(z)],\qquad~~
\mathbf{U}_{dD}(z)=[s_1(z)\mathbf{u}_{dD}(z),...,s_{|\mathcal{S}|}(z)\mathbf{u}_{dD}(z)].\end{align*}
Let us further partition $\mathbf{W}^{eq}_1$ as $\mathbf{W}^{eq}_1=\begin{bmatrix}
1 & -\mathbf{T}^{\star'}(Z) \\
\mathbf{0}_{3|\mathcal{S}|}&\mathbf{W}^{eq}_{1,2}(Z)\\
\end{bmatrix}$,
where the first row corresponds to the equality constraint $\eta_1=\mathbb{E}[\mathbf{T}^{\star'}(Z)]\eta_2$, and the rest rows correspond to the equality constraints in \eqref{cond:Gamma}. We have
\begin{align*}
&\mathbf{W}^{eq}_{1,2}(Z)=\begin{bmatrix}
\mathbf{0}_{|\mathcal{S}|\times K_0}&\mathbf{0}_{|\mathcal{S}|\times K_1}&\mathbf{0}_{|\mathcal{S}|\times K_D}&\mathbf{0}_{|\mathcal{S}|\times K_{0D}}& \mathbf{U}'_{1D}(Z)\\
\mathbf{0}_{|\mathcal{S}|\times K_0}&\mathbf{0}_{|\mathcal{S}|\times K_1}&\mathbf{0}_{|\mathcal{S}|\times K_D}& \mathbf{U}'_{0D}(Z)&\mathbf{0}_{|\mathcal{S}|\times K_{1D}}\\
\mathbf{0}_{|\mathcal{S}|\times K_0}&\mathbf{0}_{|\mathcal{S}|\times K_1}&\mathbf{U}'_{D}(Z)& \mathbf{0}_{|\mathcal{S}|\times K_{0D}}&\mathbf{0}_{|\mathcal{S}|\times K_{1D}}
\end{bmatrix}
\end{align*}
and $$\mathbf{W}^{eq}_2=
[0,s_1(Z)YD,...,
s_{|\mathcal{S}|}(Z)YD,s_1(Z)Y(1-D),...,s_{|\mathcal{S}|}(Z)Y(1-D),s_1(Z)D,...,s_{|\mathcal{S}|}(Z)D]'.$$
In addition, we use $\mathbf{W}^{ineq}_1$ to describe the inequality constraints that come from the convex relaxation in Lemma \ref{lemma:mccormick}. 
For $W,W'\in\{L,U\}$ and $d=0,1$, denote 
\begin{align*}
I^W_d(x)=\int_{\mathcal{V}} m^W_{d}(v,x)dv,\quad I^W_D(z)=\int_{\mathcal{V}} m^W_{D}(v,z)dv,\quad I^{WW'}_{dD}(z)=\int_{\mathcal{V}} m^W_{d}(v,x)m^{W'}_{D}(v,z)dv,
\end{align*}and
\begin{align*}
\mathbf{Q}^{W}_{0D}(z)=&\int_{\mathcal{V}} m^W_0(v,x)\mathbf{B}_{D}(v,z)dv,\qquad \mathbf{Q}^{W}_{1D}(z)=\int_{\mathcal{V}}m^W_1(v,x)\mathbf{B}_{D}(v,z)dv\\
\mathbf{Q}^{W}_{D0}(z)=&\int_{\mathcal{V}} (1-m^W_D(v,z))\mathbf{B}_{0}(v,x)dv,\qquad \mathbf{Q}^{W}_{D1}(z)=\int_{\mathcal{V}}m^W_D(v,z)\mathbf{B}_{1}(v,x)dv.
\end{align*}
Define 
\begin{align*}
&w_1=(-1,1)',\\ &w_2(z)=-w_1\otimes \left(\mathbf{Q}^{U}_{D0}(z),\mathbf{Q}^{L}_{D0}(z)\right)',\\ &w_3(z)=\left(-\mathbf{Q}^{L}_{0D}(z),-\mathbf{Q}^{U}_{0D}(z),\mathbf{Q}^{U}_{0D}(z),\mathbf{Q}^{L}_{0D}(z)\right)',\\ &w_4(z)=-w_1\otimes \left(\mathbf{Q}^{L}_{D1}(z),\mathbf{Q}^{U}_{D1}(z)\right)',\\ &w_5(z)=\left(\mathbf{Q}^{L}_{1D}(z),\mathbf{Q}^{U}_{1D}(z),-\mathbf{Q}^{U}_{1D}(z),-\mathbf{Q}^{L}_{1D}(z)\right)',\\ &w_6=(-1,-1,1,1)'.
\end{align*}
Consider a partition
$\mathbf{W}^{ineq}_1=\begin{bmatrix}
\mathbf{0}_{d_{ineq}}& \mathbf{W}^{ineq}_{1,2}\\
\end{bmatrix},$
where the zeros in the first column are the coefficients in the inequalities corresponding to $\eta_1$, and $\mathbf{W}^{ineq}_{1,2}$ defined below consists of the coefficients in the inequalities corresponding to $\eta_2$. Then, for any given $z\in\mathcal{Z}$, let us define
\begin{align}\label{alm_B_21}
\mathbf{W}^{ineq}_{1,2}(z)=
&\left[\begin{array}{c:c:c:c:c}
w_1\otimes\mathbf{u}'_{0}(x)&\mathbf{0}_{2\times K_1}&\mathbf{0}_{2\times K_D}&\mathbf{0}_{2\times K_{0D}}&\mathbf{0}_{2\times K_{1D}}\\
\mathbf{0}_{2\times K_{0}}&w_1\otimes\mathbf{u}'_{1}(x)&\mathbf{0}_{2\times K_{D}}&\mathbf{0}_{2\times K_{0D}}&\mathbf{0}_{2\times K_{1D}}\\
\mathbf{0}_{2\times K_{0}}&\mathbf{0}_{2\times K_{1}}&w_1\otimes\mathbf{u}'_{D}(z)&\mathbf{0}_{2\times K_{0D}}&\mathbf{0}_{2\times K_{1D}}\vspace{2mm}\\
\hdashline[3pt/3pt]
w_2(z) &\mathbf{0}_{4\times K_1}&w_3(z)&w_6\otimes\mathbf{u}'_{0D}(z)&\mathbf{0}_{4\times K_{1D}}\\
\mathbf{0}_{4\times K_0}&w_4(z)&w_5(z)&\mathbf{0}_{4\times K_{0D}}&w_6\otimes\mathbf{u}'_{1D}(z)\\
\end{array}\right],
\end{align}
where the upper blocks above the dashed line represent the inequalities in \eqref{eq:mccormick_bounds_m} for any given $z$, and the lower blocks below the dashed line represent the inequalities in \eqref{eq:mccormick_relax1_0} and \eqref{eq:mccormick_relax4_1} for any given $z$. In addition, we denote
\begin{align}\mathbf{W}^{ineq}_2(z)=&
[-I^{L}_0(x),I^{U}_0(x),-I^{L}_1(x),I^{U}_1(x),-I^{L}_D(z),I^{U}_D(z),\nonumber\\
-&I^{LU}_{0D}(z),-I^{UL}_{0D}(z),I^{UU}_{0D}(z),I^{LL}_{0D}(z),
I^{LL}_{1D}(z),I^{UU}_{1D}(z),-I^{UL}_{1D}(z),-I^{LU}_{1D}(z)]'.\end{align}
If $Z$ has a discrete support, we can define $\mathbf{W}^{ineq}_{1,2}$ by stacking $\mathbf{W}^{ineq}_{1,2}(z)$ for all $z\in\mathcal{Z}$ and define $\mathbf{W}^{ineq}_{2}$ by stacking $\mathbf{W}^{ineq}_{2}(z)$ for all $z\in\mathcal{Z}$. If $Z$ is continuous, then we can choose candidate values in $\mathcal{Z}$, say $\{z_1,...,z_{K_Z}\}$, and define $\mathbf{W}^{ineq}_{1,2}$ and $\mathbf{W}^{ineq}_{2}$ by stacking $\mathbf{W}^{ineq}_{1,2}(z)$ and $\mathbf{W}^{ineq}_{2}(z)$ for all candidate values $z$ in $\{z_1,...,z_{K_Z}\}$, respectively.

\section{Proofs}\label{appendix_proofs}
In this section, we present the proofs for the results in the main text.

\begin{proof}[{Proof of Lemma \ref{lemma_Y0Y1K2}}]
Define the random variable $V=(V_0,V_1)$ based on $(Y_0,Y_1,Z)$ by 
$$
(V_0,V_1)=\left(F_{Y_0\mid X}(Y_0),F_{Y_1\mid Y_0,X}(Y_1)\right).
$$

First, we are going to show Assumption \ref{assumption_DGP} (i). 
By the invertibility of $F_{Y_0\mid X}$ and $F_{Y_1\mid Y_0,X}$, we have 
$$
(Y_0,Y_1)=\left(F_{Y_0\mid X}^{-1}(V_0),F_{Y_1\mid Y_0,X}^{-1}(V_1)\right).
$$
It implies 
$\mathbb{E}[Y_d\mid D,Z,{V}]=Y_d={\mathbb{E}}[Y_d\mid X,{V}]$  for each $d=0,1$.

Second, we are going to show Assumption \ref{assumption_DGP} (ii). 
It suffices to show 
$V_0\mid Z\sim \mathrm{Uniform}[0,1]$ and $V_1\mid V_0,Z\sim \mathrm{Uniform}[0,1]$. 
The conditional distribution of $V_0$ given $Z$ satisfies the following equalities: 
\begin{align*}
\mathbb{P}(V_0\leq v\mid Z)
&=
\mathbb{P}(Y_0\leq F_{Y_0\mid X}^{-1}(v)\mid Z)
\\
&=
\mathbb{E}[\mathbb{P}(Y_0\leq F_{Y_0\mid X}^{-1}(v)\mid Z,X)\mid Z]
\\
&=
\mathbb{E}[\mathbb{P}(Y_0\leq F_{Y_0\mid X}^{-1}(v)\mid X)\mid Z]
\\
&=
\mathbb{E}[F_{Y_0\mid X}(F_{Y_0\mid X}^{-1}(v))\mid Z]
\\
&=
\mathbb{E}[v\mid Z]
\\
&=
v,
\end{align*}
where the first equality follows from the strict monotonicity of $F_{Y_0\mid X}$, the second equality follows from the law of iterated expectations, the third equality follows from the independence of $Y_0$ and $Z$ given $X$, and the fourth equality follows from $F_{Y_0\mid X}(F_{Y_0\mid X}^{-1}(v))=v$. 
The conditional distribution of $V_1$ given $(V_0,Z)$ satisfies the following equalities: 
\begin{align*}
\mathbb{P}(V_1\leq v\mid V_0,Z)
&=
\mathbb{P}(Y_1\leq F_{Y_1\mid Y_0,X}^{-1}(v)\mid V_0,Z)
\\
&=
\mathbb{E}[\mathbb{P}(Y_1\leq F_{Y_1\mid Y_0,X}^{-1}(v)\mid Y_0,Z,X)
\mid V_0,Z]
\\
&=
\mathbb{E}[\mathbb{P}(Y_1\leq F_{Y_1\mid Y_0,X}^{-1}(v)\mid Y_0,X)
\mid V_0,Z]
\\
&=
\mathbb{E}[v
\mid V_0,Z]
\\
&=
v,
\end{align*}
where the first equality follows from the strict monotonicity of $F_{Y_1\mid Y_0,X}$, the second equality follows from the law of iterated expectations, the third equality follows from the independence of $(Y_0,Y_1)$ and $Z$ given $X$, and the fourth equality follows from $F_{Y_1\mid Y_0,X}(F_{Y_1\mid Y_0,X}^{-1}(v))=v$.







\end{proof}

\begin{proof}[{Proof of Eq. \eqref{MTS_constraints}}]If $d=1$, we can show that $\mathbb{E}[Y_1\mid X=x,D=1]=\mathbb{E}[Y\mid X=x,D=1],$ and 
\begin{align*}
\mathbb{E}[Y_1\mid X=x,D=0]= &\mathbb{E}[Y_1(1-D)\mid X=x]/\mathbb{P}(D=0\mid X=x)\\=&\mathbb{E}\{\mathbb{E}[Y_1(1-D)\mid X=x, V,Z]\mid X=x\}/\mathbb{P}(D=0\mid X=x)\\=&\mathbb{E}[m_{\mathbb{P},1}(V,x)(1-m_{\mathbb{P},D}(V,Z))\mid X=x]/\mathbb{P}(D=0\mid X=x)\\
=&\mathbb{E}[m_{\mathbb{P},1}(V,x)-m_{\mathbb{P},1D}(V,Z))\mid X=x]/\mathbb{P}(D=0\mid X=x),\end{align*} where the third equality is based on Assumption \ref{assumption_DGP}, and the last equality comes from the definition of $m_{\mathbb{P},1D}$. Plugging these two expressions into the inequality in Condition \ref{condition:mono_select} gives us the first inequality in Eq. \eqref{MTS_constraints}. Similar arguments can be applied to show the second inequality in Eq. \eqref{MTS_constraints} when $d=0$.
\end{proof}

\begin{proof}[{Proof of Theorem \ref{eq:theorem:sharpness}}]
Note that under Assumption \ref{assumption_DGP}, we can show
\begin{align*}
\mathbb{E}[YD\mid Z]
&=
\mathbb{E}[\mathbb{E}[Y\mid D=1,V,Z]\mathbb{E}[D\mid V,Z]\mid Z]
\\
&=
\mathbb{E}[\mathbb{E}[Y_1\mid D=1,V,Z]\mathbb{E}[D\mid V,Z]\mid Z]
\\
&=
\mathbb{E}[\mathbb{E}[Y_1\mid V,X]\mathbb{E}[D\mid V,Z]\mid Z]
\\
&=
\mathbb{E}[m_{\mathbb{P},1}(V,X)m_{\mathbb{P},D}(V,Z)\mid Z]\\
=&
\int_{{\mathcal{V}}}m_{\mathbb{P},1}(v,X)m_{\mathbb{P},D}(v,Z)dv\\
=&\int_{{\mathcal{V}}}m_{\mathbb{P},1D}(v,Z)dv,
\end{align*}
where the first equality is by the law of iterated expectation and the mean independence between $Y_d$ and $D$ given $(V,Z)$, the third equality is based on $\mathbb{E}[Y_1\mid D,Z,{V}]=\mathbb{E}[Y_1\mid X,{V}]$, and the last equality is by the definition of $m_{\mathbb{P},1D}$.    
Similarly, we can show Eqs. \eqref{sharp_2} and \eqref{sharp_3} hold for $m_{\mathbb{P},0D}$ and $m_{\mathbb{P},D}$.
\end{proof}

\begin{proof}[{Proof of Theorem \ref{theorem_sharp}}] 
For any $s\in\mathcal{S}$, we can get
\begin{align*} 
\mathbb{E}[YDs(Z)\mid Z]
&= 
\mathbb{E}[YD\mid Z]s(Z)
= 
\int_{{\mathcal{V}}}m_{\mathbb{P},1D}(v,Z)s(Z)dv,
\end{align*}
where the second equality is because we have shown that Eq. \eqref{sharp_1} holds for $m_{\mathbb{P},1D}$ in Theorem \ref{eq:theorem:sharpness}. Therefore, $m_{\mathbb{P}}$ satisfies Eq.\eqref{intro_contraint_Gamma_s0}. Similarly, we can show that $m_{\mathbb{P}}$ satisfies Eqs. \eqref{intro_contraint_Gamma_s1} and \eqref{intro_contraint_Lambda_s}. 
\end{proof}

\begin{proof}[{Proof of Corollary \ref{cor:bilinar_reprentation}}]The first part of the corollary, $m_\mathbb{P}\in\mathcal{M}_{\mathcal{S}}$,  is obvious by the definition of $m_{\mathbb{P},0D}$ and  $m_{\mathbb{P},1D}$ and by Theorem  \ref{theorem_sharp}.

Next, we show the second statement of this corollary that any $m\in\mathcal{M}_{\mathcal{S}}$ with $\mathcal{S}=\mathbf{L}^2(\mathcal{Z},\mathbb{R})$ satisfies Eqs. \eqref{sharp_1} to \eqref{sharp_3} and \eqref{cond:bilinear}. The proof is based on \cite{mogstad/santos/torgovitsky:2017}, but we modify it to fit the framework of this paper. 
Let $m$ be any element of $\mathcal{M}_{\mathcal{S}}$. Then, we know that $m$ satisfies Eq.  \eqref{cond:bilinear}.
Therefore, we only need to show that $m$ satisfies Eqs. \eqref{sharp_1} to \eqref{sharp_3}. First, we are going to show that $m$ satisfies Eq. \eqref{sharp_1} with $$
s(Z)=\mathbb{E}[YD\mid Z]-\mathbb{E}[m_1(V,X)m_D(V,Z)\mid Z].
$$
By rearranging terms in Eq. \eqref{intro_contraint_Gamma_s0}, 
we have 
$$
\mathbb{E}\Big[\Big(YD-\int_{{\mathcal{V}}}m_1(v,X)m_D(v,Z)dv\Big)s(Z)\Big]=0.$$
Therefore,  plugging $s(Z)$ into the above equation gives us
$$
\mathbb{E}\Big[\Big(YD-\mathbb{E}[m_1(V,X)m_D(V,Z)\mid Z]\Big)\Big(\mathbb{E}[YD\mid Z]-\mathbb{E}[m_1(V,X)m_D(V,Z)\mid Z]\Big)\Big]=0.
$$
Using the law of iterated expectations, we can equivalently write this equation as 
$$
\mathbb{E}\Big[\Big(\mathbb{E}[YD\mid Z]-\mathbb{E}[m_1(V,X)m_D(V,Z)\mid Z]\Big)^2\Big]=0,
$$
which implies Eq. \eqref{sharp_1}. We can prove Eq. \eqref{sharp_2} with $s(Z)=\mathbb{E}[Y(1-D)\mid Z]-\mathbb{E}[m_0(V,X)(1-m_D(V,Z))\mid Z]$ and Eq. \eqref{sharp_3} with $s(Z)=\mathbb{E}[D\mid Z]-\mathbb{E}[m_D(V,Z)\mid Z]$.   
\end{proof}

\begin{proof}[{Proof of Theorem \ref{theorem:main_result_relax}}] By Lemma \ref{lemma:mccormick}, any $m\in\mathcal{M}$ that satisfies Eq. \eqref{cond:bilinear} also satisfies Eq. \eqref{eq:mccormick_relax1_0} and \eqref{eq:mccormick_relax4_1}.
This theorem then follows from the definition of $\mathcal{M_S}$ and Corollary \ref{cor:bilinar_reprentation} that $m_{\mathbb{P}}\in\mathcal{M_S}$ for any $\mathcal{S}\subseteq\mathbf{L}^2(\mathcal{Z},\mathbb{R})$. 
\end{proof}

\begin{proof}[{Proof of Proposition \ref{prop_reduce_dim}}]
Given any $m\in\mathcal{M}$, we can see that $\tilde{m}$ defined in \eqref{prop_reduce_dim_fun} satisfies $\tilde{m}\in\tilde{\mathcal{M}}$. In what follows, we first show that if $m\in\mathcal{M}$ satisfies constraints in \eqref{cond:Gamma}, \eqref{eq:mccormick_bounds_m}, \eqref{eq:mccormick_relax1_0}, and \eqref{eq:mccormick_relax4_1}, then $\tilde{m}\in\tilde{\mathcal{M}}$ also satisfies those constraints. 
For any $s\in\mathcal{S}$, we can get
\begin{eqnarray}\label{red_dim_Gamma}
\Gamma_s(m)&=&\mathbb{E}\Big\{s(Z)\int_{\mathcal{V}} (m_{1D}(v,Z),m_{0D}(v,Z),m_D(v,Z))'dv\Big\}\nonumber\\
&=&\Gamma_s(\tilde{m}).
\end{eqnarray}
Besides, given that $m^L_d(V,X)$, $m^U_d(V,X)$, $m^L_{dD}(V,Z)$, and $m^U_{dD}(V,Z)$ do not depend on $V$ for all $d=0,1$, if $m$ satisfies constraints in \eqref{eq:mccormick_bounds_m}, \eqref{eq:mccormick_relax1_0}, and \eqref{eq:mccormick_relax4_1}, then integrating all these inequalities with respect to $v$ leads to the result that $\tilde m$ also satisfies those constraints in \eqref{eq:mccormick_bounds_m}, \eqref{eq:mccormick_relax1_0}, and \eqref{eq:mccormick_relax4_1}. Thus, if $m\in\mathcal{M}^r_{\mathcal{S}}$, then we know that $\tilde{m}\in\tilde{\mathcal{M}}^r_{\mathcal{S}}$.

Next, we show if $m^\star$ solves the minimization problem in \eqref{opt_relax} over $\mathcal{M}^r_{\mathcal{S}}$, then $\tilde{m}^\star$ obtained from plugging $m^\star$ into \eqref{prop_reduce_dim_fun} solves the minimization problem in \eqref{opt_relax} over $\tilde{\mathcal{M}}^r_{\mathcal{S}}$. Because $\omega_{dd'}^\star(V,Z)$ does not depend on $V$ for all $d,d'=0,1$, we can write it as $\omega_{dd'}^\star(Z)$. We can show that
\begin{eqnarray}\label{red_dim_Gammastar}
\Gamma^\star(m)&=&\mathbb{E}\Big\{\int_{{\mathcal{V}}} m_{0D}(v,Z)\omega_{00}^\star(Z)+(m_0(v,X)-m_{0D}(v,Z))\omega_{01}^\star(Z)dv
\nonumber\\&&~~~~+
\int_{{\mathcal{V}}} (m_1(v,X)-m_{1D}(v,Z))\omega_{10}^\star(Z)+m_{1D}(v,Z)\omega_{11}^\star(Z)dv\Big\}\nonumber\\
&=&\Gamma^\star(\tilde{m}).
\end{eqnarray}

Suppose $m^\star$ solves the minimization problem in \eqref{opt_relax} over $\mathcal{M}^r_{\mathcal{S}}$, and $\tilde{m}^\star$ is obtained from plugging $m^\star$ into \eqref{prop_reduce_dim_fun}. From \eqref{red_dim_Gammastar} we know that
$$\Gamma^\star(m^\star)=\Gamma^\star(\tilde{m}^\star).$$
Suppose there exists another $\tilde{m}^+\in\tilde{\mathcal{M}}^r_{\mathcal{S}}$ such that  $\Gamma^\star(\tilde{m}^+)<\Gamma^\star(\tilde{m}^\star)$. Given $\tilde{m}^+$, let us construct $m^+=(m_0^+,m_1^+,m_D^+,m_{0D}^+,m_{1D}^+)'\in\mathcal{M}$, where we set $m_d^+(v,x)=\tilde{m}_d^+(x)$, $m_D^+(v,z)=\tilde{m}_D^+(z)$, and $m_{dD}^+(v,z)=m_{dD}^+(z)$ for $d=0,1$ to be degenerate functions in $v$. Then, $m^+$ lies in $\mathcal{M}$ and it also satisfies the constraints in \eqref{cond:Gamma}, \eqref{eq:mccormick_bounds_m}, \eqref{eq:mccormick_relax1_0}, and \eqref{eq:mccormick_relax4_1}, i.e., $m^+\in\mathcal{M}^r_{\mathcal{S}}$. If so, then $\Gamma^\star(m^+)=\Gamma^\star(\tilde{m}^+)<\Gamma^\star(\tilde{m}^\star)=\Gamma^\star(m^\star)$, which contradicts the fact that $m^\star$ is the solution of the  minimization problem in \eqref{opt_relax} over $\mathcal{M}^r_{\mathcal{S}}$. Thus, we can conclude that $\tilde{m}^\star$ solves the minimization problem in \eqref{opt_relax} over $\tilde{\mathcal{M}}^r_{\mathcal{S}}$.
\end{proof}

\begin{proof}[{Proof of Theorem \ref{theorem:sharp_binary_P_m}}]First, we show that any $m_{P}\in\mathcal{R}$ with $P\in\mathcal{P}$ satisfies $m\in \mathcal{M_S}\mbox{ with }\mathcal{S}=\mathbf{L}^2(\mathcal{Z},\mathbb{R})$. By the definition of $m_{P,dD}$, Eq. \eqref{cond:bilinear} holds. By the definition of $\mathcal{P}$, we know that $P_{Y_D,D,Z}=\mathbb{P}_{Y,D,Z}$ for any $P\in\mathcal{P}$, and thus the following holds:
\begin{align*}
\mathbb{E}[YD\mid Z]=E_{P}[YD\mid Z]
&=
E_{P}[E_{P}[Y\mid D=1,V,Z]E_{P}[D\mid V,Z]\mid Z]
\\
&=
E_{P}[E_{P}[Y_1\mid D=1,V,Z]E_{P}[D\mid V,Z]\mid Z]
\\
&=
E_{P}[E_{P}[Y_1\mid V,X]E_{P}[D\mid V,Z]\mid Z]
\\
&=
E_{P}[m_{P,1}(V,X)m_{P,D}(V,Z)\mid Z]\\
&=
\int_{{\mathcal{V}}}m_{P,1D}(v,Z)dv,
\end{align*}
where the third line is because $E_P[Y_d\mid D,Z,{V}]=E_P[Y_d\mid X,{V}]$, the fourth line follows from the definition of $m_P$, and the last line is because $P_{V|Z}$ is a uniform distribution over $\mathcal{V}$. Then, for any $s\in\mathcal{S}=\mathbf{L}^2(\mathcal{Z},\mathbb{R})$, the above equation implies
\begin{align*}
\mathbb{E}[YDs(Z)]=\mathbb{E}\{s(Z)\mathbb{E}[YD\mid Z]\}
=&
\mathbb{E}\{s(Z)\int_{{\mathcal{V}}}m_{P,1D}(v,Z)dv\}.
\end{align*}
Applying similar arguments to $\mathbb{E}[Y(1-D)s(Z)]$ and $\mathbb{E}[Ds(Z)]$, we can prove that $m_{P}\in\mathcal{M_S}\mbox{ with }\mathcal{S}=\mathbf{L}^2(\mathcal{Z},\mathbb{R})$.

Next, we show the other direction holds. Let $m$ be any element of $\mathcal{R}$ satisfying \eqref{cond:bilinear}-\eqref{cond:Gamma} with $\mathcal{S}=\mathbf{L}^2(\mathcal{Z},\mathbb{R})$.
Given such a $m$, we construct a distribution $P$ as below.
Define the marginal distribution of $Z$, $P_Z$, to be the same as the marginal distribution $\mathbb{P}_{Z}$. 
Define the distribution of $V$ given $Z$, $P_{V|Z}$, as the uniform distribution over $\mathcal{V}$. 
Define the distribution of $D$ given $(V,Z)$ as  
$$
P(D=1\mid V,Z)=m_D(V,Z),
$$
for which the definition of $\mathcal{R}\subseteq\mathcal{M}$ guarantees the right hand side is between $0$ and $1$. 
For the distribution of $(Y_0,Y_1)\mid (D,Z,V)$, for $y_0,y_1\in\{0,1\}$, define  
\begin{align}\label{proof_sharp_jointY}
P((Y_0,Y_1)=(y_0,y_1)\mid D,Z,V)=m_0(V,X)^{y_0}(1-m_0(V,X))^{1-y_0}m_1(V,X)^{y_1}(1-m_1(V,X))^{1-y_1},
\end{align}
for which $\mathcal{R}\subseteq\mathcal{M}$ guarantees the right hand side is between $0$ and $1$. 
Now we are going to show (i) $P_{Y_D,D,Z}=\mathbb{P}_{Y,D,Z}$, (ii) $E_P[Y_d\mid D,Z,{V}]=E_P[Y_d\mid X,{V}]$ for each $d=0,1$, (iii)  $(m_{0},m_{1},m_{D})=(m_{P,0},m_{P,1},m_{P,D})$.
We can show (i) as follows. 
By Eqs. \eqref{cond:bilinear}-\eqref{cond:Gamma}, we have 
\begin{align}\label{proof_sharp1}
P(Y_D=1\mid D=1,Z)
&=
\frac{P(Y_1=1,D=1\mid Z)}{P(D=1\mid Z)}\nonumber\\
&=
\frac{E_P[P(Y_1=1,D=1\mid V,Z)\mid Z]}{E_P[P(D=1\mid V,Z)\mid Z]}\nonumber\\
&=
\frac{\int_{\mathcal{V}} P(Y_1=1,D=1\mid V=v,Z)dv}{\int_{\mathcal{V}} P(D=1\mid V=v,Z)dv}\nonumber\\
&=
\frac{\int_{\mathcal{V}} m_1(v,X)m_D(v,Z)dv}{\int_{\mathcal{V}} m_D(v,Z)dv}\nonumber\\
&=
\frac{\mathbb{E}[YD\mid Z]}{\mathbb{E}[D\mid Z]}\nonumber\\
&=
\mathbb{P}(Y=1\mid D=1,Z),
\end{align}
where the fourth equality is by the construction of $P$ that $P(D=1\mid V,Z)=m_D(V,Z)$ and $P(Y_1=y_1\mid D,Z,V)=m_1(V,X)^{y_1}(1-m_1(V,X))^{1-y_1}$ obtained from \eqref{proof_sharp_jointY}, the fifth equality is because $m$ satisfies Eqs. \eqref{cond:bilinear}-\eqref{cond:Gamma} with $\mathcal{S}=\mathbf{L}^2(\mathcal{Z},\mathbb{R})$ which implies $m$ also satisfies Eqs. \eqref{sharp_1} to \eqref{sharp_3} by Corollary \ref{cor:bilinar_reprentation}. Similarly, we can show $P_{Y_D\mid D=0,Z}=\mathbb{P}_{Y\mid D=0,Z}$. In addition, the proof for the denominator in \eqref{proof_sharp1} shows that $P(D=1\mid Z)=\mathbb{P}(D=1\mid Z)$. Since $P_Z=\mathbb{P}_Z$ by the construction of $P$, we have $P_{Y_D,D,Z}=\mathbb{P}_{Y,D,Z}$. 

For (ii), for $d=0,1$, by the definition of the distribution $P_{Y_d\mid (D,V,Z)}$ in \eqref{proof_sharp_jointY}, we have 
\begin{align}\label{proof_middlestep1}
    E_P[Y_d\mid D,Z,{V}]=m_d(V,X).
\end{align}
Since $E_P[Y_d\mid D,Z,{V}]$ only depends on $(V,X)$, it yields from the law of iterated expectation that
\begin{align}\label{proof_middlestep2}
E_P[Y_d\mid X,{V}]=E_P\{E_P[Y_d\mid D,Z,{V}]\mid X,{V}\}=E_P[Y_d\mid D,Z,{V}].
\end{align}
For (iii), for $d=0,1$, we can see that  
$$
m_{P,d}(V,X)=E_P[Y_d\mid X,{V}]=E_P[Y_d\mid D,Z,{V}]=m_{d}(V,X),
$$
where the first equality is by definition, the second equality is by \eqref{proof_middlestep2}, and the last equality is based on \eqref{proof_middlestep1}. Besides, by the construction of $P$, we have 
$m_D(V,Z)=P(D=1\mid V,Z)=m_{P,D}(V,Z).$
\end{proof}

\begin{proof}[{Proof of Theorem \ref{theorem:sharp_binary_determonistic_monot}}]
Under the assumption of this theorem, i.e., $m_D^L(V,Z)=m_D^U(V,Z)=m_D(V,Z)$, inequality constraints in \eqref{eq:mccormick_relax1_0} and \eqref{eq:mccormick_relax4_1} become equality constraints in \eqref{cond:bilinear}. 
Therefore, the conclusion of this theorem holds. 
\end{proof}

\begin{proof}[{Proof of Corollary \ref{coro_equiv_CvR_HV}}]We show the results for $\overline\beta^\star$. Similar arguments can be applied to show the results for $\underline\beta^\star$. Denote $p(Z)=\mathbb{P}(D=1\mid Z)$. When there are no covariates, $\Gamma^\star(m)$ for the ATE can be simplified to 
$$
\Gamma^\star(m)=\int_\mathcal{V}m_1(v)dv-\int_\mathcal{V}m_0(v)dv.
$$
Consider the constraints that characterize $\mathcal{M}^r_\mathcal{S}$. Since we set $m_D^L(V,Z)=m_D^U(V,Z)=1\{p(Z)\geq V\}$, it implies that we restrict $m_D(V,Z)=1\{p(Z)\geq V\}$ as well for any $m\in\mathcal{M}^r_\mathcal{S}$. Because we set $m_d^L(V,Z)=0$ and $m_d^U(V,Z)=1$ for $d=0,1$, the constraints in \eqref{eq:mccormick_bounds_m}, \eqref{eq:mccormick_relax1_0}, and \eqref{eq:mccormick_relax4_1} become
\begin{align}\label{const_cvr1}
&0\leq m_d(v)\leq 1,\text{ for }d=0,1\\\label{const_cvr2}
&m_D(v,z)=1\{p(z)\geq v\},\\\label{const_cvr3}
&m_{0D}(v,z)=m_0(v)(1-m_D(v,z))\text{ and }m_{1D}(v,z)=m_1(v)m_D(v,z),
\end{align}
for $\forall (v,z)\in{\mathcal{V}}\times \mathcal{Z}$.

Given $\mathcal{S}=\mathbf{L}^2(\mathcal{Z},\mathbb{R})$ and based on \eqref{const_cvr3}, we can see that the constraints in \eqref{cond:Gamma} are equivalent to 
\begin{align}\label{const_pz1}
\mathbb{E}[YD\mid Z]=&\int_{\mathcal{V}}m_1(v)1\{p(Z)\geq v\}dv=\int_{0}^{p(Z)}m_1(v)dv,\\\label{const_pz2}
\mathbb{E}[Y(1-D)\mid Z]=&\int_{\mathcal{V}}m_0(v)1\{p(Z)<v\}dv=\int_{p(Z)}^1m_0(v)dv,\\\label{const_pz3}
\mathbb{E}[D\mid Z]=&\int_\mathcal{V}1\{p(Z)\geq v\}dv=p(Z).
\end{align}
It is easy to see that the constraints in  \eqref{const_cvr2},  \eqref{const_cvr3}, and \eqref{const_pz3} impose no restrictions on $\int_\mathcal{V}m_1(v)dv$ and $\int_\mathcal{V}m_0(v)dv$. Therefore, $\overline\beta^\star$ can be rewritten to be
\begin{equation*}
\overline\beta^\star=\sup_{m\in\mathcal{M}}\left\{\int_\mathcal{V}m_1(v)dv-\int_\mathcal{V}m_0(v)dv\right\},~~m\text{ subject to \eqref{const_cvr1}, \eqref{const_pz1}, and \eqref{const_pz2}.}
\end{equation*}
Denote the upper bound of ATE of \cite{heckman/vylacil:2001:book} as
\begin{align*}
\overline\beta=&p^U\mathbb{E}[Y\mid p(Z)=p^U,D=1]+(1-p^U)-(1-p^L)\mathbb{E}[Y\mid p(Z)=p^L,D=0].
\end{align*}
For any $p(z)\in[0,1]$ and $d=0,1$, the following equation holds:
\begin{align*}
    \int_\mathcal{V}m_d(v)dv=\int_0^{p(z)}m_d(v)dv+\int_{p(z)}^1m_d(v)dv.
\end{align*}
Therefore, for any $m_0,m_1$ that satisfy \eqref{const_cvr1}, \eqref{const_pz1}, and \eqref{const_pz2}, we can obtain 
\begin{align*}
\int_\mathcal{V}m_0(v)dv=&\int_0^{p^L}m_0(v)dv+\int_{p^L}^1m_0(v)dv\geq \mathbb{E}[Y(1-D)\mid Z=\underline{z}],\\
\int_\mathcal{V}m_1(v)dv=&\int_0^{p^U}m_1(v)dv+\int_{p^U}^1m_1(v)dv\leq (1-p^U)+\mathbb{E}[YD\mid Z=\overline{z}],
\end{align*}
which further imply
\begin{align*}
\int_\mathcal{V}m_1(v)dv-\int_\mathcal{V}m_0(v)dv\leq&(1-p^U)+\mathbb{E}[YD\mid Z=\overline{z}]-\mathbb{E}[Y(1-D)\mid Z=\underline{z}]\\
=&(1-p^U)+\mathbb{E}[Y\mid Z=\overline{z},D=1]p^U-\mathbb{E}[Y\mid Z=\underline{z},D=0](1-p^L)\\
=&\overline\beta.
\end{align*}
Because the above inequality holds for any $m_0,m_1$ that satisfy \eqref{const_cvr1}, \eqref{const_pz1}, and \eqref{const_pz2}, it also holds for the supremum of $\int_\mathcal{V}m_1(v)dv-\int_\mathcal{V}m_0(v)dv$. Thus, $\overline\beta^\star\leq \overline\beta$.

Next, we show that $\overline\beta^\star\geq \overline\beta$ by showing that there exist $m'_0,m'_1$ that satisfy \eqref{const_cvr1}, \eqref{const_pz1}, and \eqref{const_pz2} such that $\overline\beta=\int_\mathcal{V}m'_1(v)dv-\int_\mathcal{V}m'_0(v)dv$. 
For some $y_0^*,y_1^*\in[y^L,y^U]=[0,1]$, we define 
\begin{equation*}\begin{aligned}
m'_{0}(V)
&=
\begin{cases}
\mathbb{E}[Y_0\mid V]&\mbox{ if }V> p^L\\
y_0^*&\mbox{ otherwise},
\end{cases}~~\text{ and }~~
m'_{1}(V)
=
\begin{cases}
\mathbb{E}[Y_1\mid V]&\mbox{ if }V\leq p^U\\
y_1^*&\mbox{ otherwise}.
\end{cases}
\end{aligned}\end{equation*}
Since $Y$ is binary and $y_0^*,y_1^*\in[0,1]$, we can see that $m'_0,m'_1$ satisfy \eqref{const_cvr1}. Because $p(Z)\geq p^L$ holds for any $Z\in\mathcal{Z}$, by the law of iterated expectation, the uniform distribution of $V$ given $Z$, and the constraint $\mathbb{E}[D\mid Z,V]=1[V\leq p(Z)]$, we have
\begin{align}
\mathbb{E}[Y(1-D)\mid Z]=&\mathbb{E}\{\mathbb{E}[Y(1-D)\mid Z,V]\mid Z\}\nonumber\\
=&\mathbb{E}\{\mathbb{E}[Y_0\mid Z,V,D=0]\mathbb{E}[1-D\mid Z,V]\mid Z\}\nonumber\\
=&\mathbb{E}\{\mathbb{E}[Y_0\mid V]\mathbb{E}[1-D\mid Z,V]\mid Z\}\nonumber\\
=&\mathbb{E}\{\mathbb{E}[Y_0\mid V]1[V>p(Z)]\mid Z\}\nonumber\\
=&\int_0^1\mathbb{E}[Y_0\mid V=v]1[v>p(Z)]dv\nonumber\\
=&\int_0^1m'_0(v)1[v>p(Z)]dv.
\end{align}
Similarly, we can show that $\mathbb{E}[YD\mid Z]=\int_0^1m'_1(v)1[p(Z)\geq v]dv$. Thus, $m'_0,m'_1$ also satisfy \eqref{const_pz1} and \eqref{const_pz2}.
Moreover,
by the definition $m'_0$ and Assumption \ref{assumption_DGP}, we have
\begin{align*}
    \mathbb{E}[m'_{0}(V)]
=&
\mathbb{E}\big\{\mathbb{E}[Y_0\mid V]1[V> p^L]+
y_0^*1[V\leq p^L]\big\}\\
=&\mathbb{E}\big\{\mathbb{E}[Y_0\mid V,p(Z)=p^L,D=0]\mathbb{E}[1-D\mid V,p(Z)=p^L]\big\}+y_0^*p^L\\
=&\mathbb{E}\left\{\mathbb{E}[Y(1-D)\mid V,p(Z)=p^L]\right\}+y_0^*p^L\\
=&\mathbb{E}[Y(1-D)\mid p(Z)=p^L]+y_0^*p^L\\
=&\mathbb{E}[Y\mid p(Z)=p^L,D=0](1-p^L)+y_0^*p^L.
\end{align*}
Similar, we can show that
$
\mathbb{E}[  m'_{1}(V)]  =p^U\mathbb{E}[Y\mid D=1,p(Z)=p^U]+(1-p^U)y_1^*.$ If we set $y_0^*=0$ and $y_1^*=1$, then we get
\begin{align*}
  \mathbb{E}[  m'_{1}(V)]-  \mathbb{E}[m'_{0}(V)]
=&p^U\mathbb{E}[Y\mid D=1,p(Z)=p^U]+(1-p^U)-\mathbb{E}[Y\mid p(Z)=p^L,D=0](1-p^L),
\end{align*}
which is the same as $\overline{\beta}$. 
\end{proof}

Below, we present Lemmas \ref{lemma_Gamma_simplified}, \ref{lemma_outerset_m0m1_cvr}, and \ref{lemma_mtilde_sharp} that will be used in the proof of Theorem \ref{coro_equav_MST_CVR}.
\begin{lemma}\label{lemma_Gamma_simplified}Suppose $\omega_{dd'}^\star(V,Z)=\omega_{dd'}^\star(Z)$ does not depend on $V$ for $d,d'=0,1$, and we set $\mathcal{S}=\mathbf{L}^2(\mathcal{Z},\mathbb{R})$. For any generic functions $m_0,m_1:\mathcal{V}\times \mathcal{X}\mapsto\mathbb{R}$, define a map
$$\Sigma^\star(m_0,m_1)=\mathbb{E}\left[c^\star_{01}(X)\int_\mathcal{V}m_0(v,X)dv+c^\star_{10}(X)\int_\mathcal{V}m_1(v,X)dv\right]+C^\star,
$$
where $c^\star_{01}(X)=\mathbb{E}[\omega_{01}^\star(Z)\mid X]$, $c^\star_{10}(X)=\mathbb{E}[\omega_{10}^\star(Z)\mid X]$, and $C^\star$ is a constant given by
\begin{align*}
C^\star=&\mathbb{E}\left\{\mathbb{E}[Y(1-D)\mid Z](\omega^\star_{00}(Z)-\omega^\star_{01}(Z))\right\}+\mathbb{E}\left\{\mathbb{E}[YD\mid Z](\omega^\star_{11}(Z)-\omega^\star_{10}(Z))\right\}.
\end{align*}
Then, for any $m$ that satisfies the constraints in Eq. \eqref{cond:Gamma}, we have
$\Gamma^\star(m)=\Sigma^\star(m_0,m_1)$, which depends only on  $\int_{\mathcal{V}} m_0(v,x)dv$ and $\int_{\mathcal{V}}m_1(v,x)dv$ for all $x\in\mathcal{X}$.
\end{lemma}
\begin{proof}[{Proof of Lemma \ref{lemma_Gamma_simplified}}]Recall that  our target parameter is
\begin{align*}\Gamma^\star(m)&=\mathbb{E}\left[\int_{{\mathcal{V}}} \Big[m_{0D}(v,Z)\omega_{00}^\star(v,Z)+(m_0(v,X)-m_{0D}(v,Z))\omega_{01}^\star(v,Z)\Big]dv\right]
\nonumber\\&\quad+
\mathbb{E}\left[\int_{{\mathcal{V}}} \Big[(m_1(v,X)-m_{1D}(v,Z))\omega_{10}^\star(v,Z)+m_{1D}(v,Z)\omega_{11}^\star(v,Z)\Big]dv\right].
\end{align*}
    Given the assumption that $\mathcal{S}=\mathbf{L}^2(\mathcal{Z},\mathbb{R})$, the constraints in \eqref{cond:Gamma} imply the following constraints:
\begin{equation}\label{constraints_data_conditional}
\begin{aligned}
\mathbb{E}[YD\mid Z]&=\int_{{\mathcal{V}}}m_{1D}(v,Z)dv,\\
\mathbb{E}[Y(1-D)\mid Z]&=\int_{{\mathcal{V}}}m_{0D}(v,Z)dv,\\
\mathbb{E}[D\mid Z]&=\int_{{\mathcal{V}}}m_D(v,Z)dv.
\end{aligned}
\end{equation}
Since $\omega_{dd'}^\star(V,Z)$ does not depend on $V$, we can denote $\omega_{dd'}^\star(V,Z)=\omega_{dd'}^\star(Z)$. Then, together with the constraints in \eqref{constraints_data_conditional}, we can obtain
\begin{align}\label{expression_Gamma*}
\Gamma^\star(m)&=\mathbb{E}\left\{ \mathbb{E}[Y(1-D)\mid Z](\omega_{00}^\star(Z)-\omega_{01}^\star(Z))+\int_{{\mathcal{V}}}m_0(v,X)dv\mathbb{E}[\omega_{01}^\star(Z)\mid X]\right\}
\nonumber\\&\quad+
\mathbb{E}\left\{\int_{{\mathcal{V}}}m_1(v,X)dv\mathbb{E}[\omega_{10}^\star(Z)\mid X]+\mathbb{E}[YD\mid Z](\omega_{11}^\star(Z)-\omega_{10}^\star(Z))\right\}\nonumber\\
&=\mathbb{E}\left\{c^\star_{01}(X)\int_{\mathcal{V}}m_0(v,X)dv+c^\star_{10}(X)\int_{\mathcal{V}}m_1(v,X)dv\right\}+C^\star.
\end{align}
\end{proof}

\begin{lemma}\label{lemma_outerset_m0m1_cvr}Suppose we set $\mathcal{S}=\mathbf{L}^2(\mathcal{Z},\mathbb{R})$ and $m_d^L(v,x)=y^L$, $m_d^U(v,x)=y^U$, $m_D^L(v,z)=0$, and $m_D^U(v,z)=1$  for $d=0,1$ in Lemma \ref{lemma:mccormick}. Then, for any $x\in\mathcal{X}$ and any $m\in\mathcal{M}$ that satisfies the constraints in \eqref{cond:Gamma},  
\begin{align*}
&M_1^L(x)\leq \mathbb{E}[m_{1}(V,x)]\leq M_1^U(x),\text{ where }\\
M_1^L(x)=\sup_{z_0\in \mathcal{Z}_0}&\big\{\mathbb{E}[YD\mid X=x,Z_0=z_0]+y^L\mathbb{E}[1-D\mid X=x,Z_0=z_0]\big\},\nonumber\\
M_1^U(x)= \inf_{z_0\in \mathcal{Z}_0}&\big\{\mathbb{E}[YD\mid X=x,Z_0=z_0]+y^U\mathbb{E}[1-D\mid X=x,Z_0=z_0]\big\},\nonumber
\end{align*}
and
\begin{align*}
&M_0^L(x)\leq \mathbb{E}[m_{0}(V,x)]\leq M_0^U(x),\text{ where }\\
M_0^L(x)=\sup_{z_0\in \mathcal{Z}_0}&\big\{\mathbb{E}[Y(1-D)\mid X=x,Z_0=z_0]+y^L\mathbb{E}[D\mid X=x,Z_0=z_0]\big\},\nonumber\\
M_0^U(x)= \inf_{z_0\in \mathcal{Z}_0}&\big\{\mathbb{E}[Y(1-D)\mid X=x,Z_0=z_0]+y^U\mathbb{E}[D\mid X=x,Z_0=z_0]\big\}.\nonumber
\end{align*}
\end{lemma}
\begin{proof}[{Proof of Lemma \ref{lemma_outerset_m0m1_cvr}}]For $\mathbb{E}[m_{1}(V,x)]=\int_{{\mathcal{V}}}m_1(v,x)dv$, given the fact that we set $m_d^L(v,x)=y^L$, $m_d^U(v,x)=y^U$, $m_D^L(v,z)=0$, and $m_D^U(v,z)=1$, the constraint \eqref{eq:mccormick_relax4_1} of Lemma \ref{lemma:mccormick} implies that, for any $(v,z)\in \mathcal{V}\times\mathcal{Z}$, the following inequalities hold
\begin{align}\begin{cases}
m_{1D}(v,z)\geq y^Lm_{D}(v,z)\\
m_{1D}(v,z)\geq y^Um_{D}(v,z)+m_{1}(v,x)-y^U\\
m_{1D}(v,z)\leq y^Um_{D}(v,z)\\
m_{1D}(v,z)\leq y^Lm_{D}(v,z)+m_{1}(v,x)-y^L.
\end{cases}
\end{align}
Integrate both sides of the above inequalities over $V$ given $Z=z$, we can obtain
\begin{align}\label{manski_em1D_sharp}\begin{cases}
\mathbb{E}[m_{1D}(V,z)]\geq y^L\mathbb{E}[m_{D}(V,z)]\\
\mathbb{E}[m_{1D}(V,z)]\geq y^U\mathbb{E}[m_{D}(V,z)]+\mathbb{E}[m_{1}(V,x)]-y^U\\
\mathbb{E}[m_{1D}(V,z)]\leq y^U\mathbb{E}[m_{D}(V,z)]\\
\mathbb{E}[m_{1D}(V,z)]\leq y^L\mathbb{E}[m_{D}(V,z)]+\mathbb{E}[m_{1}(V,x)]-y^L.
\end{cases}
\end{align}
Recall that the constraints in \eqref{cond:Gamma} and $\mathcal{S}=\mathbf{L}^2(\mathcal{Z},\mathbb{R})$ imply $\mathbb{E}[m_{1D}(V,z)]=\mathbb{E}[YD\mid Z=z]$ and $\mathbb{E}[m_{D}(V,z)]=\mathbb{E}[D\mid Z=z]$. Then, the first and third inequalities in \eqref{manski_em1D_sharp} imply
$$y^L\mathbb{E}[D\mid Z=z]\leq\mathbb{E}[YD\mid Z=z]\leq y^U\mathbb{E}[D\mid Z=z],$$
which holds trivially and do not impose any restrictions on $\mathbb{E}[m_1(V,z)]$. While,  the second and fourth inequalities in \eqref{manski_em1D_sharp} imply
\begin{align*}
\mathbb{E}[YD\mid Z=z]+y^L\mathbb{E}[1-D\mid Z=z]\leq& \mathbb{E}[m_{1}(V,x)]\leq \mathbb{E}[YD\mid Z=z]+y^U\mathbb{E}[1-D\mid Z=z].
\end{align*}
Recall we have $z=(x,z_0)$. Since $\mathbb{E}[m_{1}(V,x)]$ does not depend on $z_0$, we can get the following intersection bounds for $\mathbb{E}[m_{1}(V,x)]$:
\begin{align*}
\sup_{z_0\in \mathcal{Z}_0}\big\{\mathbb{E}[YD\mid X=x,Z_0=z_0]+y^L\mathbb{E}[1-D\mid X=x,Z_0=z_0]\big\}&\leq \mathbb{E}[m_{1}(V,x)]\nonumber\\
\leq \inf_{z_0\in \mathcal{Z}_0}\big\{\mathbb{E}[YD\mid X=x,Z_0=z_0]+y^U&\mathbb{E}[1-D\mid X=x,Z_0=z_0]\big\}.
\end{align*}
Using the constraint \eqref{eq:mccormick_relax1_0} of Lemma \ref{lemma:mccormick} and applying similar arguments, we can obtain the intersection bounds for $\mathbb{E}[m_{0}(V,x)]$ as below
\begin{align*}
\sup_{z_0\in \mathcal{Z}_0}\big\{\mathbb{E}[Y(1-D)\mid X=x,Z_0=z_0]+y^L\mathbb{E}[D\mid X=x,Z_0=z_0]\big\}&\leq \mathbb{E}[m_{0}(V,x)]\nonumber\\
\leq \inf_{z_0\in \mathcal{Z}_0}\big\{\mathbb{E}[Y(1-D)\mid X=x,Z_0=z_0]+y^U&\mathbb{E}[D\mid X=x,Z_0=z_0]\big\}.
\end{align*}
\end{proof}

\begin{lemma}\label{lemma_mtilde_sharp}Suppose the true DGP satisfies Assumption \ref{assumption_DGP} and the classical
monotonicity in \eqref{model_MST} so that $V$ is a scalar unobserved heterogeneity that follows a uniform distribution over $[0,1]$ conditional on $Z$. Denote $p(Z)=\mathbb{E}[D\mid Z]$. Let
\begin{align}\label{extreme_PS}
p^L_x=\inf_{z_0\in\mathcal{Z}_0}p(x,z_0)\text{ and }p^U_x=\sup_{z_0\in\mathcal{Z}_0}p(x,z_0).
\end{align}
Let $M_d^L(x)$ and $M_d^U(x)$ be the objects defined in Lemma \ref{lemma_outerset_m0m1_cvr}. 
Then, for every $s_0(x)\in [M_0^L(x),M_0^U(x)]$ and every $s_1(x)\in [M_1^L(x),M_1^U(x)]$, 
there exist $y_0^*(x),y_1^*(x)\in[y^L,y^U]$  such that 
$$\mathbb{E}[m_d(V,x)]=s_d(x)$$ every $d=0,1$, where 
\begin{equation}\label{construct_m_tilde_single_dim}\begin{aligned}
m_{0}(V,x)
&=
\begin{cases}
\mathbb{E}[Y_0\mid V,X=x]&\mbox{ if }V> p^L_x\\
y_0^*(x)&\mbox{ otherwise},
\end{cases}
\\
m_{1}(V,x)
&=
\begin{cases}
\mathbb{E}[Y_1\mid V,X=x]&\mbox{ if }V\leq p^U_x\\
y_1^*(x)&\mbox{ otherwise}.
\end{cases}
\end{aligned}
\end{equation}
\end{lemma}


\begin{proof}[{Proof of Lemma \ref{lemma_mtilde_sharp}}]
By the proof of Theorem 2 in \cite{heckman/vylacil:2001:book}, we know that under the threshold-crossing structure in \eqref{model_MST}, $M_d^L(x)$ and $M_d^U(x)$ for $d=0,1$ are equal to
\begin{equation}\label{HV_m0bounds}\begin{aligned}
M_0^L(x)=&(1-p^L_x)\mathbb{E}[Y\mid D=0,p(Z)=p^L_x,X=x]+y^Lp^L_x\\
M_0^U(x)=&(1-p^L_x)\mathbb{E}[Y\mid D=0,p(Z)=p^L_x,X=x]+y^Up^L_x,
\end{aligned}\end{equation}
and
\begin{equation}\label{HV_m1bounds}\begin{aligned}
M_1^L(x)=&p^U_x\mathbb{E}[Y\mid D=1,p(Z)=p^U_x,X=x]+(1-p^U_x)y^L\\
M_1^U(x)=&p^U_x\mathbb{E}[Y\mid D=1,p(Z)=p^U_x,X=x]+(1-p^U_x)y^U.
\end{aligned}\end{equation}
For any $s_0(x)\in [M_0^L(x),M_0^U(x)]$ and $s_1(x)\in [M_1^L(x),M_1^U(x)]$, we can express them as below
\begin{align*}
    s_0(x)=&(1-p^L_x)\mathbb{E}[Y\mid D=0,p(Z)=p^L_x,X=x]+p^L_xy_0^*(x),\\
    s_1(x)=&p^U_x\mathbb{E}[Y\mid D=1,p(Z)=p^U_x,X=x]+(1-p^U_x)y_1^*(x),
\end{align*}
for some $y_0^*(x),y_1^*(x)\in[y^L,y^U]$. Under the classical
monotonicity, we have $\mathbb{E}[D\mid V,Z]=1[V\leq p(Z)]$. In addition,
by the definition of $m_{0}(V,x)$ given in \eqref{construct_m_tilde_single_dim}, for any given $x$, we have
\begin{align*}
    \mathbb{E}[m_{0}(V,x)]
=&
\mathbb{E}\big\{\mathbb{E}[Y_0\mid V,X=x]1[V> p^L_x]+
y_0^*(x)1[V\leq p^L_x]\big\}\\
=&\mathbb{E}\big\{\mathbb{E}[Y_0\mid V,X=x,p(Z)=p^L_x,D=0]\mathbb{E}[1-D\mid V,X=x,p(Z)=p^L_x]\big\}+y_0^*(x)p^L_x\\
=&\mathbb{E}\left\{\mathbb{E}[Y(1-D)\mid V,X=x,p(Z)=p^L_x]\right\}+y_0^*(x)p^L_x\\
=&\mathbb{E}[Y(1-D)\mid X=x,p(Z)=p^L_x]+y_0^*(x)p^L_x\\
=&\mathbb{E}[Y\mid X=x,p(Z)=p^L_x,D=0](1-p^L_x)+y_0^*(x)p^L_x\\
=&s_0(x),
\end{align*}
where the second equality is due to Assumption \ref{assumption_DGP} that $Y_0$ and $(D,Z)$ are mean-independent given $(V,X)$ and the fact that $\mathbb{E}[D\mid V,Z]=1[V\leq p(Z)]$ under the threshold-crossing structure, the fourth equality is by the law of iterated expectation given that $V$ is independent of $Z$, and the last equality is by the definition of $s_0(x)$. 
Similarly, we can show that
\begin{align*}
 \mathbb{E}[  m_{1}(V,x)]  =&p^U_x\mathbb{E}[Y\mid D=1,p(Z)=p^U_x,X=x]+(1-p^U_x)y_1^*(x)=s_1(x).
\end{align*}
\end{proof}

\begin{proof}[{Proof of Theorem \ref{coro_equav_MST_CVR}}] (i) First, we consider the case with no shape restrictions. Under Assumption \ref{assumption_DGP} and the threshold-crossing structure in \eqref{model_MST}, both $\mathcal{M}^r_\mathcal{S}$ and $\tilde{\mathcal{M}}_{01,\tilde{\mathcal{S}}}$ are non-empty, as $m_{\mathbb{P}}\in\mathcal{M}^r_\mathcal{S}$ and $(m_{\mathbb{P},0},m_{\mathbb{P},1})\in\mathcal{M}^r_{01,\mathcal{S}}$. Because $\tilde{\Gamma}^\star(\tilde{m}_0,\tilde{m}_1)$ and the constraints that characterize $\tilde{\mathcal{M}}_{01,\tilde{\mathcal{S}}}$ are all linear in $(\tilde{m}_0,\tilde{m}_1)$, the set $\{\tilde{\Gamma}^\star(\tilde{m}_0,\tilde{m}_1):~(\tilde{m}_0,\tilde{m}_1)\in\tilde{\mathcal{M}}_{01,\tilde{\mathcal{S}}}\}$ is the same as an interval with two ending points 
$$\inf_{(\tilde{m}_0,\tilde{m}_1)\in \tilde{\mathcal{M}}_{01,\tilde{\mathcal{S}}}}\tilde{\Gamma}^\star(\tilde{m}_0,\tilde{m}_1)\text{ and }\sup_{(\tilde{m}_0,\tilde{m}_1)\in \tilde{\mathcal{M}}_{01,\tilde{\mathcal{S}}}}\tilde{\Gamma}^\star(\tilde{m}_0,\tilde{m}_1).$$
Recall that $\left\{\Gamma^\star(m): m\in\mathcal{M}^r_\mathcal{S}\right\}$ is also equal to an interval with two ending points 
$$\inf_{m\in \mathcal{M}^r_\mathcal{S}}\Gamma^\star(m)\text{ and }\sup_{m\in \mathcal{M}^r_\mathcal{S}}\Gamma^\star(m).$$
Below, we show that $\sup_{m\in \mathcal{M}^r_\mathcal{S}}\Gamma^\star(m)=\sup_{(\tilde{m}_0,\tilde{m}_1)\in \tilde{\mathcal{M}}_{01,\tilde{\mathcal{S}}}}\tilde{\Gamma}^\star(\tilde{m}_0,\tilde{m}_1)$. Similar arguments can be applied to show that the same result holds for the infimum.

First, we show that for any $ (\tilde{m}_0,\tilde{m}_1)\in\tilde{\mathcal{M}}_{01,\tilde{\mathcal{S}}}$, $\tilde{\Gamma}^\star(\tilde{m}_0,\tilde{m}_1)$ depends only on $\int_0^1 \tilde{m}_0(u,x)du$ and $\int_0^1\tilde{m}_1(u,x)du$ for all $x$ in its support; and for any $m\in\mathcal{M}^r_\mathcal{S}$, $\Gamma^\star(m)$ depends only on $\int_{{\mathcal{V}}}m_0(v,x)dv$ and $\int_{{\mathcal{V}}}m_1(v,x)dv$ for all $x$ in its support. Recall that under condition (ii), $\omega_{dd'}^\star(V,Z)=\omega_{dd'}^\star(Z)$ does not depend on $V$. Then, under condition (iii), the target parameter of MST can be expressed as 
\begin{align}
\tilde{\Gamma}^\star(\tilde{m}_0,\tilde{m}_1)=&\mathbb{E}\left[\int_0^1\tilde{m}_0(u,X)\tau^\star_0(u,Z)du\right]+\mathbb{E}\left[\int_0^1\tilde{m}_1(u,X)\tau^\star_1(u,Z)du\right]\nonumber\\
=&\mathbb{E}\left[\int_0^1\tilde{m}_0(u,X)\Big[1[u>p(Z)]\omega^\star_{00}(Z)+1[u\leq p(Z)]\omega^\star_{01}(Z)\Big]du\right]\nonumber\\
&\quad+\mathbb{E}\left[\int_0^1\tilde{m}_1(u,X)\Big[1[u>p(Z)]\omega^\star_{10}(Z)+1[u\leq p(Z)]\omega^\star_{11}(Z)\Big]du\right]\nonumber\\
=&\mathbb{E}\left[\int_0^1\tilde{m}_0(u,X)\omega^\star_{01}(Z)du+\int_0^1\tilde{m}_0(u,X)1[u>p(Z)](\omega^\star_{00}(Z)-\omega^\star_{01}(Z))du\right]\nonumber\\
&\quad+\mathbb{E}\left[\int_0^1\tilde{m}_1(u,X)\omega^\star_{01}(Z)du+\int_0^1\tilde{m}_1(u,X)1[u\leq p(Z)](\omega^\star_{11}(Z)-\omega^\star_{10}(Z))du\right].
\end{align}
Note that $\tilde{\mathcal{S}}=\mathbf{L}^2(\{0,1\}\times\mathcal{Z},\mathbb{R})$ implies $\tilde{\mathcal{S}}=\{\tilde{s}(d,z):~\tilde{s}(d,z)=1[d=d']s(z)\text{ for }d'\in\{0,1\},~s\in\mathbf{L}^2(\mathcal{Z},\mathbb{R})\}$. For $ (\tilde{m}_0,\tilde{m}_1)\in\tilde{\mathcal{M}}_{01,\tilde{\mathcal{S}}}$, based on the constraints in \eqref{MST_restrictions} and condition (i) that $\tilde{\mathcal{S}}=\mathbf{L}^2(\{0,1\}\times\mathcal{Z},\mathbb{R})$, we know that $\tilde{m}_0$ and $\tilde{m}_1$ satisfy
\begin{align}\label{HV_m0restriction}
    \int_0^1\tilde{m}_0(u,X)1[u>p(Z)]du=&\mathbb{E}[Y(1-D)\mid Z],\\\label{HV_m1restriction}
    \int_0^1\tilde{m}_1(u,X)1[u\leq p(Z)]du=&\mathbb{E}[YD\mid Z].
\end{align}
Plugging \eqref{HV_m0restriction} and \eqref{HV_m1restriction} into $\tilde{\Gamma}^\star(\tilde{m}_0,\tilde{m}_1)$, we can obtain
\begin{align}\label{equv_Gamma*tilde}
\tilde{\Gamma}^\star(\tilde{m}_0,\tilde{m}_1)=&
\mathbb{E}\left[\int_0^1\tilde{m}_0(u,X)\mathbb{E}[\omega^\star_{01}(Z)\mid X]du+\mathbb{E}[Y(1-D) \mid Z](\omega^\star_{00}(Z)-\omega^\star_{01}(Z))\right]\nonumber\\
&\quad+\mathbb{E}\left[\int_0^1\tilde{m}_1(u,X)\mathbb{E}[\omega^\star_{01}(Z)\mid X]du+\mathbb{E}[YD\mid Z](\omega^\star_{11}(Z)-\omega^\star_{10}(Z))\right]\nonumber\\
=&\mathbb{E}\left[c^\star_{01}(X)\int_0^1\tilde{m}_0(u,X)du+c^\star_{10}(X)\int_0^1\tilde{m}_1(u,X)du\right]+C^\star\nonumber\\
=&\Sigma^\star(\tilde{m}_0,\tilde{m}_1),
\end{align}
where $\Sigma^\star(\tilde{m}_0,\tilde{m}_1)$ is defined in Lemma \ref{lemma_Gamma_simplified}, which depends only on $\int_0^1 \tilde{m}_0(u,x)du$ and $\int_0^1\tilde{m}_1(u,x)du$ for all $x$ in its support.
For our target parameter,
because $m\in\mathcal{M}^r_\mathcal{S}$ satisfies \eqref{cond:Gamma}, under condition (i) that $\mathcal{S}=\mathbf{L}^2(\mathcal{Z},\mathbb{R})$ and condition (ii) that $\omega_{dd'}^\star(V,Z)=\omega_{dd'}^\star(Z)$, Lemma \ref{lemma_Gamma_simplified} implies that
\begin{align}\label{equv_Gamma*}
\Gamma^\star(m)&=\Sigma^\star(m_0,m_1),
\end{align}
which depends only on $\int_{{\mathcal{V}}}m_0(v,x)dv$ and $\int_{{\mathcal{V}}}m_1(v,x)dv$ for all $x$ in its support.

Second, we show that 
\begin{align}\label{equv_ours_MST_geq}\sup_{m\in \mathcal{M}^r_\mathcal{S}}\Gamma^\star(m)\geq \sup_{(\tilde{m}_0,\tilde{m}_1)\in \tilde{\mathcal{M}}_{01,\tilde{\mathcal{S}}}}\tilde{\Gamma}^\star(\tilde{m}_0,\tilde{m}_1).\end{align}
For any $(\tilde{m}_0,\tilde{m}_1)\in \tilde{\mathcal{M}}_{01,\tilde{\mathcal{S}}}$, we can construct  $m:=(m_0,m_1,m_D,m_{0D},m_{1D})'\in\mathcal{M}$, where
\begin{equation}\label{equv_def_mtilde}\begin{aligned}
m_0(v,x)=&\tilde{m}_0(u,x),\\
m_1(v,x)=&\tilde{m}_1(u,x),\\
    m_D(v,z)=&1[p(z)-u\geq 0],\\
    m_{0D}(v,z)=&m_0(v,x)(1-m_D(v,z)),\\
    m_{1D}(v,z)=&m_1(v,x)m_D(v,z).
\end{aligned}\end{equation}
In the construction of $m$ above, we can understand $u$ as the first element of $v$, and all functions in $m$ are degenerate in the  remaining elements of $v$. By construction, $m$ satisfies \eqref{cond:bilinear}. Under the threshold-crossing structure, we know that $D=1[p(Z)\geq U]$ in the true DGP. Thus, we have $\mathbb{E}[D\mid U,Z]=1[U\leq p(Z)]=m_D(V,Z)$, where the second equality is by the definition of $m_D$ in \eqref{HV_m1restriction}. Thus, by the law of iterated expectation, we obtain
$$
\mathbb{E}[D\mid Z]
=\mathbb{E}\{\mathbb{E}[D\mid U,Z]\mid Z\}
=\mathbb{E}[m_D(V,Z)\mid Z].
$$
In addition, by the definition of $m_{1}$, $m_{D}$, and $m_{1D}$ given in \eqref{equv_def_mtilde}, we have
\begin{align*}
\mathbb{E}[m_{1D}(V,Z)\mid Z]=&\mathbb{E}[m_1(V,X)m_D(V,Z)\mid Z]\\
=&\mathbb{E}[\tilde{m}_1(U,X)1[U\leq p(Z)]\mid Z]\\
=&\int_0^1\tilde{m}_1(u,X)1[u\leq p(Z)]du\\
=& \mathbb{E}[YD\mid Z],
\end{align*}
where the last equality is because $(\tilde{m}_0,\tilde{m}_1)\in \tilde{\mathcal{M}}_{01,\tilde{\mathcal{S}}}$ satisfies \eqref{HV_m1restriction}. Similarly, we can show that $\mathbb{E}[Y(1-D)\mid Z]=\mathbb{E}[m_{0D}(V,Z)\mid Z]$. Therefore, $m$ also satisfies the constraints in \eqref{cond:Gamma}, so that $m\in \mathcal{M}_\mathcal{S}\subseteq\mathcal{M}^r_\mathcal{S}$. Moreover, for $m$ constructed in \eqref{equv_def_mtilde}, we know that $\tilde{\Gamma}^\star(\tilde{m}_0,\tilde{m}_1)=\Sigma^\star(\tilde{m}_0,\tilde{m}_1)=\Gamma^\star(m)$, where the first equality holds by \eqref{equv_Gamma*tilde}, and the second equality holds by \eqref{equv_Gamma*} and the fact that $\int_{\mathcal{V}}m_d(v,x)dv=\int_{0}^1\tilde{m}_d(u,x)du$ for $d=0,1$. Since this result holds for any $(\tilde{m}_0,\tilde{m}_1)\in \tilde{\mathcal{M}}_{01,\tilde{\mathcal{S}}}$, we know that
the inequality in \eqref{equv_ours_MST_geq} holds. 

Third, we show that 
\begin{align}\label{equv_ours_MST_leq}
\sup_{m\in \mathcal{M}^r_\mathcal{S}}\Gamma^\star(m)\leq \sup_{(\tilde{m}_0,\tilde{m}_1)\in \tilde{\mathcal{M}}_{01,\tilde{\mathcal{S}}}}\tilde{\Gamma}^\star(\tilde{m}_0,\tilde{m}_1).
\end{align}
Recall that for any $m\in\mathcal{M}^r_\mathcal{S}$, by Lemma \ref{lemma_Gamma_simplified} we know that $\Gamma^\star(m)=\Sigma^\star(m_0,m_1)$ depends only on $\int_{{\mathcal{V}}}m_1(v,x)dv$ and $\int_{{\mathcal{V}}}m_0(v,x)dv$.  Under assumptions in Lemma \ref{lemma:mccormick} and conditions (i) and (iv), Lemma \ref{lemma_outerset_m0m1_cvr} implies the following outer sets for  $\int_{{\mathcal{V}}}m_1(v,x)dv$ and $\int_{{\mathcal{V}}}m_0(v,x)dv$, where $m_0$ and $m_1$ are the first two elements in $m\in\mathcal{M}^r_\mathcal{S}$:
\begin{align}\label{manski_m1bound_coro}
&M_1^L(x)\leq \mathbb{E}[m_{1}(V,x)]\leq M_1^U(x),\text{ where }\\
M_1^L(x)=\sup_{z_0\in \mathcal{Z}_0}&\big\{\mathbb{E}[YD\mid X=x,Z_0=z_0]+y^L\mathbb{E}[1-D\mid X=x,Z_0=z_0]\big\},\nonumber\\
M_1^U(x)= \inf_{z_0\in \mathcal{Z}_0}&\big\{\mathbb{E}[YD\mid X=x,Z_0=z_0]+y^U\mathbb{E}[1-D\mid X=x,Z_0=z_0]\big\},\nonumber
\end{align}
and
\begin{align}\label{manski_m0bound_coro}
&M_0^L(x)\leq \mathbb{E}[m_{0}(V,x)]\leq M_0^U(x),\text{ where }\\
M_0^L(x)=\sup_{z_0\in \mathcal{Z}_0}&\big\{\mathbb{E}[Y(1-D)\mid X=x,Z_0=z_0]+y^L\mathbb{E}[D\mid X=x,Z_0=z_0]\big\},\nonumber\\
M_0^U(x)= \inf_{z_0\in \mathcal{Z}_0}&\big\{\mathbb{E}[Y(1-D)\mid X=x,Z_0=z_0]+y^U\mathbb{E}[D\mid X=x,Z_0=z_0]\big\}.\nonumber
\end{align}
Then, it is sufficient to show that for any $s_0(x)\in [M_0^L(x),M_0^U(x)]$ and $s_1(x)\in [M_1^L(x),M_1^U(x)]$, we can find a $(\tilde{m}_0,\tilde{m}_1)\in \tilde{\mathcal{M}}_{01,\tilde{\mathcal{S}}}$ such that $\int_0^1\tilde{m}_d(u,x)=s_d(x)$ holds for both $d=0,1$. If so, we know that there exists a $(\tilde{m}^\star_0,\tilde{m}^\star_1)\in \tilde{\mathcal{M}}_{01,\tilde{\mathcal{S}}}$ such that
$\sup_{m_0,m_1\text{ satisfy \eqref{manski_m1bound_coro},\eqref{manski_m0bound_coro}}}\Sigma^\star(m_0,m_1)=\Sigma^\star(\tilde{m}^\star_0,\tilde{m}^\star_1).$
Then, we have
\begin{align}\label{opposite_ineq_HV_Manski}
\sup_{m\in \mathcal{M}^r_\mathcal{S}}\Gamma^\star(m)\leq \sup_{m_0,m_1\text{ satisfy \eqref{manski_m1bound_coro},\eqref{manski_m0bound_coro}}}\Sigma^\star(m_0,m_1)=\Sigma^\star(\tilde{m}^\star_0,\tilde{m}^\star_1)\leq\sup_{(\tilde{m}_0,\tilde{m}_1)\in \tilde{\mathcal{M}}_{01,\tilde{\mathcal{S}}}}\tilde{\Gamma}^\star(\tilde{m}_0,\tilde{m}_1).\end{align}
Recall that we define $p^L_x=\inf_{z_0\in\mathcal{Z}_0}p(x,z_0)\text{ and }p^U_x=\sup_{z_0\in\mathcal{Z}_0}p(x,z_0)$. 
Under the threshold-crossing structure, Lemma \ref{lemma_mtilde_sharp} implies that for any $s_0(x)\in [M_0^L(x),M_0^U(x)]$ and $s_1(x)\in [M_1^L(x),M_1^U(x)]$, there exist some $y^\star_0(x),y^\star_1(x)\in[y^L,y^U]$ such that $\mathbb{E}[\tilde{m}_{0}(U,x)]
=s_0(x)$ and $
\mathbb{E}[  \tilde{m}_{1}(U,x)]  =s_1(x)$ hold, where
\begin{equation}\label{construct_m_tilde_single_dim_coro}
\begin{aligned}
\tilde{m}_{0}(U,x)
&=
\begin{cases}
\mathbb{E}[Y_0\mid U,X=x]&\mbox{ if }U> p^L_x\\
y_0^*(x)&\mbox{ otherwise},
\end{cases}
\\
\tilde{m}_{1}(U,x)
&=
\begin{cases}
\mathbb{E}[Y_1\mid U,X=x]&\mbox{ if }U\leq p^U_x\\
y_1^*(x)&\mbox{ otherwise},
\end{cases}
\end{aligned}\end{equation}
and $U$ is a scalar unobserved heterogeneity in $D=1[p(Z)\geq U]$ that follows a uniform distribution over $[0,1]$ conditional on $Z$. For any $p(Z)\geq p^L_x$, we have
\begin{align}\label{HV_rest_Y1-D}
\mathbb{E}[Y(1-D)\mid Z]=&\mathbb{E}\{\mathbb{E}[Y(1-D)\mid Z,U]\mid Z\}\nonumber\\
=&\mathbb{E}\{\mathbb{E}[Y_0\mid Z,U,D=0]\mathbb{E}[1-D\mid Z,U]\mid Z\}\nonumber\\
=&\mathbb{E}\{\mathbb{E}[Y_0\mid X,U]\mathbb{E}[1-D\mid Z,U]\mid Z\}\nonumber\\
=&\mathbb{E}\{\mathbb{E}[Y_0\mid X,U]1[U>p(Z)]\mid Z\}\nonumber\\
=&\int_0^1\mathbb{E}[Y_0\mid U=u,X]1[u>p(Z)]du\nonumber\\
=&\int_0^1\tilde{m}_0(u,X)1[u>p(Z)]du,
\end{align}
where the fourth equality is due to the threshold-crossing structure, and the last equality is by the definition of $\tilde{m}_{0}(u,x)$ in \eqref{construct_m_tilde_single_dim_coro} and the fact that $p(Z)\geq p^L_x$.
Because $p(Z)\geq p^L_x$ holds for any $Z\in\mathcal{Z}$ by the definition of $p^L_x$, we know that the equality above is satisfied for all $Z\in\mathcal{Z}$.
Similarly, we can show that for all $Z\in\mathcal{Z}$, 
\begin{align}\label{HV_rest_YD}
    \int_0^1\tilde{m}_1(u,X)1[u\leq p(Z)]du=&\mathbb{E}[YD\mid Z].
\end{align}
Since under condition (i), the constraints in \eqref{MST_restrictions} are equivalent to \eqref{HV_rest_Y1-D} and \eqref{HV_rest_YD}, we know that $(\tilde{m}_0,\tilde{m}_1)$ defined in \eqref{construct_m_tilde_single_dim_coro} satisfies $(\tilde{m}_0,\tilde{m}_1)\in \tilde{\mathcal{M}}_{01,\tilde{\mathcal{S}}}$. 
Then, the inequality in \eqref{opposite_ineq_HV_Manski} also holds, which fulfills the proof.

(ii) {Next, we consider the case under the same shape restrictions on the integrated marginal treatment response functions, $$(\int_{0}^1\tilde{m}_0(u,x)du,\int_{0}^1\tilde{m}_1(u,x)du)\text{ and }(\int_{\mathcal{V}}m_0(v,x)dv,\int_{\mathcal{V}}m_1(v,x)dv).$$
We prove that the two inequalities in \eqref{equv_ours_MST_geq} and \eqref{equv_ours_MST_leq} still hold under these shape restrictions.
For any $(\tilde{m}_0,\tilde{m}_1)$ such that $\int_{0}^1\tilde{m}_0(u,x)du\text{ and }\int_{0}^1\tilde{m}_1(u,x)du$ satisfy certain shape restrictions, we can see that $\int_{\mathcal{V}}m_0(v,x)dv\text{ and }\int_{\mathcal{V}}m_1(v,x)dv$, with $(m_0,m_1)$ constructed in \eqref{equv_def_mtilde}, also satisfy the same shape restrictions. Therefore, the same arguments used to prove the inequality in \eqref{equv_ours_MST_geq} can also be applied to show that it holds under the shape restrictions. 
}

 {Next, we show that the inequality in \eqref{equv_ours_MST_leq} holds. Denote $\mathcal{M}_{01}$ as the projection of $\mathcal{M}$ onto the space of $(m_0,m_1)$. Under the shape restrictions on $(\int_{\mathcal{V}}m_0(v,x)dv,\int_{\mathcal{V}}m_1(v,x)dv)$, the outer set of $\int_{\mathcal{V}}m_d(v,x)dv$ becomes 
\begin{align}\label{new_outer}
[M_d^L(x),M_d^U(x)]\cap\Big\{&\int_{\mathcal{V}}m_d(v,x)dv:~ \text{there exists }(m_0,m_1)\in\mathcal{M}_{01}\nonumber\\&\text{ that satisfies the restrictions on $(\int_{\mathcal{V}}m_0(v,x)dv,\int_{\mathcal{V}}m_1(v,x)dv)$}\Big\}.
\end{align}
In Lemma \ref{lemma_mtilde_sharp}, we have shown that for any $s_d(x)\in [M_d^L(x),M_d^U(x)]$, we can construct $\tilde{m}_d$ such that $\mathbb{E}[\tilde{m}_{d}(U,x)]
=s_d(x)$. Since the set in \eqref{new_outer} is a subset of $[M_d^L(x),M_d^U(x)]$, we know that the same result still holds for the subset. That is, for any $s_d(x)$ in the set defined in \eqref{new_outer}, we can construct $\tilde{m}_d$ such that $\mathbb{E}[\tilde{m}_{d}(U,x)]
=s_d(x)$. In addition, using the same proof as in the third step of (i), we can conclude that such $(\tilde{m}_0,\tilde{m}_1)$ belongs to $\tilde{\mathcal{M}}_{01,\tilde{\mathcal{S}}}$. Moreover, since any $(s_0(x),s_1(x))$ in the set of \eqref{new_outer} satisfies the shape restrictions on $(\int_{\mathcal{V}}m_0(v,x)dv,\int_{\mathcal{V}}m_1(v,x)dv)$, we know that $(\mathbb{E}[\tilde{m}_{0}(U,x)],\mathbb{E}[\tilde{m}_{1}(U,x)])
=(s_0(x),s_1(x))$ also satisfies the same shape restrictions. Hence, the inequality in \eqref{equv_ours_MST_leq} holds.}
\end{proof}

\begin{proof}[{Proof of Proposition \ref{prop_finite_dim_approx}}]
Because $\mathbf{B}'\eta_2\in\mathcal{M}$, it is easy to see that $\underline{\beta}_a^\star\geq\underline{\beta}^\star$ and $\overline{\beta}_a^\star\leq\overline{\beta}^\star$. In what follows, we show $\underline{\beta}_a^\star\leq\underline{\beta}^\star$ and $\overline{\beta}_a^\star\geq\overline{\beta}^\star$.  Recall that $X\in\{x_1,...,x_{K_X}\}$. Suppose for any given $X=x_j$, we have $Z=(X,Z_0)=(x_j,Z_0)\in\{z_{j1},...,z_{jK_j}\}$. Denote
\begin{equation}
\begin{aligned}
\mu_{kj}=&\mathbb{E}[m_d(V,X)\mid V\in \mathcal{V}_k,X=x_j],\\
\varphi_{kjl}=&\mathbb{E}[m_D(V,Z)\mid V\in \mathcal{V}_k,X=x_j,Z=z_{jl}],\\
\xi_{kjl}=&\mathbb{E}[m_{dD}(V,Z)\mid V\in \mathcal{V}_k,X=x_j,Z=z_{jl}].
\end{aligned}
\end{equation}
With notation abuse, let $b_{D,kjl}(v,z)=1[v\in\mathcal{V}_k,x=x_j,z=z_{jl}]$.
Then, we can rewrite the terms in \eqref{constant_spline_appx} as
\begin{equation}
\begin{aligned}
\Lambda m_d(v,x)=&\sum_{k=1}^{K_V}\sum_{j=1}^{K_X}\mu_{kj}b_{d,kj}(v,x),~~\Lambda m_D(v,z)=\sum_{k=1}^{K_V}\sum_{j=1}^{K_X}\sum_{l=1}^{K_j}\varphi_{kjl}b_{D,kjl}(v,z),\\
\Lambda m_{dD}(v,x)=&\sum_{k=1}^{K_V}\sum_{j=1}^{K_X}\sum_{l=1}^{K_j}\xi_{kjl}b_{D,kjl}(v,z).
\end{aligned}
\end{equation}
First, we show that for some function $w(v,z)$ that is a constant over $\mathcal{V}_k$ for a given $z\in\mathcal{Z}$, we have $\mathbb{E}[m_d(v,X)w(v,Z)]=\mathbb{E}[\Lambda m_d(v,X)w(v,Z)]$. Because $b_{d,kj}(v,x)=1[v\in\mathcal{V}_k,x=x_j]$, we have
\begin{align}\label{finite_approx1}
\mathbb{E}\Big[\int_{{\mathcal{V}}}\Lambda m_d(v,X)w(v,Z) dv\Big]=&\mathbb{E}\left\{\int_{{\mathcal{V}}}\Big[\sum_{k=1}^{K_V}\sum_{j=1}^{K_X}\mu_{kj}b_{d,kj}(v,X)\Big]w(v,Z) dv\right\}\nonumber\\
=&\mathbb{E}\left\{\sum_{j=1}^{K_X}1[X=x_j]\int_{{\mathcal{V}}}\Big[\sum_{k=1}^{K_V}\sum_{j=1}^{K_X}\mu_{kj}b_{d,kj}(v,x_j)\Big]w(v,Z) dv\right\}\nonumber\\
=&\mathbb{E}\left\{\sum_{j=1}^{K_X}1[X=x_j]\sum_{k=1}^{K_V}\int_{\mathcal{V}_k}\mu_{kj}b_{d,kj}(v,x_j)w(v,Z) dv\right\}\nonumber\\
=&\mathbb{E}\left\{\sum_{j=1}^{K_X}\sum_{l=1}^{K_j}1[X=x_j,Z=z_{jl}]\sum_{k=1}^{K_V}\int_{\mathcal{V}_k}\mu_{kj}w(v,z_{jl})dv\right\}.
\end{align}
By the definition of $\mu_{kj}$ and the fact that $w(v,z)$ is a constant over $\mathcal{V}_k$, we know that
\begin{align}\label{finite_approx2}
\sum_{k=1}^{K_V}\int_{\mathcal{V}_k}\mu_{kj}w(v,z_{jl})dv=&\sum_{k=1}^{K_V}\int_{\mathcal{V}_k}\mathbb{E}[m_d(V,X)w(V,z_{jl})\mid V\in \mathcal{V}_k,X=x_j]dv\nonumber\\
=&\sum_{k=1}^{K_V}\mathbb{E}[m_d(V,X)w(V,z_{jl})\mid V\in \mathcal{V}_k,X=x_j]\mathbb{P}(V\in\mathcal{V}_k)\nonumber\\
=&\mathbb{E}[m_d(V,x_j)w(V,z_{jl})\mid X=x_j],
\end{align}
where the last line is by the law of iterated expectation and the assumption that $V$ given $X$ is uniformly distributed. Substituting \eqref{finite_approx2} into \eqref{finite_approx1} leads to
\begin{align}
\mathbb{E}\Big[\int_{{\mathcal{V}}}\Lambda m_d(v,X)w(v,Z) dv\Big]=&\mathbb{E}\left\{\sum_{j=1}^{K_X}\sum_{l=1}^{K_l}1[X=x_j,Z=z_{jl}]\mathbb{E}[m_d(V,x_j)w(V,z_{jl})\mid X=x_j]\right\}\nonumber\\
=&\mathbb{E}\Big[\int_{{\mathcal{V}}}m_d(v,X)w(v,Z)dv\Big].
\end{align}
Next, we show that for some function $w(v,z)$ that is a constant over $\mathcal{V}_k$ for a given $z\in\mathcal{Z}$, we have $\mathbb{E}[m_{dD}(v,Z)w(v,Z)]=\mathbb{E}[\Lambda m_{dD}(v,Z)w(v,Z)]$.
Note that 
\begin{align}
\mathbb{E}\Big[\int_{{\mathcal{V}}}\Lambda m_{dD}(v,Z)w(v,Z) dv\Big]=&\mathbb{E}\left\{\int_{{\mathcal{V}}}\Big[\sum_{k=1}^{K_V}\sum_{j=1}^{K_X}\sum_{l=1}^{K_j}\xi_{kjl}b_{D,kjl}(v,Z)\Big]w(v,Z) dv\right\}\nonumber\\
=&\mathbb{E}\left\{\sum_{j=1}^{K_X}\sum_{l=1}^{K_j}1[X=x_j,Z=z_{jl}]\int_{{\mathcal{V}}}\Big[\sum_{k=1}^{K_V}\xi_{kjl}b_{D,kjl}(v,z_{jl})\Big]w(v,z_{jl}) dv\right\}\nonumber\\
=&\mathbb{E}\left\{\sum_{j=1}^{K_X}\sum_{l=1}^{K_j}1[X=x_j,Z=z_{jl}]\sum_{k=1}^{K_V}\int_{\mathcal{V}_k}\xi_{kjl}w(v,z_{jl})dv\right\}\label{finite_approx_g1};\mbox{ and }
\\
\sum_{k=1}^{K_V}\int_{\mathcal{V}_k}\xi_{kjl}w(v,z_{jl})dv=&\sum_{k=1}^{K_V}\int_{\mathcal{V}_k}\mathbb{E}[m_{dD}(V,Z)\mid V\in \mathcal{V}_k,X=x_j,Z=z_{jl}]w(v,z_{jl})dv\nonumber\\
=&\sum_{k=1}^{K_V}\mathbb{E}[m_{dD}(V,Z)w(V,Z)\mid V\in \mathcal{V}_k,X=x_j,Z=z_{jl}]\mathbb{P}(V\in\mathcal{V}_k)\nonumber\\
=&\mathbb{E}[m_{dD}(V,z_{jl})w(V,z_{jl})\mid X=x_j].\label{finite_approx_g2}
\end{align}
Plugging \eqref{finite_approx_g2} into \eqref{finite_approx_g1} gives the desirable result.
Similar arguments can be applied to show that the same result holds for $m_D$. Let us set $w(v,z)=\omega^\star_{dd'}(v,z)$. Then, we have $\Gamma^\star(\Lambda m)=\Gamma^\star(m)$.

Since all constraints in \eqref{cond:Gamma} and \eqref{eq:mccormick_bounds_m}-\eqref{eq:mccormick_relax4_1} are linear in $m$, then for any $m\in\mathcal{M}$ that satisfies \eqref{cond:Gamma} and  \eqref{eq:mccormick_bounds_m}-\eqref{eq:mccormick_relax4_1}, it must be the case that $\Lambda m$ also satisfies \eqref{cond:Gamma} and  \eqref{eq:mccormick_bounds_m}-\eqref{eq:mccormick_relax4_1}. Then, we have
\begin{align*}
\overline{\beta}^\star=&\sup_{m\in\mathcal{M}}\{\Gamma^\star(m):~\text{\eqref{cond:Gamma} for any $s\in\mathcal{S}$ and \eqref{eq:mccormick_bounds_m}-\eqref{eq:mccormick_relax4_1} satisfied by }m\}\\
\leq&\sup_{m\in\mathcal{M}}\{\Gamma^\star(\Lambda m):~\text{\eqref{cond:Gamma}  for any $s\in\mathcal{S}$ and \eqref{eq:mccormick_bounds_m}-\eqref{eq:mccormick_relax4_1} satisfied by }\Lambda m\}\\
\leq&\sup_{\tilde m\in \mathcal{M}_a}\{\Gamma^\star(\tilde m):~\text{\eqref{cond:Gamma}  for any $s\in\mathcal{S}$ and  \eqref{eq:mccormick_bounds_m}-\eqref{eq:mccormick_relax4_1} satisfied by }\tilde m\}\\
=&\overline{\beta}_a^\star,
\end{align*}
where the second line is because of $\Gamma^\star(\Lambda m)=\Gamma^\star(m)$, the third line is due that $\Lambda m\in\mathcal{M}_a:=\{m\in\mathcal{M}:~m=\mathbf{B}'\eta_2\mbox{ for constant splines }\mathbf{B}'\eta_2\}$ for any $m\in\mathcal{M}$, and last line follows from the definition of $\overline{\beta}_a^\star$. Similarly, we can show the same result for the infimum.
\end{proof}

\end{document}